\documentclass[pldi-cameraready,nocopyrightspace]{sigplanconf-pldi16}

\usepackage[scaled]{helvet} 
\usepackage{url}                  
\usepackage{listings}          
\usepackage{enumitem}      
\usepackage[colorlinks=true,allcolors=blue,breaklinks,draft=false]{hyperref}   

\usepackage{inconsolata}
\usepackage{calrsfs}

\usepackage[T1]{fontenc}

\DeclareMathAlphabet{\origcal}{OMS}{zplm}{m}{n}

\usepackage{letltxmacro}
\usepackage{mathtools}
\usepackage{mdframed}
\usepackage{amsfonts}
\usepackage{amsthm}
\usepackage{amssymb}
\usepackage{amsmath}
\usepackage{stmaryrd}
\usepackage{comment}
\usepackage{xcolor}
\usepackage{color}
\usepackage{xspace}
\usepackage{mathpartir}
\usepackage{array}
\usepackage{multirow}
\usepackage{subcaption}
\usepackage{enumitem}
\usepackage[font=small,skip=0pt]{caption}
\usepackage{url}
\def\withcolor{}

\ifdefined\withcolor
	\definecolor{haskellblue}{rgb}{0.0, 0.0, 1.0}
	\definecolor{haskellblue}{rgb}{1.0, 0.0, 0.0}
	\definecolor{gray_ulisses}{gray}{0.55}
	\definecolor{castanho_ulisses}{rgb}{0.71,0.33,0.14}
	\definecolor{preto_ulisses}{rgb}{0.41,0.20,0.04}
	\definecolor{green_ulisses}{rgb}{0.0,0.4,0.0}
        \definecolor{mygreen}{rgb}{0,0.6,0}
        \definecolor{mygray}{rgb}{0.5,0.5,0.5}
        \definecolor{mymauve}{rgb}{0.58,0,0.82}
        \definecolor{darkgray}{gray}{0.30}
        \definecolor{darkblue}{rgb}{0,0,0.75}
        \definecolor{darkgreen}{rgb}{0,0.4,0}
\else
	\definecolor{haskellblue}{gray}{0.1}
	\definecolor{haskellred}{gray}{0.1}
	\definecolor{gray_ulisses}{gray}{0.1}
	\definecolor{castanho_ulisses}{gray}{0.1}
	\definecolor{preto_ulisses}{gray}{0.1}
	\definecolor{green_ulisses}{gray}{0.1}
        \definecolor{mygreen}{gray}{0.6}
        \definecolor{mygray}{gray}{0.7}
        \definecolor{mymauve}{gray}{0.8}
        \definecolor{darkgray}{gray}{0.9}
        \definecolor{darkblue}{gray}{0.8}
        \definecolor{darkgreen}{gray}{0.8}
\fi

\def\codesize{\footnotesize}

\lstdefinelanguage{RefScript}{
  basicstyle=\ttfamily\codesize,
  sensitive=true,
  keywords={typeof, instanceof, new, catch, function, return,
            null, catch, switch, immutable, var, if, in, for, while, do,
            else, case, break, forall, type, enum,
            constructor, super, public, private, class,
            interface, extends, impl, implements, module,
            declare, this,let, then, undefined, fn,
            },
  literate={<=}{{$\leq$}}1
           {>>=}{{>>=}}3
           {>=}{{$\geq$}}1
           {/\\}{{$\wedge$}}1
           {\\phi}{{$\phi$}}1
           {!=}{{$\neq$}}1
           {forall}{{$\forall$}}1
           {->}{{$\rightarrow$}}2
           {<=>}{{$\Leftrightarrow$}}3
           {=>}{{$\Rightarrow$}}2
           {||-}{{$\vdash$}}1
           {|->}{{$\mapsto$}}2
           {<:}{{$\preceq$}}1
           {<IM>}{{\color{darkgreen}{<IM>}\color{black}}}4
           {<UQ>}{{\color{darkgreen}{<UQ>}\color{black}}}4
           {<MU>}{{\color{darkgreen}{<MU>}\color{black}}}4,
  keywordstyle=\bfseries,
  morecomment=[s]{/*}{*/},
  morecomment=[l]//,
  morestring=[b]",
  morestring=[b]',
  commentstyle=\color{mygreen},
  stringstyle=\color{darkgreen},
  showstringspaces=false,
  numberstyle=\scriptsize,
  numberblanklines=true,
  showspaces=false,
  breaklines=true,
  showtabs=false,
  emph={[1] getDensity,setDensity,reset,reduce, 
            distinct, setF,
            $reduce, $reduce1, $reduce2, step, minIndex,minIndexFO,
            minIndexHO,createType,
            minIndex, neg, getPropertiesOfType,
            idxOf, concat, pop, push, set, get,
						length, assert,
						resolveObjTypeMem,
            head, head0,
            slice
       },
  emphstyle={[1]\color{darkblue}},
  emph={[2] boolean,bool,tt,ff,any,void,
						pos,okW,okH,grid,number,string,nat,index,
						NEArray, ArrayN, natN,natLE,idx,Type,Symbol,
						TypeFlags,ObjType,ObjectType,InterfaceType,Array
            ReadOnly
  },
  emphstyle={[2]\color{darkgreen}},
  emph={[3] @Mutable, @Immutable, @ReadOnly
  },
  emphstyle={[3]\color{purple}},
 escapeinside={`}{`}
}

\lstnewenvironment{codewithnumbers}
{
  \lstset{
    language=RefScript,
    numbers=left,
    xleftmargin=5.0ex
  }
}
{}

\lstnewenvironment{code}{\lstset{language=RefScript}}{}

\lstMakeShortInline[language=RefScript, basicstyle=\normalsize\ttfamily]@

\usepackage{commands}
\usepackage{tabularx}
\newtheorem{assumption}{Assumption}
\newtheorem{theorem}{Theorem}
\newtheorem*{theorem*}{Theorem}
\newtheorem{definition}{Definition}
\newtheorem{corollary}[theorem]{Corollary}
\newtheorem{lemma}{Lemma}
\usepackage{chngcntr}
\usepackage{hyphenat}
\usepackage{afterpage}

\begin{document}

%
%

\title{\mytitle}

\authorinfo{Panagiotis Vekris} {University of California, San Diego}{\email{pvekris}{cs.ucsd.edu}}
\authorinfo{Benjamin Cosman}  {University of California, San Diego}{\email{blcosman}{cs.ucsd.edu}}
\authorinfo{Ranjit Jhala} {University of California, San Diego}{\email{jhala}{cs.ucsd.edu}}

\maketitle

\begin{abstract}
  We present \lang (\rsc), a lightweight refinement type system for
\tsc, that enables static verification of higher-order,
imperative programs.
We develop a formal core of \rsc that delineates the interaction
between refinement types and mutability.
Next, we extend the core to account for the imperative and dynamic
features of \tsc.
Finally, we evaluate \rsc on a set of real world
benchmarks, including parts of the Octane benchmarks,
D3, Transducers, and the \tsc compiler.
\end{abstract}

\section{Introduction}\label{sec:intro}

Modern \emph{scripting} languages -- like JavaScript,
Python, and Ruby --  have  popularized the use of
higher-order constructs that were once solely in the
\emph{functional} realm.
%
%
This trend towards abstraction and reuse
poses two related problems for static analysis:
\emph{modularity} and \emph{extensibility}.
First, how should analysis precisely track the
flow of values across {higher-order functions}
and {containers} or \emph{modularly}
account for {external} code like closures
or library calls?
Second, how can analyses be easily \emph{extended}
to new, domain specific properties, ideally by
developers, while they are designing and implementing
the code? (As opposed to by experts who can
at best develop custom analyses run
\emph{ex post facto} and are of little use
\emph{during} development.)

Refinement types hold the promise of a
precise, modular and extensible analysis for programs with
higher-order functions and containers.
Here, \emph{basic} types are decorated with \emph{refinement}
predicates that constrain the values inhabiting the type
\cite{Rushby98,XiPfenning99}.
%
%
%
The extensibility and modularity offered by refinement
types have enabled their use in a variety of applications
in \emph{typed}, \emph{functional} languages, like
ML~\cite{XiPfenning99,LiquidPLDI08},
Haskell~\cite{Vazou2014},
and $F^\sharp$~\cite{Swamy2011}.
Unfortunately, attempts to apply refinement typing to
scripts have proven to be impractical due to the interaction
of the machinery that accounts for imperative updates
and higher-order functions~\cite{Chugh2012} (\S\ref{sec:related}).

In this paper, we introduce \lang\ (\rsc): a novel, \emph{lightweight}
refinement type system for \tsc, a typed superset of \jsc.
Our design of \rsc\ addresses three intertwined problems by
carefully integrating and extending existing ideas from the
literature.
First, \rsc\ accounts for \emph{mutation} by using ideas from
IGJ~\cite{Zibin2007} to track which fields may be mutated, and
to allow refinements to depend on immutable fields, and by using
SSA-form to recover path and flow-sensitivity that is essential
for analyzing real world applications.
Second, \rsc\ accounts for \emph{dynamic typing} by using a
recently proposed technique called two-phase typing
\cite{rsc-ecoop15}, where dynamic behaviors are specified via
union and intersection types, and verified by reduction to
refinement typing.
Third, the above are carefully designed to permit refinement
\emph{inference} via the Liquid Types~\cite{LiquidPLDI08}
framework to render refinement typing practical
on real world programs.
Concretely, we make the following contributions:
\begin{itemize}
\item We develop a core calculus that formalizes the
      interaction of mutability and refinements via declarative
      refinement type checking that we prove
      sound~(\S\ref{sec:language}).

\item We extend the core language to \tsc\ by describing how
      we account for its various \emph{dynamic} and
      \emph{imperative} features; in particular we
      show how \rsc\ accounts for type reflection
      via intersection types, encodes interface
      hierarchies via refinements, and crucially
      permits locally flow-sensitive reasoning
      via SSA translation~(\S\ref{sec:typescript}).

\item We implement \toolname, a refinement type-checker
      for \tsc, and evaluate it on  a suite of real world
      programs from the @Octane@ benchmarks, Transducers,
      @D3@ and the \tsc\ compiler.
      We show that \rsc's refinement typing is \emph{modular}
      enough to analyze higher-order functions, collections and
      external code, and \emph{extensible} enough to verify
      a variety of properties from classic array-bounds checking
      to program specific invariants needed to ensure safe
      reflection: critical invariants that are well beyond
      the scope of existing techniques for imperative scripting
      languages~(\S\ref{sec:evaluation}).
\end{itemize}

\begin{figure}[t!]
\begin{code}
    function reduce(a, f, x) {
      var res = x, i;
      for (var i = 0; i < a.length; i++)
        res = f(res, a[i], i);
      return res;
    }
    function minIndex(a) {
      if (a.length <= 0) return -1;
      function step(min, cur, i) {
        return cur < a[min] ? i : min;
      }
      return reduce(a, step, 0);
    }
\end{code}
\caption{Computing the Min-Valued Index with \texttt{reduce}}
\label{fig:reduce}
\label{fig:minindex}
\end{figure}

\newcommand{\tyreduce}{T_{\primty{reduce}}}

\section{Overview}\label{sec:overview}

We begin with a high-level overview
of refinement types in \rsc, 
their applications (\S\ref{sec:overview-applications}),
and how \rsc\ handles imperative, higher-order
constructs~(\S\ref{sec:overview-analysis}).



\mypara{Types and Refinements} A basic refinement type
is a basic type, \eg @number@, refined with a logical
formula from an SMT decidable logic~\cite{Nelson81}.
For example, the types:
\begin{code}
  type nat     = {v:number | 0 <= v}
  type pos     = {v:number | 0 < v}
  type natN<n> = {v:nat    | v = n}
  type idx<a>  = {v:nat    | v < len(a)}
\end{code}
describe (the set of values corresponding to)
\emph{non-negative} numbers,
\emph{positive} numbers,
numbers \emph{equal to} some value @n@, and
\emph{valid indexes} for an array @a@, respectively.
Here, @len@ is an \emph{uninterpreted function}
that describes the size of the array @a@.
%
We write @t@ to abbreviate trivially refined types,
\ie @{v:t | true}@; \eg @number@ abbreviates @{v:number | true}@.

\mypara{Summaries} Function Types $\funtypen{x}{\type}{\type}$,
where arguments are named $x_i$ and have types $\type_i$
and the output is a $\type$, are used to specify the
behavior of functions.
In essence, the \emph{input} types $\type_i$ specify the
function's preconditions, and the \emph{output} type $\type$
describes the postcondition. Each input type and the
output type can \emph{refer to} the arguments $x_i$, yielding
precise function contracts. For example,
%
%
$\funtype{\bind{x}{\tnat}}{\reftp{\tnat}{x < \vv}}$
is a function type that describes functions that \emph{require}
a non-negative input, and \emph{ensure} that the output exceeds
the input.

\mypara{Higher-Order Summaries}
This approach generalizes directly to precise descriptions
for \emph{higher-order} functions.
For example, @reduce@ from Figure~\ref{fig:reduce} can be specified as $\tyreduce$:
\begin{equation}
  \text{\lstinline{<A,B>(a:A[], f:(B, A, idx<a>) => B, x:B) => B}}  \label{fig:head:sig}
\end{equation}
This type is a precise \emph{summary} for the higher-order behavior
of @reduce@: it describes the relationship between the input
array @a@, the step (``callback'') function @f@, and the initial
value of the accumulator, and stipulates that the output satisfies
the same \emph{properties} @B@ as the input @x@.
Furthermore, it critically specifies that the callback @f@ is only
invoked on valid indices for the array @a@ being reduced.

\subsection{Applications} \label{sec:overview-applications}

Next, we show how refinement types let programmers \emph{specify}
and statically \emph{verify} a variety of properties --- array safety,
reflection (value-based overloading), and down-casts --- potential sources
of runtime problems that cannot be prevented via existing techniques.

\subsubsection{Array Bounds}


\mypara{Specification} We specify safety by defining
suitable refinement types for array creation and
access. For example, we view read @a[i]@,
write @a[i] = e@ and length access @a.length@ as
calls @get(a,i)@, @set(a,i,e)@ and @length(a)@ where:
\begin{code}
  get    : (a:T[],i:idx<a>) => T
  set    : (a:T[],i:idx<a>,e:T) => void
  length : (a:T[]) => natN<len(a)>
\end{code}

\mypara{Verification} Refinement typing ensures
that the \emph{actual} parameters supplied at each
\emph{call} to @get@ and @set@ are subtypes of the
\emph{expected} values specified in the signatures,
and thus verifies that all accesses are safe.
As an example, consider the function that returns
the ``head'' element of an array:
\begin{code}
  function head(arr:NEArray<T>){ return arr[0]; }
\end{code}
The input type requires that @arr@ be \emph{non-empty}:
\begin{code}
  type NEArray<T> = {v:T[] | 0 < len(v)}
\end{code}
We convert @arr[0]@ to @get(arr,0)@ which is checked under environment
$\Gamma_\thead$ defined as $\bind{\primty{arr}}{\reftp{\tarray{\primty{T}}}{0 < \tlen{\vv}}}$
yielding the subtyping obligation:
\begin{alignat*}{2}
\Gamma_\thead & \vdash\ \refp{\vv = 0} && \ \lqsubt\ \tidx{\primty{arr}} \\ 
\intertext{which reduces to the logical \emph{verification condition} (VC):}
0 < \tlen{\primty{arr}} & \Rightarrow\  (\vv = 0 && \ \Rightarrow\ 0 \leq
\vv < \tlen{\primty{arr}}) 
\end{alignat*}
The VC is proved \emph{valid} by an SMT solver~\cite{Nelson81},
verifying subtyping, and hence, the array access' safety.

\mypara{Path Sensitivity} is obtained by adding branch conditions
into the typing environment. Consider:
\begin{code}
  function head0(a:number[]): number {
    if (0 < a.length) return head(a);
    return 0;
  }
\end{code}
Recall that @head@ should only be invoked with \emph{non-empty}
arrays. The call to @head@ above occurs under $\Gamma_\theadz$ defined
as:
$\bind{\primty{a}}{\tarray{\tnumber}},\ \eguard{0 < \tlen{\primty{a}}}$
\ie\ which has the binder for the formal @a@, and the guard predicate
established by the branch condition. Thus, the call to @head@ yields the
obligation:
$$\Gamma_\theadz \ \vdash\ \refp{\vv = \primty{a}} \ \lqsubt\ \talias{NEArray}{\tnumber}$$
yielding the valid VC:
$$
0 < \tlen{\primty{a}}  \Rightarrow\  (\vv = \primty{a}  \ \Rightarrow\ 0 < \tlen{\vv})
$$


\mypara{Polymorphic, Higher Order Functions}
Next, let us \emph{assume} that @reduce@ has the type $\tyreduce$
described in (\ref{fig:head:sig}),
and see how to verify the array safety of @minIndex@
(Figure~\ref{fig:minindex}). The challenge here is to
precisely track which values can flow into @min@
(used to index into @a@), which is tricky since
those values are actually produced inside @reduce@.

Types make it easy to track such flows: we need only
determine the \emph{instantiation} of the polymorphic
type variables of @reduce@ at this call site inside @minIndex@.
The type of the @f@ parameter in the instantiated type
corresponds to a signature for the closure @step@
which will let us verify the closure's implementation.
Here, \toolname automatically instantiates (by building complex logical
predicates from simple terms that have been predefined in 
a prelude):
\begin{equation}
\primty{A} \mapsto \tnumber \qquad \qquad \primty{B} \mapsto \tidx{\primty{a}} \label{eq:minindex:inst}
\end{equation}

Let us reassure ourselves that this instantiation is valid, by
checking that @step@ and @0@ satisfy the instantiated type.
If we substitute (\ref{eq:minindex:inst}) into $\tyreduce$ we obtain the
following types for @step@ and @0@, \ie @reduce@'s second and third arguments:
\begin{code}
  step:(idx<a>,number,idx<a>)=>idx<a>   0:idx<a>
\end{code}
The initial value @0@ is indeed a valid @idx<a>@ thanks to the
\linebreak @a.length@ check at the start of the function.
To check @step@, assume that its inputs have the above types:
\begin{code}
  min:idx<a>, curr:number, i:idx<a>
\end{code}
The body is safe as the index @i@ is trivially a subtype of the
required @idx<a>@, and the output is one of @min@ or @i@ and
hence, of type @idx<a>@ as required.

\subsubsection{Overloading}\label{sec:overview:overload}


Dynamic languages extensively use \emph{value-based overloading} to simplify
library interfaces. For example, a library may export:
\begin{code}
function $reduce(a, f, x) {
  if (arguments.length===3) return reduce(a,f,x);
  return reduce(a.slice(1),f,a[0]);
}
\end{code}
%
%
The function @$reduce@ has \emph{two} distinct types
depending on its parameters' \emph{values}, rendering
it impossible to statically type without path-sensitivity.
Such overloading is ubiquitous: in more than 25\% of libraries,
more than 25\% of the functions are value-overloaded~\cite{rsc-ecoop15}.


\mypara{Intersection Types}
Refinements let us statically \emph{verify} value-based
overloading via an approach called 
\emph{Two-Phased Typing}~\cite{rsc-ecoop15}.
First, we specify overloading as an intersection type.
For example, @$reduce@ gets the following signature,
which is just the conjunction of the two overloaded
behaviors:
%
\begin{code}
 <A,B>(a:A[]`$^+$`, f:(A, A, idx<a>)=>A)=>A      // 1
 <A,B>(a:A[]`\phantom{$^+$}`, f:(B, A, idx<a>)=>B, x:B)=>B // 2
\end{code}
The type @A[]@$^+$ in the first conjunct indicates that 
the first argument needs to be a non-empty array, so that 
the call to @slice@ and the access of @a[0]@ both succeed.

\mypara{Dead Code Assertions}
Second, we check each conjunct separately, replacing
ill-typed terms in each context with @assert(false)@.
This requires the refinement type checker to prove
that the corresponding expressions are \emph{dead code},
as @assert@ requires its argument to always be @true@:
\begin{code}
  assert : (b:{v:bool | v = true}) => A
\end{code}
To check @$reduce@, we specialize it per overload context:
\begin{code}
function $reduce1 (a,f) {
  if (arguments.length===3) return assert(false);
  return reduce(a.slice(1), f, a[0]);
}
function $reduce2 (a,f,x) {
  if (arguments.length===3) return reduce(a,f,x);
  return assert(false);
}
\end{code}
In each case, the ``ill-typed'' term (for the corresponding
input context) is replaced with @assert(false)@.
Refinement typing easily verifies the @assert@s,
as they respectively occur under the \emph{inconsistent}
environments:
\begin{align*}
\Gamma_1 & \ \defeq\ \bind{\targs}{\refp{\tlen{\vv} = 2}},\ \tlen{\targs} = 3 \\
\Gamma_2 & \ \defeq\ \bind{\targs}{\refp{\tlen{\vv} = 3}},\ \tlen{\targs} \not = 3
\end{align*}
which bind \targs  to an array-like object corresponding to the arguments
passed to that function, and include the branch condition under which the
call to @assert@ occurs.
%


\subsection{Analysis}\label{sec:overview-analysis}

Next, we outline how \toolname\ uses refinement types to analyze programs
with closures, polymorphism,  assignments, classes and mutation.

\subsubsection{Polymorphic Instantiation}

\toolname\ uses the framework of Liquid Typing~\cite{LiquidPLDI08}
to \emph{automatically synthesize} the instantiations of (\ref{eq:minindex:inst}).
In a nutshell, \toolname
(a)~creates \emph{templates} for unknown refinement type instantiations,
(b)~performs type-checking over the templates to generate
    \emph{subtyping constraints} over the templates that
    capture value-flow in the program,
(c)~solves the constraints via a \emph{fixpoint} computation
    (abstract interpretation).

\mypara{Step 1: Templates} Recall that @reduce@ has the
polymorphic type $\tyreduce$. At the call-site in @minIndex@,
the type variables @A@, @B@ are instantiated with the \emph{known}
base-type @number@. Thus, \toolname creates fresh templates
for the (instantiated) @A@, @B@:
\begin{align*}
\primty{A} \ \mapsto \ \reftp{\tnumber}{\kvar_A}  \qquad 
\primty{B} \ \mapsto \ \reftp{\tnumber}{\kvar_B}
\end{align*}
where the \emph{refinement variables} $\kvar_A$ and $\kvar_B$
represent the \emph{unknown refinements}. We substitute the
above in the signature for @reduce@ to obtain a
\emph{context-sensitive} template:
\begin{equation} \label{sig:reduceinst}
\funtype{\bind{\primty{a}}{\tarray{\vkvar{A}}},\ {\funtype{\vkvar{B}, \vkvar{A}, \tidx{\primty{a}}}{\vkvar{B}}},\ {\vkvar{B}}}{\vkvar{B}}
\end{equation}

\mypara{Step 2: Constraints}
Next, \toolname generates \emph{subtyping} constraints
over the templates. Intuitively, the templates describe
the \emph{sets} of values that each static entity
(\eg variable) can evaluate to at runtime.
The subtyping constraints capture the \emph{value-flow}
relationships \eg at assignments, calls and returns,
to ensure that the template solutions -- and hence
inferred refinements -- soundly over-approximate
the set of runtime values of each corresponding
static entity.

We generate constraints by performing type checking
over the templates. As @a@, @0@, and @step@ are passed
in as arguments, we check that they respectively have
the types $\tarray{\vkvar{A}}$, $\vkvar{B}$ and
%
${\funtype{\vkvar{B}, \vkvar{A}, \tidx{\primty{a}}}{\vkvar{B}}}$.
%
Checking @a@ and @0@ yields the subtyping constraints:
$$
\Gamma \ \vdash\ \tarray{\tnumber} \ \lqsubt\ \tarray{\vkvar{A}} \qquad
\Gamma \ \vdash\ \refp{\vv = 0}    \ \lqsubt\ \vkvar{B}
$$%
%
where
${\Gamma \ \doteq \ \bind{\primty{a}}{\tarray{\tnumber}},\ 0 < \tlen{\primty{a}}}$
from the \emph{else-guard} that holds at the call to \treduce.
We check @step@ by checking its body under the environment $\Gamma_\tstep$
that binds the input parameters to their respective types: 
$$\Gamma_\tstep \ \doteq \ \bind{\primty{min}}{\vkvar{B}},\ \bind{\primty{cur}}{\vkvar{a}},\ \bind{i}{\tidx{\primty{a}}}$$
As $\primty{min}$ is used to index into the array $\primty{a}$ we get:
$$\Gamma_{\tstep} \ \vdash\ \vkvar{B} \ \lqsubt\ \tidx{a}$$
As $\primty{i}$ and $\primty{min}$ flow to the output type $\vkvar{B}$, we get:
$$
\Gamma_\tstep \ \vdash\ \tidx{\primty{a}} \ \lqsubt\ \vkvar{B} \qquad
\Gamma_\tstep \ \vdash\ \vkvar{B} \ \lqsubt\ \vkvar{B}
$$


\mypara{Step 3: Fixpoint}
The above subtyping constraints over the $\kvar$ variables
are reduced via the standard rules for co- and contra-variant
subtyping, into \emph{Horn implications} over the $\kvar$s.
%
%
\toolname\ solves the Horn implications via (predicate) abstract
interpretation~\cite{LiquidPLDI08} to obtain the solution
${\vkvar{A} \mapsto\ \vtrue}$ and  ${\vkvar{B} \mapsto\ 0 \leq \vv < \tlen{\primty{a}}}$
which is exactly the instantiation in (\ref{eq:minindex:inst})
that satisfies the subtyping constraints, and proves @minIndex@
is array-safe.

\subsubsection{Assignments}\label{sec:overview:assignments}

Next, let us see how the signature for @reduce@ in Figure~\ref{fig:reduce}
is verified by \toolname. Unlike in the functional setting, where
refinements have previously been studied, here, we must deal with
imperative features like assignments and @for@-loops.


\mypara{SSA Transformation}
We solve this problem in three steps. First, we convert
the code into SSA form, to introduce new binders at each
assignment. Second, we generate fresh templates that represent
the unknown types (\ie set of values) for each $\phi$ variable.
Third, we generate and solve the subtyping constraints to infer
the types for the $\phi$-variables, and hence, the ``loop-invariants''
needed for verification.

Let us see how this process lets us verify @reduce@ from Figure~\ref{fig:reduce}.
First, we convert the body to SSA form (\S\ref{sec:formal-language}) %
%
%
\begin{code}
  function reduce(a, f, x) {
    var r0 = x, i0 = 0;
    while [i2,r2 = \phi((i0, r0), (i1, r1))]
          (i2 < a.length) {
      r1 = f(r2, a[i2], i2); i1 = i2 + 1;
    }
    return r2;
  }
\end{code}
where @i2@ and @r2@ are the $\phi$ variables for @i@ and @r@ respectively.
Second, we generate templates for the $\phi$ variables: 
\begin{equation}
\bind{\ttt{i2}}{\reftp{\tnumber}{\vkvar{i2}}} \qquad \bind{\ttt{r2}}{\reftp{\primty{B}}{\vkvar{r2}}}
\label{eq:templates}
\end{equation}
We need not generate templates for the SSA
variables @i0@, @r0@, @i1@ and @r1@ as their types are
those of the expressions they are assigned.
Third, we generate subtyping constraints as before; the
$\phi$ assignment generates \emph{additional} constraints:
$$\begin{array}{lcl}
\Gamma_0 \ \vdash\ \refp{\vv = \primty{i0}} \ \lqsubt \vkvar{i2} & &
\Gamma_1 \ \vdash\ \refp{\vv = \primty{i1}} \ \lqsubt \vkvar{i2} \\
\Gamma_0 \ \vdash\ \refp{\vv = \primty{r0}} \ \lqsubt \vkvar{r2} & &
\Gamma_1 \ \vdash\ \refp{\vv = \primty{r1}} \ \lqsubt \vkvar{r2}
\end{array}$$
%
where $\Gamma_0$ is the environment at the ``exit'' of
the basic blocks where \primty{i0,r0} are defined: 
%
%
\begin{align*}
\Gamma_0 \doteq\ & \bind{\ttt{a}}{\kw{number}\kw{[]}},\;\bind{\ttt{x}}{\kw{B}}
                 ,\ \bind{\ttt{i0}}{\tnatN{0}}
                 ,\ \bind{\ttt{r0}}{\reftp{\kw{B}}{\vv=\ttt{x}}} \\
\intertext{Similarly, the environment $\Gamma_1$ includes bindings for variables
\ttt{i1} and \ttt{r1}. In addition, code executing the loop body has
passed the conditional check, so our \emph{path-sensitive} environment
is strengthened by the corresponding guard:}
\Gamma_1 \doteq\ & \Gamma_0
                 ,\ \bind{\ttt{i1}}{\tnatN{\ttt{i2} + 1}}
                 ,\ \bind{\ttt{r1}}{\kw{B}}
                 ,\ \ttt{i2} < \tlen{\ttt{a}}
\end{align*}
Finally, the above constraints are solved to: $$\vkvar{i2} \mapsto\ 0
\leq \vv < \tlen{\primty{a}} \qquad \vkvar{r2} \mapsto\ \vtrue$$ which
verifies that the ``callback'' @f@ is indeed called with values of
type $\tidx{\primty{a}}$, as it is only called with
$\bind{\primty{i2}}{\tidx{\primty{a}}}$, obtained by plugging the
solution into the template in (\ref{eq:templates}).

\subsubsection{Mutation}

In the imperative, object-oriented setting (common to
dynamic scripting languages), we must account for \emph{class}
and \emph{object} invariants and their preservation in the
presence of field mutation. For example, consider the code
in Figure~\ref{fig:array2d}, modified from the Octane Navier-Stokes benchmark.

\begin{figure}[t]
\begin{code}
  type ArrayN<T,n> = {v:T[] | len(v) = n}
  type grid<w,h>   = ArrayN<number,(w+2)*(h+2)>
  type okW         = natLE<this.w>
  type okH         = natLE<this.h>

  class Field {
    immutable w : pos;
    immutable h : pos;
    dens        : grid<this.w, this.h>;

    constructor(w:pos,h:pos,d:grid<w,h>){
      this.h = h; this.w = w; this.dens = d;
    }
    setDensity(x:okW, y:okH, d:number) {
      var rowS = this.w + 2;
      var i    = x+1 + (y+1) * rowS;
      this.dens[i] = d;
    }
    getDensity(x:okW, y:okH) : number {
      var rowS = this.w + 2;
      var i    = x+1 + (y+1) * rowS;
      return this.dens[i];
    }
    reset(d:grid<this.w,this.h>){
      this.dens = d;
    }
  }
\end{code}
\caption{Two-Dimensional Arrays}
\label{fig:array2d}
\end{figure}

\mypara{Class Invariants} Class @Field@ implements a
\emph{2-dimensional} vector, ``unrolled'' into a single
array @dens@, whose size is the product of the @w@idth
and @h@eight fields. We specify this invariant by
requiring that width and height be strictly
positive (\ie @pos@) and that @dens@ be a @grid@
with dimensions specified by @this.w@ and @this.h@.
An advantage of SMT-based refinement typing is that
modern SMT solvers support non-linear reasoning, which
lets \toolname specify and verify program specific
invariants outside the scope of generic bounds
checkers.

\mypara{Mutable and Immutable Fields} The above invariants
are only meaningful and sound if fields @w@ and @h@
cannot be modified after object creation.
We specify this via the @immutable@ qualifier, which is
used by \toolname to then
(1)~\emph{prevent} updates to the field outside the @constructor@, and
(2)~\emph{allow} refinements of fields (\eg @dens@) to soundly
   refer to the values of those immutable fields.

\mypara{Constructors} We can create \emph{instances} of
@Field@, by using @new Field(...)@ which invokes the
@constructor@ with the supplied parameters. \toolname
ensures that at the \emph{end} of the constructor,
the created object actually satisfies all specified
class invariants \ie field refinements. Of course, this
only holds if the parameters passed to the constructor
satisfy certain preconditions, specified via
the input types. Consequently, \toolname accepts the first
call, but rejects the second:
\begin{code}
  var z = new Field(3,7,new Array(45)); // OK
  var q = new Field(3,7,new Array(44)); // BAD
\end{code}

\mypara{Methods} \toolname uses class invariants to
verify @setDensity@ and @getDensity@, that are
checked \emph{assuming} that the fields of @this@ enjoy
the class invariants, and
method inputs satisfy their given types. The resulting VCs
are valid and hence, check that the methods are array-safe.
Of course, clients must supply appropriate arguments to the
methods. Thus, \toolname accepts the first call, but rejects
the second as the @x@ co-ordinate @5@ exceeds the actual width
(\ie @z.w@), namely @3@:
\begin{code}
  z.setDensity(2, 5, -5)   // OK
  z.getDensity(5, 2);      // BAD
\end{code}

\mypara{Mutation} The @dens@ field is \emph{not} @immutable@
and hence, may be updated outside of the constructor. However,
\toolname requires that the class invariants still hold, and
this is achieved by ensuring that the \emph{new} value assigned
to the field also satisfies the given refinement. Thus, the
@reset@ method requires inputs of a specific size, and
updates @dens@ accordingly. Hence:
\begin{code}
  var z = new Field(3,7,new Array(45));
  z.reset(new Array(45));  // OK
  z.reset(new Array(5));   // BAD
\end{code}


\section{Formal System}\label{sec:language}

Next, we formalize the ideas outlined in \S\ref{sec:overview}.
We introduce our formal core \srcLang:
an imperative, mutable, object-oriented subset of \lang,
that closely follows the design of \cfj~\cite{Nystrom2008},
(the language used to formalize X10), which in turn is based on
Featherweight Java~\cite{Igarashi2001}.
To ease refinement reasoning,
we translate \srcLang\ to a functional, yet still mutable,
intermediate language \ssaLang.
We then formalize our static semantics in terms of \ssaLang.

\subsection{Formal Language}\label{sec:formal-language}

\subsubsection{Source Language (\srcLang)}
%
%
The syntax of this language is given below.
Meta-variable $\srcExpr$ ranges over expressions, which can be
variables           $\srcEvar$,
constants           $\srcVconst$,
property accesses   $\srcDotref{\srcExpr}{\srcFieldname}$,
method calls        $\srcMethcall{\srcExpr}{\srcMethodname}{\many{\srcExpr}}$,
object construction $\srcNewexpr{\cname}{\many{\srcExpr}}$, and
cast operations     $\srcCast{\rtype}{\srcExpr}$.
Statements $\srcStmt$ include
variable declarations,
field updates,
assignments,
conditionals,
concatenations and
empty statements.
Method declarations include a type signature, specifying input and
output types, and a body, \ie a statement immediately followed by a
returned expression.
%
Class definitions distinguish between immutable and
mutable members, using $\srcImmfieldbinding{\srcFieldname}{\rtype}$
and $\srcMutfieldbinding{\srcFieldname}{\rtype}$, respectively.
As in CFJ, each class and method definition is associated with an
invariant $\pred$.
%
%
\[
  \begin{array}{rcl}
%
%
  \srcExpr
    & \prod  &  \srcEvar
      \spmid    \srcVconst
      \spmid    \srcThis
      \spmid    \srcDotref{\srcExpr}{\srcFieldname}
      \spmid    \srcMethcall{\srcExpr}{\srcMethodname}{\many{\srcExpr}}
      \spmid    \srcNewexpr{\cname}{\many{\srcExpr}}
      \spmid    \srcCast{\rtype}{\srcExpr}  \\
%
%
  \srcStmt
    & \prod   & \srcVarDecl{\srcEvar}{\srcExpr}
      \spmid    \srcDotassign{\srcExpr}{\srcFieldname}{\srcExpr}
      \spmid    \srcAssign{\srcEvar}{\srcExpr}
      \spmid    \srcIte{\srcExpr}{\srcStmt}{\srcStmt} \spmid \\
    &         & \srcSeq{\srcStmt}{\srcStmt}
      \spmid    \srcSkip \\
\srcBody
    & \prod   & \srcSeq{\srcStmt}{\srcReturn{\srcExpr}}   \\
  \srcMethodDefsSym
    & \prod   & \srcMethdef{\srcMethodname}{\srcEvar}{\rtype}{\pred}{\rtype}
                                            {\srcBody}    \\
%
  \srcFieldsSym
    & \prod   & \cdot                                         \spmid
                \srcImmfieldbinding{\fieldname}{\rtype}       \spmid
                \srcMutfieldbinding{\fieldname}{\rtype}       \spmid
                \concatenate{\srcFieldsSym_1}{\srcFieldsSym_2}  \\
  \srcCldeclname
    & \prod   & \srcClassdecl{\cname}{\pred}{\rtypec}
                             {\srcFieldsSym}{\srcMethodDefsSym} \\
%
%
\end{array}
\]

The core system does not formalize:
(a)~method overloading,
which is orthogonal to the current contribution and has been
investigated in previous work~\cite{rsc-ecoop15}, or
(b)~method overriding, which means that method names
are distinct from the ones defined in parent classes.
%

\subsubsection{Intermediate Language (\ssaLang)}\label{subsec:ssalang}
\srcLang, while syntactically similar to \ts,
is not entirely suitable for refinement
type checking in its current form, due
to features like assignment.
To overcome this challenge we translate \srcLang to a
functional language \ssaLang through a Static Single
Assignment (SSA) transformation, which produces programs
that are equivalent (in a sense that we will make precise
in the sequel).
%
%
%
In \ssaLang, statements are replaced by let-bindings and
new variables are introduced for each variable being reassigned in
the respective \srcLang code.
Thus, \ssaLang has the following syntax:
\[
  \begin{array}{rcl}
%
%
  \expr
    & \prod   &   \evar                                       \spmid
                  \vconst                                     \spmid
                  \ethis                                      \spmid
                  \dotref{\expr}{\fieldname}                  \spmid
                  \methcall{\expr}{\methodname}{\many{\expr}} \spmid
                  \newexpr{\cname}{\many{\expr}}              \spmid
\\
    &         &   \cast{\rtype}{\expr}                        \spmid
                  \dotassign{\expr}{\fieldname}{\expr}        \spmid
                  \ssactxidx{\ssactx}{\expr}                  \spmid
\\
  \ssactx     & \prod   &   \hole \spmid
                            \letin{\evar}{\expr}{\hole} \spmid
\\
  & &                       \letifshort
                              {\many{\evar}, \many{\evar}_1, \many{\evar}_2}
                              {\expr}
                              {\ssactx_1}{\ssactx_2}{\hole}
\\
  \fieldsSym  & \prod   &   \cdot                                         \spmid
                            \immfieldbinding{\fieldname}{\rtype}       \spmid
                            \mutfieldbinding{\fieldname}{\rtype}       \spmid
                            \concatenate{\fieldsSym_1}{\fieldsSym_2}
\\
  \methodDefsSym  & \prod &   \cdot
                            \spmid
                            \methdef{\methodname}{\evar}{\rtype}
                                    {\pred}{\rtype}{\expr}
                            \spmid
                            \concatenate{\methodDefsSym_1}{\methodDefsSym_2}
\\
  \ssaCldeclSym & \prod   &   \classdecl{\cname}{\pred}{\rtypec}{\fieldsSym}{\methodDefsSym}
%
\end{array}
\]

\begin{figure*}[t]
\judgementHeadFour{SSA Transformation}
  {\tossaexpr{\ssaenv}{\srcExpr}{\ssaexpr}}
  {\tossastmt{\ssaenv}{\srcStmt}{\ssactx}{\ssaenv'}}
  {\tossabody{\ssaenv}{\srcBody}{\expr}}
  {\tossameth{\srcMethodDefsSym}{\methodDefsSym}}
\begin{mathpar}
\inferrule[\ssavarref]
{}{\tossaexpr{\ssaenv}{\srcEvar}{\idxMapping{\ssaenv}{\srcEvar}}}
\and
\inferrule[\ssathis]
{}{\tossaexpr{\ssaenv}{\srcThis}{\ethis}}
\and
\inferrule*[left=\ssavardecl]
{
  \tossaexpr{\ssaenv}{\srcExpr}{\ssaexpr}
  \\
  \ssaenv' = \updMapping{\ssaenv}{\srcEvar}{\evar}
  \\
  \fresh{\evar}
}
{
  \tossastmt{\ssaenv}{\srcVarDecl{\srcEvar}{\srcExpr}}
          {\letin{\evar}{\ssaexpr}{\hole}}{\ssaenv'}
}
\and
%
\inferrule[\ssaite]
{
  \tossaexpr{\ssaenv}{\srcExpr}{\ssaexpr}
  \\
  \tossastmt{\ssaenv}{\srcStmt_1}{\ssactx_1}{\ssaenv_1}
  \\
  \tossastmt{\ssaenv}{\srcStmt_2}{\ssactx_2}{\ssaenv_2}
  \\\\
  \ptriplet{\many{\srcEvar}}{\many{\evar}_1}{\many{\evar}_2} 
    = \envDiff{\ssaenv_1}{\ssaenv_2}  
  \\
  \ssaenv' = \updMapping{\ssaenv}{\many{\srcEvar}}{\many{\evar}'}
  \\
  \fresh{\many{\evar}'}
}
{
  \tossastmt{\ssaenv}
          {\srcIte{\srcExpr}{\srcStmt_1}{\srcStmt_2}}
          {
          \letif{\many{\evar}'}{\many{\evar}_1}{\many{\evar}_2}
                  {\ssaexpr}{\ssactx_1}{\ssactx_2}{\hole}
          }
          {\ssaenv'}
}
\and
\inferrule[\ssaasgn]
{
  \tossaexpr{\ssaenv}{\srcExpr}{\ssaexpr}
  \\
  \evar = \idxMapping{\ssaenv}{\srcEvar}
  \\\\
  \ssaenv' = \updMapping{\ssaenv}{\srcEvar}{\evar'}
  \\
  \fresh{\evar'}
}
{
  \tossastmt{\ssaenv}
          {\srcAssign{\srcEvar}{\srcExpr}}
          {\letin{\evar'}{\ssaexpr}{\hole}}
          {\ssaenv'}
}
\and
\inferrule[\ssadasgn]{
  \tossaexpr{\ssaenv}{\srcExpr}{\ssaexpr}
  \\
  \tossaexpr{\ssaenv}{\srcExpr'}{\ssaexpr'}
}
{
  \tossastmt{\ssaenv}{\srcDotassign{\srcExpr}{\srcFieldname}{\srcExpr'}}
          {\letin{\_}{\dotassign{\ssaexpr}{\fieldname}{\ssaexpr'}}{\hole}}{\ssaenv}
}
\and
\inferrule[\ssaseq]
{
  \tossastmt{\ssaenv}{\srcStmt_1}{\ssactx_1}{\ssaenv_1}
  \\
  \tossastmt{\ssaenv_1}{\srcStmt_2}{\ssactx_2}{\ssaenv_2}
}
{
  \tossastmt{\ssaenv}{\srcSeq{\srcStmt_1}{\srcStmt_2}}
          {\ssactxidx{\ssactx_1}{\ssactx_2}}
          {\ssaenv_2}
}
\and
\inferrule[\ssaskip]
{}
{
  \tossastmt{\ssaenv}{\srcSkip}{\hole}{\ssaenv}
}
\and
\inferrule[\ssabody]
{
  \tossastmt{\ssaenv}{\srcStmt}{\ssactx}{\ssaenv'}
  \\
  \tossaexpr{\ssaenv'}{\srcExpr}{\ssaexpr}
}
{
  \tossaexpr{\ssaenv}{\srcSeq{\srcStmt}{\srcReturn{\srcExpr}}}
           {\ssactxidx{\ssactx}{\expr}}
}
\and
\inferrule[\ssamethdecl]
{
  \toString{\srcMethodname} = \toString{\methodname}
  \\
  \ssaenv = \mapping{\many{\srcEvar}}{\many{\evar}},
            \mapping{\srcThis}{\ethis}
  \\\\
  \tossabody{\ssaenv}{\srcBody}{\expr}
  \\
  \fresh{\methodname, \many{\evar}}
}
{
  \tossameth{
    \srcMethdef{\srcMethodname}{\srcEvar}{\rtype}{\pred}{\rtype}
               {\srcBody} 
  }
  {
    \methdef{\methodname}{\evar}{\rtype}{\pred}{\rtype}{\expr}
  }
}
\end{mathpar}
\nocaptionrule
\caption{Selected SSA Transformation Rules}
\label{fig:ssa}
\end{figure*}

The majority of the expression forms $\expr$ are unsurprising.
An exception is the form of the \emph{SSA context} $\ssactx$,
which corresponds to the translation of a statement $\srcStmt$ and
contains a \emph{hole} $\hole$ that will hold the translation
of the continuation of $\srcStmt$.

\paragraph{SSA Transformation}
Figure~\ref{fig:ssa} describes the SSA transformation, that
uses a \emph{translation environment} \ssaenv,
to map \srcLang variables \srcEvar to \ssaLang variables \evar.
The translation of expressions \srcExpr to \expr is
routine: as expected, \ssavarref maps the source level \evar to the
current binding of \evar in \ssaenv.
The translating judgment of
statements \srcStmt has the form:
$\tossastmt{\ssaenv}{\srcStmt}{\ssactx}{\ssaenv'}$.
The output environment $\ssaenv'$ is used for the
translation of the expression that will fill the hole in $\ssactx$.




The most interesting case is that of the conditional statement (rule \ssaite).
The conditional expression and each branch are translated separately.
To compute variables that get updated in either branch ($\Phi$-variables),
we combine the produced translation states $\ssaenv_1$ and $\ssaenv_2$ as
$\envDiff{\ssaenv_1}{\ssaenv_2}$ defined as:
$$
  \{ \ptriplet{\srcEvar}{\evar_1}{\evar_2} \;|\;
         \mapping{\srcEvar}{\evar_1} \in \ssaenv_1
         ,\,  \mapping{\srcEvar}{\evar_2} \in \ssaenv_2
         ,\,  \evar_1 \neq \evar_2
  \}
$$
Fresh $\Phi$-variables $\many{\evar}'$
populate the output SSA environment $\ssaenv'$.
Along with the versions of the $\Phi$-variables for each
branch ($\many{\evar}_1$ and $\many{\evar}_2$),
they are used to annotate the produced structure.


Assignment statements introduce a new SSA variable
and bind it to the updated source-level variable (rule \ssaasgn).
Statement sequencing is emulated with nesting SSA contexts (rule
\ssaseq); empty statements introduce a hole (rule \ssaskip); and,
finally, method declarations fill in the hole introduced by the method
body with the translation of the return expression (rule \ssamethdecl).

\subsubsection{Consistency}\label{subsec:consistency}
To validate our transformation, we provide a consistency result that
guarantees that stepping in the target language preserves the
transformation relation, after the program in the source language has
made an appropriate number of steps.
We define a \emph{runtime configuration} $\srcRtConf$ for \srcLang
(resp. $\ssaRtConf$ for \ssaLang)
for a \emph{program} $\srcProg$ (resp. $\ssaProg$) as:
\[
  \begin{array}{ll}
  \srcProg    \defeq \pair{\srcSignatures}{\srcBody} \qquad\qquad & \ssaProg \defeq \pair{\ssaSignatures}{\expr}          \\
  \srcRtConf  \defeq \pair{\srcRtState}{\srcBody}    \qquad\qquad       & \ssaRtConf \defeq \pair{\ssaRtState}{\expr}           \\
  \srcRtState \defeq\quadruple{\srcSignatures}{\srcStore}{\srcStack}{\srcHeap} \qquad \qquad & \ssaRtState \defeq \pair{\ssaSignatures}{\ssaHeap}
\end{array}
\]
Runtime state $\srcRtState$ consists of
the call stack $\srcStack$,
the local store of the current stack frame $\srcStore$
and the heap $\srcHeap$.
The runtime state for \ssaLang,  $\ssaRtConf$
only consists of the signatures $\ssaSignatures$ and a heap $\ssaHeap$.

\begin{figure*}[!t]
  \judgementHeadTwo
    {Operational Semantics for \srcLang}
    {\stepsFour{\srcRtState}{\srcExpr}{\srcRtState'}{\srcExpr'}}
    {\stepsFour{\srcRtState}{\srcStmt}{\srcRtState'}{\srcStmt'}}
\begin{mathpar}
%
\inferrule[\srcRcVal]
{}
{\stepsFour{\srcRtState}{\srcVal}{\srcRtState}{\srcSkip}}
\and
\inferrule[\srcRcEvalCtx]
{
  \stepsFour{\quadruple{\srcSignatures}{\srcStore }{\cdot}{\srcHeap }}{\srcExpr}
            {\quadruple{\srcSignatures}{\srcStore'}{\cdot}{\srcHeap'}}{\srcExpr'}
}
{
  \stepsFour{\quadruple{\srcSignatures}{\srcStore }{\srcStack}{\srcHeap }}{\idxEctx{\srcEvalCtx}{\srcExpr}}
            {\quadruple{\srcSignatures}{\srcStore'}{\srcStack}{\srcHeap'}}{\idxEctx{\srcEvalCtx}{\srcExpr'}}
}
\and
\inferrule[\srcRcVar]
{}
{
  \stepsFour{\srcRtState}{\srcEvar}
            {\srcRtState}{\idxMapping{\accessState{\srcRtState}{\srcStore}}{\srcEvar}}
}
\and
\inferrule[\srcRcDotRef]
{
  \idxMapping{\accessState{\srcRtState}{\srcHeap}}{\srcLoc} = 
    \heapObject{\srcLoc'}{\srcFieldDefsSym}
  \\\\
  \srcFieldDef{\srcFieldname}{\srcVal} \in \srcFieldDefsSym 
}
{
  \stepsFour{\srcRtState}{\srcDotref{\srcLoc}{\srcFieldname}}
            {\srcRtState}{\srcVal}
}
\and
%
\inferrule*[left=\srcRcNew]
{
  \idxMapping{\srcHeap}{\srcLoc_0} = \heapClassObject{\cname}{\srcLoc_0'}{\srcMethodDefsSym}
  \\
  \srcRtFields{\srcSignatures}{\cname} = \srcFieldbindings{\srcFieldname}{\rtype}
  \\\\
  \srcObject = \heapObject{\srcLoc_0}{\srcFieldDefs{\srcFieldname}{\srcVal}}
  \\
  \srcHeap' = \updMapping{\srcHeap}{\srcLoc}{\srcObject}
  \\
  \fresh{\srcLoc}
}
{
  \stepsFour{\quadruple{\srcSignatures}{\srcStore}{\srcStack}{\srcHeap}}
            {\srcNewexpr{\cname}{\many{\srcVal}}}
            {\quadruple{\srcSignatures}{\srcStore}{\srcStack}{\srcHeap'}}
            {\srcLoc}
}
\and
\inferrule[\srcRcCast]
{}
{
  \stepsFour{\srcRtState}{\srcCast{\rtype}{\srcExpr}}
            {\srcRtState}{\srcCheck{\rtype}{\srcExpr}}
}
\and
%
\inferrule[\srcRcVarDecl]
{
  \srcStore' = \updMapping{\accessState{\srcRtState}{\srcStore}}{\srcEvar}{\srcVal}
}
{
  \stepsFour{\srcRtState}{\srcVarDecl{\srcEvar}{\srcVal}}
            {\updState{\srcRtState}{\srcStore'}}{\srcSkip}
}
\and
\inferrule[\srcRcCall]
{
  \srcResolveMeth{\srcHeap}{\srcLoc}{\srcMethodname}{
    \srcMethdefUnty{\srcMethodname}{\srcEvar}{\srcSeq{\srcStmt}{\srcReturn{\srcExpr}}}}
  \\\\
  \srcStore' = \concatenate{\mappings{\srcEvar}{\srcVal}}{\mapping{\srcThis}{\srcLoc}}
  \\
  \srcStack' = \concatenate{\srcStack}{\srcStore,\srcEvalCtx}
}
{
  \stepsFour{\quadruple{\srcSignatures}{\srcStore}{\srcStack}{\srcHeap}}
            {\idxEctx{\srcEvalCtx}{\srcMethcall{\srcLoc}{\srcMethodname}{\many{\srcVal}}}}
            {\quadruple{\srcSignatures}{\srcStore'}{\srcStack'}{\srcHeap}}
            {\srcSeq{\srcStmt}{\srcReturn{\srcExpr}}}
}
\and
\inferrule[\srcRcDotAsgn]
{
  \srcHeap' = \updMapping{\accessState{\srcRtState}{\srcHeap}}{\srcLoc}
                         {\updMapping{\idxMapping{\accessState{\srcRtState}{\srcHeap}}{\srcLoc}}
                                     {\srcFieldname}{\srcVal}}
}
{
  \stepsFour{\srcRtState}{\srcDotassign{\srcLoc}{\srcFieldname}{\srcVal}}
            {\updState{\srcRtState}{\srcHeap'}}
            {\srcVal}
}
\and
\inferrule[\srcRcAssign]
{
  \srcStore' = \updMapping{\accessState{\srcRtState}{\srcStore}}{\srcEvar}{\srcVal}
}
{
  \stepsFour{\srcRtState}{\srcAssign{\srcEvar}{\srcVal}}
            {\updState{\srcRtState}{\srcStore'}}{\srcVal}
}
\and
\inferrule[\srcRcIte]
{
  \srcVconst = \vtrue  \imp \index = 1  \\\\
  \srcVconst = \vfalse \imp \index = 2
}
{
  \stepsFour{\srcRtState}{\srcIte{\srcVconst}{\srcStmt_1}{\srcStmt_2}}
            {\srcRtState}{\srcStmt_{\index}}
}
\and
\inferrule[\srcRcRet]
{
  \accessState{\srcRtState}{\srcStack} = \concatenate{\srcStack'}{\srcStore,\srcEvalCtx}
}
{\stepsFour{\srcRtState}
           {\srcReturn{\srcVal}}
           {\updState{\srcRtState}{\srcStack', \srcStore}}
           {\idxEctx{\srcEvalCtx}{\srcVal}}
}
\and
\inferrule[\srcRcSkip]
{}
{\stepsFour{\srcRtState}{\srcSeq{\srcSkip}{\srcStmt}}{\srcRtState}{\srcStmt}}
\end{mathpar}

\judgementHead{Operational Semantics for \ssaLang}
              {\stepsFour{\ssaRtState}{\expr}{\ssaRtState'}{\expr'}}
\begin{mathpar}
\inferrule[\rcevalctx]
{
  \stepsFour{\ssaRtState}{\expr}{\ssaRtState'}{\expr'}
}
{
  \stepsFour{\ssaRtState}{\idxEctx{\ssaEvalCtx}{\expr}}{\ssaRtState'}{\idxEctx{\ssaEvalCtx}{\expr'}}
}
\and
\inferrule[\rfield]
{
  \idxMapping{\accessState{\ssaRtState}{\ssaHeap}}{\loc} = 
    \heapObject{\loc'}{\fieldDefsSym}
  \\\\
  \fieldDef{\fieldname}{\val} \in \fieldDefsSym 
}
{
  \stepsFour{\ssaRtState}
        {\dotref{\loc}{\fieldname}}
        {\ssaRtState}
        {\val}
}
\and
\inferrule[\rinvoke]
{
  \resolveMeth{\ssaHeap}{\loc}{\methodname}
              {\parens{\methdef{\methodname}{\evar}{\rtypeb}{\pred}{\rtype}{\expr}}}
  \\\\
  \eval{\appsubst{\tsubsttwo{\many{\val}}{\many{\evar}}{\loc}{\ethis}}{\pred}} = \ssaTrue
%
}
{
  \stepsFour{\ssaRtState}
        {\methcall{\loc}{\methodname}{\many{\val}}}
        {\ssaRtState}
        {\appsubst{\tsubsttwo{\many{\val}}{\many{\evar}}{\loc}{\ethis}}{\expr}}
}
\and
\inferrule[\rnew]
{
  \idxMapping{\srcHeap}{\loc_0} = \heapClassObject{\cname}{\loc_0'}{\methodDefsSym}
  \\
  \rtFields{\ssaSignatures}{\cname} = \fieldbindings{\fieldname}{\rtype}
  \\\\
  \ssaObject = \heapObject{\loc_0}{\fieldDefs{\fieldname}{\val}}
  \\
  \ssaHeap' = \updMapping{\ssaHeap}{\loc}{\ssaObject}
  \\
  \fresh{\loc}
}
{
  \stepsFour{\pair{\ssaSignatures}{\ssaHeap}}
            {\newexpr{\cname}{\many{\val}}}
            {\pair{\ssaSignatures}{\ssaHeap'}}
            {\loc}
}
\and
\inferrule[\rletin]
{}
{
  \stepsFour{\ssaRtState}{\letin{\evar}{\val}{\expr}}
            {\ssaRtState}
            {\appsubst{\esubst{\val}{\evar}}{\expr}}
}
\and
%
\inferrule[\rdotasgn]
{
  \ssaHeap' = \updMapping{\accessState{\ssaRtState}{\ssaHeap}}{\loc}
                         {\updMapping{\idxMapping{\accessState{\ssaRtState}{\ssaHeap}}{\loc}}
                                     {\fieldname}{\val}}
}
{
  \stepsFour{\ssaRtState}{\dotassign{\loc}{\fieldname}{\val}}
            {\updState{\ssaRtState}{\ssaHeap'}}{\val}
}
\and
\inferrule[\rcast]
{
  \wfconstr{\env}{\varbinding{\idxMapping{\ssaRtState}{\loc}}{\rtypeb}; \rtypeb\leq\rtype}
}
{
  \stepsFour{\ssaRtState}{\cast{\rtype}{\loc}}{\ssaRtState}{\loc}
}
\and
%
\inferrule[\rletif] {
  \vconst = \ssaTrue  \imp \index = 1 
  \\
  \vconst = \ssaFalse \imp \index = 2
}
{
  \stepsFour{\ssaRtState}{\letif
      {\many{\evar}}{\many{\evar}_1}{\many{\evar}_2}
      {\vconst}{\ssactx_{1}}{\ssactx_{2}}{\expr}
    }{\ssaRtState}{\ssactxidx{\ssactx_{\index}}
     {\appsubst{\esubst{\many{\evar_{\index}}}{\many{\evar}}}{\expr}}}
}
\end{mathpar}
\nocaptionrule\caption{Reduction Rules for \srcLang (adapted from Safe
TypeScript~\cite{Rastogi2015}) and \ssaLang}
\label{fig:opsem}
\end{figure*}

We establish the \emph{consistency} of the SSA transformation by means of
a weak forward simulation theorem that connects the dynamic semantics of the
two languages.
To that end, we define small-step operational semantics for
both languages, of the form $\steps{\srcRtConf}{\srcRtConf'}$ and
$\steps{\ssaRtConf}{\ssaRtConf'}$.
Figure~\ref{fig:opsem} presents the dynamic behavior of
the two languages. Rules for \srcLang have been adapted
from Rastogi \etal~\cite{Rastogi2015}.
Note how in rule \srcRcCast the cast operation
reduces to a call to the built-in $\texttt{check}$ function, where
$\llbracket\rtype\rrbracket$ encodes type $\rtype$.
Rules for \ssaLang are mostly routine, with the exception of
rule \rletif: expression $\expr$ has been produced assuming
$\Phi$-variables $\many{\evar}$. After the branch has been
determined we pick the actual $\Phi$-variables ($\many{\evar}_1$ or
$\many{\evar}_2$)
and replace them in $\expr$. This formulation allows
us to perform all the SSA-related book-keeping 
in a single reduction step,
which is key to preserving our consistency \emph{invariant} that
\ssaLang steps faster than \srcLang.

We also extend our SSA
transformation judgment to runtime configurations,
leveraging the SSA environments
that have been statically computed for each program entity, which
now form a
\emph{global SSA environment} $\ssaenvs$,
mapping each AST node ($\srcExpr$, $\srcStmt$, \etc)
to an SSA environment $\ssaenv$:
$$
\ssaenvs \prod \cdot
\spmid \mapping{\srcExpr}{\ssaenv}
\spmid \mapping{\srcStmt}{\ssaenv}
\spmid \dots
\spmid \concatenate{\ssaenvs_1}{\ssaenvs_2}
$$
We assume that the compile-time SSA translation
yields this environment as a side-effect
(\eg
$\tossaexpr{\ssaenv}{\srcExpr}{\expr}$ produces
$\mapping{\srcExpr}{\ssaenv}$
) and the top-level program transformation judgment
returns the net effect:
$\tossaprog{\srcProg}{\ssaProg}{\ssaenvs}$.
Hence, the SSA transformation judgment for configurations becomes:
$\tossartconf{\ssaenvs}{\srcRtState}{\srcBody}{\ssaRtState}{\expr}$.
We can now state our consistency theorem as:
\begin{theorem}[SSA Consistency]\label{theorem:consistency}
For configurations
$\srcRtConf$ and $\ssaRtConf$
and global store typing $\ssaenvs$,
if
$\tossartconfshort{\ssaenvs}{\srcRtConf}{\ssaRtConf}$,
then either
both $\srcRtConf$ and $\ssaRtConf$ are terminal,
or
if for some $\ssaRtConf'$,
  $\steps{\ssaRtConf}{\ssaRtConf'}$,
then there exists $\srcRtConf'$ \st
  $\stepsmanyone{\srcRtConf}{\srcRtConf'}$
  and
  $\tossartconfshort{\ssaenvs}{\srcRtConf'}{\ssaRtConf'}$.
\end{theorem}

%

\subsection{Static Semantics}

Having drawn a connection between source and target language we
can now describe refinement checking procedure in terms of \ssaLang.

\begin{figure*}[htpb]

\judgementHeadTwo{Typing Rules}
{\typecheck{\env}{\expr}{\rtype}}
{\typecheckctx{\env}{\ssactx}{\env'}}
\begin{mathpar}
  \inferrule[\lqchkvar]
{
  \idxMapping{\env}{\evar} = \rtype
}
{
  \typecheck{\env}{\evar}{\singleton{\rtype}{\evar}}
}
\and
\inferrule[\lqchkconst]{}
{\typecheck{\env}{\vconst}{\tconst{\vconst}}}
\and
\inferrule[\lqchkfieldimm]
{
  \typecheck{\env}{\expr}{\rtype}
  \\
  \wfconstr{\envextex{\env}{\evarb}{\rtype}}
           {\varbinding
              {\hasimm{\evarb}{\fieldname_{\index}}}
            {\rtype_{\index}}}
  \\\\
  \fresh{\evarb}
}
{
  \typecheck{\env}
            {\dotref{\expr}{\fieldname_{\index}}}
            {\texist{\evarb}{\rtype}{\singleton{\rtype_{\index}}{\tdotref{\evarb}{\fieldname_{\index}}}}}
}
\and
\inferrule[\lqchkssactx]
{
  \typecheckctx{\env}{\ssactx}{\envbinding{\many{\evar}}{\many{\rtypeb}}}
  \\\\
  \typecheck{\envextex{\env}{\many{\evar}}{\many{\rtypeb}}}{\expr}{\rtype}
}
{
  \typecheck{\env}{\ssactxidx{\ssactx}{\expr}}{\texist{\many{\evar}}{\many{\rtypeb}}{\rtype}}
}
\and
\inferrule[\lqchkfieldmut]
{
  \typecheck{\env}{\expr}{\rtype}
  \\\\
  \typecheck{\envextex{\env}{\evarb}{\rtype}}
            {\hasmut{\evarb}{\fieldnameb_{\index}}}
            {\rtype_{\index}}
  \\\\
  \fresh{\evarb}
}
{
  \typecheck{\env}{\dotref{\expr}{\fieldnameb_{\index}}}{\texist{\evarb}{\rtype}{\rtype_{\index}}}
}
\and
\inferrule[\lqchkinv]
{
  \typechecktwo{\env}{\expr}{\rtype}{\many{\expr}}{\many{\rtype}}
  \\\\
  \wfconstr{\envextex{\env}{\evarb}{\rtype}}
           {
             \has{\evarb}{\parens{\methdef{\methodname}{\evarb}{\rtypec}{\pred}{\rtypeb}{\expr'}}}
           }
       \\\\
  \wfconstr{\envextex{\envextex{\env}{\evarb}{\rtype}}{\many{\evarb}}{\many{\rtype}}}
           {\mconcatenate{\many{\rtype} \subt \many{\rtypec}}{\pred}}
  \\
  \fresh{\mconcatenate{\evarb}{\many{\evarb}}}
}
{
  \typecheck{\env}{\methcall{\expr}{\methodname}{\many{\expr}}}
            {\texist{\evarb}{\rtype}{\texist{\many{\evarb}}{\many{\rtype}}{\rtypeb}}}
}
\and
\inferrule[\lqchkasgn]
{
  \typechecktwo{\env}{\expr_1}{\rtype_1}{\expr_2}{\rtype_2}
  \\\\
  \wfconstr{\envext{\env}{\evarb_1}{\basetype{\rtype_1}}}
           {\mconcatenate
              {\varbinding{\hasmut{\evarb_1}{\fieldname}}{\rtypeb}}
              {\rtype_2 \subt \rtypeb}
           }
  \\\\
  \fresh{\evarb_1}
}
{
  \typecheck{\env}{\dotassign{\expr_1}{\fieldname}{\expr_2}}{\rtype_2}
}
\and
\inferrule[\lqchknew]
{
  \typecheck{\env}{\many{\expr}}{\ppair{\many{\rtype}_{\immindex}}{\many{\rtype}_{\mutindex}}}
  \\
  \wfconstrnopremise{\isclass{\cname}}
  \\
  \wfconstr{\envext{\env}{\evarb}{\cname}}
           {\fieldsOfVar{\evarb} = 
              \mconcatenate{\immfieldbindings{\fieldname}{\rtypec}}
                           {\mutfieldbindings{\fieldnameb}{\rtyped}}
           }
  \\\\
  \wfconstr{
      \envextex
        {\envext{\env}{\evarb}{\cname}}
        {\many{\evarb}_{\immindex}}
      {\singleton{\many{\rtype}_{\immindex}}{\tdotref{\evarb}{\many{\fieldname}}}}
    }
    {
      \mconcatenate{
        \mconcatenate{\many{\rtype}_{\immindex} \subt \many{\rtypec}}
        {\many{\rtype}_{\mutindex} \subt \many{\rtyped}}
      }
      {\classinv{\cname}{\evarb}}
    }
  \\
  \fresh{\evarb, \many{\evarb}}
}
{
  \typecheck{\env}
            {\newexpr{\cname}{\many{\expr}}}
            {
              \texist{\many{\evarb}_{\immindex}}{\many{\rtype}_{\immindex}}
              {\reftp
                {\cname}
                {
                  \concatpreds
                    {\tdotref{\vv}{\many{\fieldname}} = \many{\evarb}_{\immindex}}
                    {\classinv{\cname}{\vv}}
                }
              }
            }
}
%
\and
\inferrule[\lqchkcast]
{
  \typecheck{\env}{\expr}{\rtypeb}
  \\
  \wftype{\env}{\rtype}
  \\\\
  \semsubtype{\env}{\rtypeb}{\rtype}
}
{
  \typecheck{\env}{\cast{\rtype}{\expr}}{\rtype}
}
\and
%
%
%
\inferrule[\lqchkctxemp]{}{\typecheckctx{\env}{\hole}{\cdot}}
\and
\inferrule[\lqchkctxletin]
{
  \typecheck{\env}{\expr}{\rtype}
}
{
  \typecheckctx{\env}{\letin{\evar}{\expr}{\hole}}{\envbinding{\evar}{\rtype}}
}
\and
\inferrule[\lqchkctxletif]
{
  \typechecksubt{\env}{\expr}{\rtypeb}{\rtypeb}{\tbool}
  \\
  \typecheckctx{\envgrdext{\envext{\env}{\evarb}{\rtypeb}}{\evarb}}    {\ssactx_1}{\env_1}
  \\
  \typecheckctx{\envgrdext{\envext{\env}{\evarb}{\rtypeb}}{\neg\evarb}}{\ssactx_2}{\env_2}
  \\\\
  \issubtype{\mconcatenate{\env}{\env_1}}{\idxMapping{\env_1}{\many{\evar}_1}}{\many{\rtype}}
  \\
  \issubtype{\mconcatenate{\env}{\env_2}}{\idxMapping{\env_2}{\many{\evar}_2}}{\many{\rtype}}
  \\
  \wftype{\env}{\many{\rtype}}
  \\
  \fresh{\many{\rtype}}
}
{
  \typecheckctx{\env}{
    \letif{\many{\evar}}{\many{\evar}_1}{\many{\evar}_2}{\expr}{\ssactx_1}{\ssactx_2}{\hole}
  }{\envbinding{\many{\evar}}{\many{\rtype}}}
}
\end{mathpar}

\nocaptionrule
\caption{Static Typing Rules for \ssaLang}
\label{fig:typing}
\end{figure*}

%
\mypara{Types}
Type annotations on the source language are propagated unaltered
through the translation phase.
Our type language (shown below) resembles that of existing refinement type
systems~\cite{Knowles10, LiquidPLDI08, Nystrom2008}.
A \emph{refinement type} $\rtype$ may be an existential type or have the
form $\reftp{\btype}{\pred}$, where $\btype$ is a class name \cname
or a primitive type $\tprim$, and $\pred$ is a logical predicate (over
some decidable logic) which describes the properties that values of
the type must satisfy.
Type specifications (\eg method types) are existential-free,
while inferred types
may be existentially quantified~\cite{Knowles2009}.

\mypara{Logical Predicates}
Predicates \pred are logical formulas over terms $\cterm$.
These terms can be
variables $\evar$,
primitive constants $\vconst$,
the reserved value variable $\vv$,
the reserved variable $\this$ to denote the containing object,
field accesses $\tdotref{\cterm}{\fieldname}$,
uninterpreted function applications $\funcall{\fname}{\many{\cterm}}$
and
applications of terms on built-in operators
$\biname$, such as @==@, @<@, @+@, etc.
%
\[
\begin{array}{rcl}
  \rtype, \rtypeb, \rtypec  & \prod & \texist{\evar}{\rtype_1}{\rtype_2} \spmid  \reftp{\btype}{\pred}            \\
  \btype                    & \prod & \cname \spmid \tprim                                                    \\
  \pred                     & \prod & \pand{\pred_1}{\pred_2} \spmid  \pnot{\pred} \spmid \cterm                  \\
  \cterm                    & \prod & \evar                           \spmid
                                      \vconst                         \spmid
                                      \vv                             \spmid
                                      \this                           \spmid
                                      \tdotref{\cterm}{\fieldname}    \spmid
                                      \funcall{\fname}{\many{\cterm}} \spmid
                                      \funcall{\biname}{\many{\cterm}}
  \end{array}
\]

\mypara{Structural Constraints}
Following CFJ, we reuse the notion of an Object Constraint System, to
encode constraints related to the object-oriented nature of the
program.
Most of the rules carry over to our system; we defer them to
the supplemental material.
The key extension in our setting is we partition
\has{\cname}{\element} (that encodes inclusion of
an element $\element$ in a class $\cname$) into
two cases: \hasmut{\cname}{\element} and
\hasimm{\cname}{\element}, to account for
elements that may be mutated.
These elements can only be fields
(\ie\ there is no mutation on methods).

\mypara{Environments And Well-formedness}
A type environment $\env$ contains \emph{type bindings}
$\envbinding{\evar}{\rtype}$ and \emph{guard predicates} $\pred$ that encode
path sensitivity.
%
%
$\env$ is \emph{well-formed} if all of its bindings are
well-formed. A refinement type is well-formed in an environment $\env$
if all symbols (simple or qualified) in its logical predicate are (i)~bound in
$\env$, and (ii)~correspond to \emph{immutable} fields of objects.
We omit the rest of the well-formedness rules as they are standard in
refinement type systems (details can be found in the supplemental material).


Besides well-formedness, our system's main judgment forms are those
for subtyping and
refinement typing~\cite{Knowles10}.

\mypara{Subtyping} is defined by the judgment
$\issubtype{\env}{\rtypeb}{\rtype}$.
The rules are standard among refinement type systems with existential types.
For example, the rule for subtyping between two refinement types
$\issubtype{\env}{\reftp{\btype}{\pred}}{\reftp{\btype}{\pred'}}$
reduces to a \emph{verification condition}:
$\validimp{\env}{\pred}{\pred'}$,
where $\embed{\env}$ is the embedding of environment $\env$ into our logic
that accounts for both guard predicates and variable bindings:
$$
\embed{\env} \defeq
\bigwedge\braces{\pred\;|\;\pred\in\env}
\wedge\bigwedge
\braces{\appsubst{\tsubst{\evar}{\vv}}{\pred}\text{,}
\;|\;\envbinding{\evar}{\reftp{\btype}{\pred}} \in \env}
$$
Here, we assume existential types have been simplified to non-existential
bindings when they entered the environment.
The full set of rules is included in the supplemental material.



\mypara{Refinement Typing Rules}
Figure~\ref{fig:typing} contains most rules of the two forms of
our typing judgements:
$\typecheck{\env}{\expr}{\rtype}$ and
$\typecheckctx{\env}{\ssactx}{\env'}$.
The first form assigns a type $\rtype$ to an
expression $\expr$ under a typing environment $\env$.
The second form checks the body of an SSA context $\ssactx$
under $\env$ and returns an environment $\env'$ of the
variables introduced in $\ssactx$ that are going to be
available in its hole.
Below, we discuss the novel rules:

\begin{description}[leftmargin=0pt, font=\normalfont]

  \item [\brackets{\lqchkfieldimm}]
    Immutable object parts can be assigned a more precise type,
    by leveraging the preservation of their \emph{identity}.
    This notion, known as \emph{self-strengthening}~\cite{Knowles2009,
    Nystrom2008}, is defined with the aid of the
    \emph{strengthening} operator $\strengthenname$:
    \begin{equation*}
      \begin{aligned}[c]
        \strengthen{\reftp{\nvtype}{\pred}}{\pred'} & \defeq \reftp{\nvtype}{\pred \wedge \pred'}  \\
        \strengthen{\parens{\texist{\evar}{\rtypeb}{\rtype}}}{\pred} & \defeq \texist{\evar}{\rtypeb}{\parens{\strengthen{\rtype}{\pred}}}  \\
        \singleton{\rtype}{\cterm} & \defeq \strengthen{\rtype}{\parens{\vv = \cterm}}
      \end{aligned}
    \end{equation*}

  \item [\brackets{\lqchkfieldmut}] Here we avoid such strengthening,
    as the value of
    field $\fieldnameb_{\index}$ is mutable, so cannot appear in
    refinements.

  \item [\brackets{\lqchknew}] Similarly, only
    immutable fields are referenced in the refinement of the
    inferred type at object construction.

  \item [\brackets{\lqchkinv}] Extracting the method signature using
    the \textsf{has} operator has already performed the necessary
    substitutions to account for the specific receiver object.

  \item [\brackets{\lqchkcast}]
    Cast operations are checked \emph{statically}
    obviating the need for a dynamic check. This rule uses the notion
    of \emph{compatibility subtyping}, which is defined as:

    \begin{definition}[Compatibility Subtype]
      {\label{def:main:compatibility:subtype}}
      A type $\rtypeb$ is a compatibility subtype of a type $\rtype$
      under an environment $\env$ (we write
      $\semsubtype{\env}{\rtypeb}{\rtype}$), iff
      $\anycast{\env}{\rtypeb}{\basetype{\rtype}} = \rtypec \neq
      \fail$ with $\issubtype{\env}{\rtypec}{\rtype}$.
    \end{definition}

    Here, $\basetype{\rtype}$ extracts the base type of $\rtype$, and
    $\anycast{\env}{\rtype}{\cnameb}$ succeeds when under environment
    $\env$ we can statically prove $\cnameb$'s invariants.
    We use the predicate $\classinv{\cnameb}{\vv}$ (as in CFJ), to
    denote the conjunction of the class invariants of $\cname$ and
    its supertypes (with the necessary substitutions of $\this$ by $\vv$).
    We assume that part of these invariants is a predicate that states
    inclusion in the specific class ($\instanceof{\vv}{\cnameb}$).
    Therefore, we can prove that $\rtype$ can safely be cast
    to $\cnameb$. Formally:
    \begin{equation*}
      \begin{aligned}[c]
        \anycast{\env}{\reftp{\cname}{\pred}}{\cnameb} & \defeq
                \begin{cases}
                  \strengthen{\cnameb}{\pred}
                          & \!\!\text{if }    
                                           \embed{\env} \imp \embed{\pred} \imp \classinv{\cnameb}{\vv}  \\
                \fail     & \!\!\text{otherwise}
              \end{cases}
        \\
        \anycast{\env}{\texist{\evar}{\rtypeb}{\rtype}}{\cnameb} & \defeq
          \texist{\evar}{\rtypeb}{\anycast{\envextex{\env}{\evar}{\rtypeb}}{\rtype}{\cnameb}}
      \end{aligned}
    \end{equation*}

  \item [\brackets{\lqchkasgn}]
    Only \emph{mutable} fields may be reassigned.

  \item [\brackets{\lqchkctxletif}]
    To type conditional structures, we first infer a type for
    the condition and then check each of the branches
    $\ssactx_1$ and $\ssactx_2$,
    assuming that the condition is true or false, respectively,
    to achieve path sensitivity.
    Each branch assigns types to the $\Phi$-variables which compose
    $\env_1$ and $\env_2$, and the propagated types for these
    variables are fresh types operating as upper bounds to
    their respective bindings in $\env_1$ and $\env_2$.

\end{description}


\subsection{Type Soundness}


We reuse the operational semantics for \ssaLang defined earlier,
and extend our type checking judgment to runtime locations
$\loc$ with the use of a \emph{heap typing} $\storety$,
mapping locations to types:
$$
\inferrule[]
{\idxMapping{\storety}{\loc} = \rtype}
{\dtypecheck{\env}{\storety}{\loc}{\rtype}}
$$
We establish type soundness results for \ssaLang in the form of a
subject reduction (preservation) and a progress theorem that connect
the static and dynamic semantics of \ssaLang.



\begin{theorem}[Subject Reduction]\label{theorem:main:subj:reduc}
If
(a)~$\dtypecheck{\env}{\storety}{\expr}{\rtype}$,
(b)~$\wfstore{\env}{\ssaRtState}{\storety}{\ssaHeap}$, and
(c)~$\stepsFour{\ssaHeap}{\expr}{\ssaHeap'}{\expr'}$,
then for some
$\rtype'$ and $\storety' \supseteq \storety$:
(i)~$\dtypecheck{\env}{\storety'}{\expr'}{\rtype'}$,
(ii)~$\semsubtype{\env}{\rtype'}{\rtype}$, and
(iii)~$\wfstore{\env}{\ssaRtState}{\storety'}{\ssaHeap'}$.
\end{theorem}

\begin{theorem}[Progress]\label{theorem:main:progress}
If
$\dtypecheck{\env}{\storety}{\expr}{\rtype}$
and
$\wfstore{\env}{\ssaRtState}{\storety}{\ssaHeap}$,
then either
$\expr$ is a value, or
there exist $\expr'$, $\ssaHeap'$ and $\storety' \supseteq \storety$
\st
$\wfstore{\env}{\ssaRtState}{\storety'}{\ssaHeap'}$ and
$\stepsFour{\ssaHeap}{\expr}{\ssaHeap'}{\expr'}$.
\end{theorem}

We defer the proofs to the supplementary material. As a corollary of
the progress theorem we get that cast operators are guaranteed to
succeed, hence they can safely be removed.

%

\begin{corollary}[Safe Casts]\label{theorem:cast:redundancy}
  Cast operations can safely be erased when compiling to executable
  code.
\end{corollary}

With the use of our Consistency Theorem (Theorem~\ref{theorem:consistency})
and extending our checking judgment for terms in \ssaLang
to runtime configurations ($\vdash \ssaRtConf$), we
can state a soundness result for \srcLang:

\begin{theorem}(\srcLang Type Safety)
If
$\tossartconfshort{\ssaenvs}{\srcRtConf}{\ssaRtConf}$ and
$\vdash \ssaRtConf$
then either $\srcRtConf$ is a terminal form, or there exists
$\srcRtConf$ \st $\steps{\srcRtConf}{\srcRtConf}$.
\end{theorem}

\section{Scaling to TypeScript}\label{sec:typescript}

TypeScript (\ts) extends JavaScript (\js) with modules,
classes and a lightweight type system that enables IDE
support for auto-completion and refactoring.
\ts\ deliberately eschews soundness~\cite{Bier14} for
backwards compatibility with existing \js\ code.
In this section, we show how to use refinement types to
\emph{regain safety}, by presenting the highlights of
\lang (and our tool \toolname), that scales
the core calculus from \S\ref{sec:language} up to
\ts\ by extending the support for
\emph{types}~(\S\ref{sec:scale-types}),
\emph{reflection}~(\S\ref{sec:scale-reflection}),
\emph{interface hierarchies}~(\S\ref{sec:scale-interfaces}),
and \emph{imperative} programming~(\S\ref{sec:scale-imp}).

\subsection{Types}\label{sec:scale-types}

First, we discuss how \toolname\ handles core \ts\ features
like object literals, interfaces and primitive types.

\mypara{Object literal types}
\ts\ supports object literals, \ie\ anonymous objects with field and method
bindings.
\toolname\ types object members in the same way as class members:
method signatures need to be explicitly provided, while field
types and mutability modifiers are inferred based on use, \eg\ in:
\begin{code}
  var point = { x: 1, y: 2 }; point.x = 2;
\end{code}
the field @x@ is updated and hence, \toolname\ infers that @x@ is mutable.

\mypara{Interfaces}
\ts\ supports named object types in the form of
interfaces, and treats them in the same way as
their \emph{structurally} equivalent class types.
For example, the interface:
%
\begin{code}
  interface PointI { number x, y; }
\end{code}
is equivalent to a class @PointC@ defined as:
%
\begin{code}
  class PointC { number x, y; }
\end{code}
In \toolname\ these two types are \emph{not}
equivalent, as objects of type @PointI@ do not
necessarily have @PointC@ as their constructor:
\begin{code}
  var pI = { x: 1, y: 2 }, pC = new PointC(1,2);
  pI instanceof PointC;   // returns false
  pC instanceof PointC;   // returns true
\end{code}
However,
$\empty \vdash  \kw{PointC}  \subt  \kw{PointI}$
\ie\ instances of the \emph{class} may be used
to implement the \emph{interface}.

\mypara{Primitive types}
We extend \toolname's support for primitive types 
to model the corresponding types in \ts.
%
\ts\ has \greenprimty{undefined} and \greenprimty{null} types to represent the
eponymous values, and treats these types as the ``bottom''
of the type hierarchy, effectively allowing those values
to inhabit \emph{every} type via subtyping.
\toolname\ also includes these two types, but \emph{does not}
treat them ``bottom'' types.
Instead \toolname handles them as distinct primitive types
inhabited solely by \vundef\ and \vnull, respectively,
that can take part in unions.
Consequently, the following code is accepted by \ts\ but
\emph{rejected} by \toolname:
\begin{code}
  var x = undefined; var y = x + 1;
\end{code}

\mypara{Unsound Features}
\ts has several unsound features deliberately chosen for backwards compatibility.
These include
(1)~treating \vundef and \vnull as inhabitants of all types,
(2)~co-variant input subtyping,
(3)~allowing unchecked overloads, and
(4)~allowing a special ``dynamic'' \tany\ type to be ascribed to any term.
\toolname\ ensures soundness by
(1)~performing checks when non-null (non-undefined) types are required
(\eg during field accesses),
(2)~using the correct variance for functions and constructors,
(3)~checking overloads via two-phase typing (\S\ref{sec:overview:overload}), and,
(4)~\emph{eliminating} the \tany\ type.

Many uses of \tany\ (indeed, \emph{all} uses, in our
benchmarks~\S\ref{sec:evaluation}) can be replaced with a
combination of union or intersection types or downcasting,
all of which are soundly checked via path-sensitive
refinements.
In future work, we wish to support the full language,
namely allow \emph{dynamically} checked uses of \tany\
by incorporating orthogonal dynamic techniques from the
contracts literature.
We envisage a \emph{dynamic cast} operation
$\texttt{cast}_{\rtype} :: \parens{\varbinding{\evar}{\tany}}
\Rightarrow \reftp{\rtype}{\vv = \evar}$.
It is straightforward to implement $\texttt{cast}_{\rtype}$
for first-order types $\rtype$ as a dynamic check
that traverses the value, testing that its components
satisfy the refinements~\cite{Seidel14}.
Wrapper-based techniques from the contracts/gradual
typing literature should then let us support higher-order types.

%
%
%

\subsection{Reflection}\label{sec:scale-reflection}

\js\ programs make extensive use of reflection via
``dynamic'' type tests.
\toolname\ statically accounts for these by encoding
type-tags in refinements.
The following tests if @x@ is a @number@
before performing an arithmetic operation on it:
\begin{code}
  var r = 1; if (typeof x === "number") r += x;
\end{code}
We account for this idiomatic use of @typeof@ by
\emph{statically} tracking the ``type'' tag of
values inside refinements using uninterpreted
functions (akin to the size of an array).
Thus, values @v@ of type
@boolean@, @number@, @string@, \etc 
are refined with the predicate
%
@ttag(v) = "boolean"@,
@ttag(v) = "number"@,
@ttag(v) = "string"@, \etc, respectively.
%
Furthermore, @typeof@ has type @(z:A) => {v:string | v = ttag(z)}@
so the output type of @typeof x@ and the path-sensitive guard
under which the assignment @r = x + 1@ occurs, ensures that at
the assignment @x@ can be statically proven to be a @number@.
The above technique coupled with two-phase typing~(\S\ref{sec:overview:overload})
allows \toolname\ to statically verify reflective, value-overloaded
functions that are ubiquitous in \ts.



\subsection{Interface Hierarchies}\label{sec:scale-interfaces}

\js\ programs frequently build up object hierarchies that
represent \emph{unions} of different kinds of values, and
then use value tests to determine which kind of value is
being operated on.
%
In \ts this is encoded by building up a hierarchy
of interfaces, and then performing \emph{downcasts} based
on \emph{value} tests\footnote{\toolname handles other type tests, \eg \kw{instanceof},
via an extension of the technique used for \kw{typeof} tests;
we omit a discussion for space.}.

\mypara{Implementing Hierarchies with bit-vectors}
The following describes a slice of the hierarchy of types used by
the \tsc compiler (\tscc) v1.0.1.0:
\begin{code}
interface Type { immutable flags: TypeFlags;
                 id             : number;
                 symbol?        : Symbol; ... }

interface ObjectType extends Type { ...  }

interface InterfaceType extends ObjectType
  { baseTypes          : ObjectType[];
    declaredProperties : Symbol[]; ... }

enum TypeFlags
  { Any      = 0x00000001, String   = 0x00000002
  , Number   = 0x00000004, Class    = 0x00000400
  , Interface= 0x00000800, Reference= 0x00001000
  , Object   = Class | Interface | Reference .. }
\end{code}


%
%
\tscc uses bit-vector valued flags to encode
membership within a particular interface type, \ie  discriminate
between the different entities. (\emph{Older} versions of \tscc
used a class-based approach, where inclusion could be
tested via \kw{instanceof} tests.)
For example, the enumeration @TypeFlags@ above maps
semantic entities to bit-vector values used as masks
that determine inclusion in a sub-interface of @Type@.
Suppose @t@ of type @Type@.
The invariant here is that if @t.flags@ masked with
@0x00000800@ is non-zero, then @t@ can be safely
treated as an @InterfaceType@ value,
or an @ObjectType@ value,
since the relevant flag emerges from the bit-wise
disjunction of the @Interface@ flag with some
other flags.

\mypara{Specifying Hierarchies with Refinements}
\toolname allows developers to \emph{create} and
\emph{use} @Type@ objects with the above
invariant by specifying a predicate @typeInv@
\footnote{Modern SMT solvers easily handle formulas
over bit-vectors, including operations that shift,
mask bit-vectors, and compare them for equality.}:
\begin{code}
 isMask<v,m,t> = mask(v,m) => impl(this,t)
 typeInv<v> = isMask<v, 0x00000001, Any>
            /\ isMask<v, 0x00000002, String>
            /\ isMask<v, 0x00003C00, ObjectType>
\end{code}
and then refining @TypeFlags@ with the predicate
\begin{code}
  type TypeFlags = {v:TypeFlags | typeInv<v>}
\end{code}
Intuitively, the refined type says that when @v@ (that is the @flags@
field)  is a bit-vector with the first position set to
@1@ the corresponding object satisfies the @Any@ interface, etc.

\mypara{Verifying Downcasts}
\toolname \emph{verifies} the code that uses ad-hoc
hierarchies such as the above by proving the \ts \emph{downcast}
operations (that allow objects to be used at particular instances) safe.
%
For example, consider the following code that \emph{tests}
if @t@ implements the @ObjectType@ interface before performing
a downcast from type @Type@ to @ObjectType@ that permits
the access of the latter's fields:
\begin{code}
function getPropertiesOfType(t: Type): Symbol[] {
  if (t.flags & TypeFlags.Object) {
    var o  = <ObjectType> t; ... } }
\end{code}


\tscc erases casts, thereby missing possible runtime errors.
The same code \emph{without} the if-test, or with a \emph{wrong}
test would pass the \tsc type checker.
%
%
\toolname, on the other hand, checks casts \emph{statically}.
In particular, @<ObjectType>t@ is treated as a call to
a function with signature:
\begin{code}
  (x:{A|impl(x,ObjectType)})=>{v:ObjectType|v=x}
\end{code}
The if-test ensures that the \emph{immutable} field @t.flags@
masked with @0x00003C00@ is non-zero,
satisfying the third line 
in the type definition of @typeInv@, which, in turn
implies that @t@ in fact implements the @ObjectType@
interface.



\subsection{Imperative Features}\label{sec:scale-imp}

\mypara{Immutability Guarantees}
Our system uses ideas from Immutability Generic Java~\cite{Zibin2007} (IGJ)
to provide statically checked immutability guarantees.
In IGJ a type reference is of the
form $\texttt{C<M,}\many{\texttt{T}}\texttt{>}$, where
\emph{immutability} argument \kw{M}
works as proxy for the immutability modifiers of the
contained fields (unless overridden).
It can be one of:
\kw{Immutable} (or \kw{IM}),
when neither this reference nor any other
reference can mutate the referenced object;
\kw{Mutable} (or \kw{MU}),
when this and potentially other references can mutate the object; and
\kw{ReadOnly} (or \kw{RO}),
when this reference cannot
mutate the object, but some other reference may.
Similar reasoning holds for method annotations.
IGJ provides \emph{deep immutability},
since a class's immutability parameter is (by default)
reused for its fields; however, this is not a firm
restriction imposed by refinement type checking.

\mypara{Arrays}
\ts's definitions file provides a detailed
specification for the \ttt{Array} interface.
We extend this definition to account for the mutating
nature of certain array operations:

%
\begin{code}
  interface Array<K extends ReadOnly,T> {
    @Mutable   pop(): T;
    @Mutable   push(x:T): number;
    @Immutable get length(): {nat|v=len(this)}
    @ReadOnly  get length(): nat;
    [...]
  }
\end{code}
%
Mutating operations (@push@, @pop@, field updates)
are only allowed on mutable arrays,
and the type of @a.length@ encodes the exact length of an
immutable array @a@, and just a natural number otherwise.
For example, assume the following code:
\begin{code}
  for(var i = 0; i < a.length; i++) {
    var x = a[i];
    [...]
  }
\end{code}
To prove the access @a[i]@ safe we need to establish
@0 <=  i@ and @i < a.length@. To guarantee that
the length of @a@ is constant, @a@ needs to be immutable,
so \tsc will flag an error unless @a: Array<IM,T>@.

\mypara{Object initialization}
%
Our formal core (\S\ref{sec:language}) treats
constructor bodies in a very limiting way:
object construction is merely an assignment
of the constructor arguments to the
fields of the newly created object.
In \toolname we relax this restriction in two ways:
(a)~We allow class and field invariants to be violated
\emph{within} the body of the constructor, but checked for at
the exit.
(b)~We permit the common idiom of certain fields being
initialized \emph{outside} the constructor, via an additional
mutability variant that encodes reference \emph{uniqueness}.
In both cases, we still restrict constructor code so that it
does not \emph{leak} references of the constructed object
(\kw{this}) or \emph{read} any of its fields, as they might still
be in an uninitialized state.

\mypara{(a) Internal Initialization: Constructors}
Type invariants do not hold while the object is being ``cooked''
within the constructor.
To safely account for this idiom, \toolname defers the checking of
class invariants (\ie the types of fields) by replacing:
(a)~occurrences of
$\dotassign{\kw{this}}{\fieldname_{\index}}{\srcExpr_{\index}}$, with
$\kw{\_}\fieldname_{\index} = \srcExpr_{\index}$, where
$\kw{\_}\fieldname_{\index}$ are \emph{local} variables, and
(b)~all return points with a call
$\funcall{\kw{ctor\_init}}{\many{\kw{\_}\fieldname_{\index}}}$, where
the signature for \kw{ctor\_init} is:
$\tparens{\varbinding{\many{\fieldname}}{\many{\rtype}}} \Rightarrow
\tvoid$.
Thus, \toolname treats field initialization in a field- and path-sensitive
way (through the usual SSA conversion), and establishes the class
invariants via a single atomic step at the constructor's exit (return).

\mypara{(b) External Initialization: Unique References}
Sometimes we want to allow immutable fields to
be initialized outside the constructor. Consider
the code (adapted from \tscc):
%
\begin{code}
  function createType(flags:TypeFlags):Type<IM>{
    var r: Type<UQ> = new Type(checker, flags);
    r.id = typeCount++;
    return r;
  }
\end{code}
Field \texttt{id} is expected to be \emph{immutable}.
However, its initialization
happens after \texttt{Type}'s constructor has returned.
Fixing the type of @r@ to @Type<IM>@ right after
construction would disallow the assignment of the @id@ field
on the following line.
So, instead, we introduce @Unique@ (or @UQ@), a new mutability type
that denotes that the current reference is the \emph{only}
reference to a specific object, and hence, allows
mutations to its fields.
When @createType@ returns, we can finally fix
the mutability parameter of @r@ to @IM@.
We could also return @Type<UQ>@,
extending the \emph{cooking} phase of the
current object and allowing further initialization by
the caller.
@UQ@ references obey stricter rules to
avoid leaking of unique references:
\begin{itemize}
  \item they cannot be re-assigned,

  \item they cannot be generally
    referenced, unless this occurs at a context
    that guarantees that no aliases will be produced, \eg
    the context of @e1@ in @e1.f = e2@, or the
    context of a returned expression, and

  \item they cannot be cast to types of a different mutability
    (\eg @<C<IM>>x@), as this would allow the same reference to
    be subsequently aliased.
\end{itemize}

More expressive initialization approaches are
discussed in \S\ref{sec:related}.

%

\section{Evaluation}\label{sec:evaluation}

To evaluate \toolname, we have used it to analyze a suite
of \js and \ts programs, to answer two questions:
(1)~What kinds of properties can be statically
    verified for real-world code?
(2)~What kinds of annotations or overhead does
    verification impose?
Next, we describe the properties, benchmarks and
discuss the results.


\mypara{Safety Properties} We verify with \toolname\ the
following:


\begin{itemize}

\item \emphbf{Property Accesses} \toolname\ verifies each field (@x.f@)
  or method lookup (@x.m(...)@) succeeds. Recall that
  @undefined@ and @null@ are not considered to inhabit
  the types to which the field or methods belong,

\item \emphbf{Array Bounds} \toolname verifies that each array
  read (@x[i]@) or write (@x[i] = e@) occurs within
  the bounds of @x@,

\item \emphbf{Overloads}  \toolname verifies that functions with
  overloaded (\ie intersection) types correctly
  implement the intersections in a path-sensitive manner
  as described in~(\S\ref{sec:overview:overload}).

\item \emphbf{Downcasts} \toolname verifies that at each \ts (down)cast
  of the form @<T> e@, the expression @e@ is indeed
  an instance of \kw{T}. This requires tracking
  program-specific invariants, \eg bit-vector
  invariants that encode hierarchies~(\S\ref{sec:scale-interfaces}).

\end{itemize}

\subsection{Benchmarks}
We took a number of existing \js\ or \ts\ programs and
ported them to \toolname.
We selected benchmarks that make heavy use of
language constructs connected to the safety
properties described above.
These include parts of the Octane
test suite, developed by Google as a \jsc\ performance
benchmark~\cite{octane} and already ported to \ts by Rastogi
\etal~\cite{Rastogi2015},
the \ts\ compiler~\cite{TypeScript},
and the D3~\cite{d3} and Transducers libraries~\cite{transducers}:



\begin{itemize}

\item \kw{navier-stokes} which
simulates two-dimensional fluid motion over time;
{\kw{richards}},
which simulates a process scheduler
with several types of processes passing information packets;
{\kw{splay}},
which implements the \emph{splay tree} data structure;
and {\kw{raytrace}},
which implements a raytracer that
renders scenes involving multiple lights and objects;
all from the Octane suite,

\item \kw{transducers} a library that implements
composable data transformations, a \jsc\ port of Hickey's
Clojure library, which is extremely dynamic in that some
functions have 12 (value-based) overloads,

\item \kw{d3-arrays} the array manipulating routines
from the D3~\cite{d3} library, which makes heavy use
of higher order functions as well as value-based overloading,

\item \kw{tsc-checker} which includes parts of the \ts
compiler (v1.0.1.0), abbreviated as \tscc. We check
15 functions from \ttt{compiler/core.ts} and 14 functions
from \ttt{compiler/checker.ts} (for which we needed to
import 779 lines of type definitions from \ttt{compiler/types.ts}).
These code segments were selected among tens of thousands of
lines of code comprising the compiler codebase, as they
exemplified interesting properties, like the bit-vector
based type hierarchies explained in \S\ref{sec:scale-interfaces}.

\end{itemize}

\newcommand\TOTANNS{529\xspace}
\newcommand\TRIVANNS{334\xspace}
\newcommand\MUTANNS{104\xspace}
\newcommand\REFANNS{91\xspace}
\newcommand\TRIVPCT{63\%\xspace}
\newcommand\MUTPCT{20\%\xspace}
\newcommand\REFPCT{17\%\xspace}

\begin{figure}[t]
\begin{tabular}{l|r|rrr|r}
\textbf{Benchmark} & \textbf{LOC} & \textbf{T} & \textbf{M} & \textbf{R} & \textbf{Time (s)}    \\  
\hline
\kw{navier-stokes} & 366          & 3          & 18         & 39         & 473                \\  
\kw{splay}         & 206          & 18         & 2          & 0          &   6                \\  
\kw{richards}      & 304          & 61         & 5          & 17         &   7                \\  
\kw{raytrace}      & 576          & 68         & 14         & 2          &  15                \\  
\kw{transducers}   & 588          & 138        & 13         & 11         &  12                \\  
\kw{d3-arrays}     & 189          & 36         & 4          & 10         &  37                \\  
\kw{tsc-checker}   & 293          & 10         & 48         & 12         &  62                \\  
\hline
\textbf{TOTAL}     & 2522         & \TRIVANNS  & \MUTANNS   & \REFANNS   &                
\end{tabular}
\vspace{0.5em}
\caption{%
\textbf{LOC} is the number of non-comment lines of source (computed via \ttt{cloc} v1.62).
The number of \rsc specifications given as JML style comments is
partitioned into
\textbf{T} trivial annotations \ie\ \tsc type signatures,
\textbf{M} mutability annotations, and
\textbf{R} refinement annotations, \ie those which actually mention invariants.
\textbf{Time} is the number of seconds taken to analyze each file.
}
\label{fig:results}
\end{figure}



\mypara{Results} Figure~\ref{fig:results} quantitatively
summarizes the results of our evaluation.
Overall, we had to add about 1 line of annotation per 5
lines of code (\TOTANNS for 2522 LOC).
The vast majority (\TRIVANNS/\TOTANNS or \TRIVPCT) of the annotations
are \emph{trivial}, \ie are \ts-like types of the form
@(x:nat) => nat@;
\MUTPCT (\MUTANNS/\TOTANNS) are trivial but have \emph{mutability}
information, and only \REFPCT (\REFANNS/\TOTANNS) mention refinements,
\ie are definitions like @type nat = {v:number|0<=v}@
or dependent signatures like @(a:T[],n:idx<a>) => T@.
These numbers show \toolname has annotation overhead
comparable with \ts, as in 83\% cases the annotations
are either identical to \ts annotations or to \ts
annotations with some mutability modifiers.
Of course, in the remaining \REFPCT cases, the signatures
are more complex than the (non-refined) \ts version.

\mypara{Code Changes}
We had to modify the source in various small (but important)
ways in order to facilitate verification. The total number of
changes is summarized in Figure~\ref{fig:changes}.
The \emph{trivial} changes include the addition of type
annotations (accounted for above), and simple transforms
to work around current limitations of our front end, \eg
converting @x++@ to @x = x + 1@.
The \emph{important} classes of changes are the following:
%
\begin{itemize}

\item \emphbf{Control-Flow:}
  Some programs had to be restructured to work around \toolname's
  currently limited support for certain control flow structures
  (\eg \kw{break}). We also modified some loops to use explicit
  termination conditions.

\item \emphbf{Classes and Constructors:}
  As \toolname does not yet support
  \emph{default} constructor arguments, we modified relevant
  @new@ calls in Octane to supply those explicitly.
  We also refactored @navier-stokes@ to use traditional
  OO style classes and constructors instead of \js records
  with function-valued fields.

\item \emphbf{Non-null Checks:}
  In @splay@ we added 5 explicit non-null checks for mutable
  objects as proving those required precise heap analysis that
  is outside \toolname's  scope.

\item \emphbf{Ghost Functions:}
  @navier-stokes@ has more than a hundred (static) array
  access sites, most of which compute indices via non-linear
  arithmetic (\ie via computed indices of the form @arr[r*s + c]@);
  SMT support for non-linear integer arithmetic is brittle (and accounts
  for the anomalous time for @navier-stokes@). We factored axioms about
  non-linear arithmetic into \emph{ghost functions} whose types were
  proven once via non-linear SMT queries, and which were then
  explicitly called at use sites to instantiate the axioms
  (thereby bypassing non-linear analysis).
  An example of such a function is:
  \begin{code}
/*@ mulThm1 :: (a:nat, b:{number | b >= 2})
            => {boolean | a + a <= a * b} */
  \end{code}
  which, when \emph{instantiated} via a call @mulThm(x, y)@
  establishes the fact that (at the call-site), @x + x <=  x * y@.
  The reported performance assumes the use of ghost functions.
  In the cases where they were not used \rsc would time out.


\end{itemize}

\newcommand\TSCIMPDIFF{TODO\xspace}
\newcommand\TSCALLDIFF{TODO\xspace}
\newcommand\ALLDIFFS{TODO\xspace}
\newcommand\IMPDIFFS{TODO\xspace}

\begin{figure}[t]
  \begin{center}
\begin{tabular}{l|r|r|r}
\textbf{Benchmark} & \textbf{LOC} & \textbf{ImpDiff} & \textbf{AllDiff}  \\
\hline
\kw{navier-stokes} & 366          &  79              & 160               \\
\kw{splay}         & 206          &  58              &  64               \\
\kw{richards}      & 304          &  52              & 108               \\
\kw{raytrace}      & 576          &  93              & 145               \\
\kw{transducers}   & 588          & 170              & 418               \\
\kw{d3-arrays}     & 189          &   8              & 110               \\
\kw{tsc-checker}   & 293          &   9              & 47                \\
\hline
\textbf{TOTAL}     & 2522         & 469              & 1052
\end{tabular}
\end{center}
\vspace{0.5em}
\caption{%
\textbf{LOC} is the number of non-comment lines of
source (computed via \ttt{cloc} v1.62).
The \emph{number of lines} at which code was changed,
which is counted as either:
\textbf{ImpDiff}: the \emph{important} changes
that require restructuring the original JavaScript
code to account for limited support for control
flow constructs, to replace records with classes
and constructors, and to add ghost functions, or,
\textbf{AllDiff}: the above plus \emph{trivial}
changes due to the addition of plain or refined
type annotations (Figure~\ref{fig:results}), and
simple edits to work around current limitations
of our front end.
}
\label{fig:changes}
\end{figure}

\subsection{Transducers (A Case Study)}\label{sec:transducers}

We now delve deeper into one of our benchmarks: the Transducers
library. At its heart this library is about reducing collections, aka
performing folds. A Transformer is anything that implements three
functions: @init@ to begin computation, @step@ to consume one element
from an input collection, and @result@ to perform any post-processing.
One could imagine rewriting @reduce@ from Figure 1 by building a
Transformer where @init@ returns @x@, @step@ invokes @f@, and @result@
is the identity.  \footnote{For simplicity of discussion we will
henceforth ignore init and initialization in general, as well as some
other details.} The Transformers provided by the library are
composable - their constructors take, as a final argument, another
Transformer, and then all calls to the outer Transformer's functions
invoke the corresponding one of the inner Transformer. This gives rise
to the concept of a Transducer, a function of type
@Transformer=> Transformer@ and this library's namesake.

The main reason this library interests us is because some of its
functions are massively overloaded. Consider, for example, the
@reduce@ function it defines. As  discussed above, @reduce@ needs a
Transformer and a collection. There are two opportunities for
overloading here. First of all, the main ways that a Transformer is
more general than a simple step function is that it can be stateful
and that it defines the @result@ post-processing step. Most of the
time the user does not need these features, in which case their
Transformer is just a wrapper around a step function. Thus for
convenience, the user is allowed to pass in either a full-fledged
Transformer or a step function which will automatically get wrapped
into one. Secondly, the collection being reduced can be a stunning
array of options: an Array, a string (\ie a collection of characters,
which are themselves just strings), an arbitrary object (\ie, in
\js, a collection of key-value pairs), an iterator (an object
that defines a @next@ function that iterates through the collection),
or an iterable (an object that defines an @iterator@ function that
returns an iterator). Each of these collections needs to be dispatched
to a type-specific reduce function that knows how to iterate over that
kind of collection. In each overload, the type of the collection must
match the type of the Transformer or step function. Thus our @reduce@
begins as shown in Figure~\ref{fig:transducers}:

\begin{figure}
\begin{code}
  /*@ ((B, A) => B,          , A[]   ) => B
      (Transformer<A,B>      , A[]   ) => B
      ((B, string) => B)     , string) => B
      (Transformer<string, B>, string) => B
      ...
  */
  function reduce(xf, coll) {
    xf = typeof xf == "function" ? wrap(xf) : xf;
    if(isString(coll)) {
      return stringReduce(xf, coll);
    } else if(isArray(coll)) {
      return arrayReduce(xf, coll);
    } else
    [...]
  }
\end{code}
\caption{Adapted sample from Transducers benchmark}
\label{fig:transducers}
\end{figure}

If you count all 5 types of collection and the 2 options for step
function vs Transformer, this function has 10 distinct overloads!
Another similar function offers 5 choices of input collection and 3
choices of output collection for a total of 15 distinct overloads.

\subsection{Unhandled Cases}
This section outlines some cases that \rsc fails to handle and
explains the reasons behind them.

\mypara{Complex Constructor Patterns}
Due to our limited internal initialization scheme,
there are certain common constructor patterns
that are not supported by \rsc.
For example, the code below:
\begin{code}
  class A<M extends RO> {
    f: nat;
    constructor() { this.setF(1); }
    setF(x: number) { this.f = x; }
  }
\end{code}

Currently, \rsc does not allow method invocations on the
object under construction in the constructor, as it cannot track
the (value of the) updates happening in the method @setF@. Note that this
case is supported by IGJ.
The relevant section in the related work
(\S\ref{sec:related})
includes approaches
that could lift this restriction.

\mypara{Recovering Unique References}
\rsc cannot recover the @Unique@ state for objects
after they have been converted to @Mutable@ (or other state),
as it lacks a fine-grained alias tracking mechanism.
Assume, for example the function @distict@ below
from the \ts compiler v1.0.1.0:
\begin{codewithnumbers}
function distinct<T>(a: T[]): T[] {
  var result: T[] = [];                       `\label{distinct:1}`
  for (var i = 0, n = a.length; i < n; i++) { `\label{distinct:2}`
    var current = a[i];                       `\label{distinct:3}`
    for (var j = 0; j < result.length; j++) { `\label{distinct:4}`
      if (result[j] === current) {            `\label{distinct:5}`
        break;                                `\label{distinct:6}`
      }                                       `\label{distinct:7}`
    }                                         `\label{distinct:7}`
    if (j === result.length) {                `\label{distinct:8}`
      result.push(current);                   `\label{distinct:9}`
    }                                         `\label{distinct:10}`
  }                                           `\label{distinct:11}`
  return result;                              `\label{distinct:12}`
}
\end{codewithnumbers}

The @results@ array is defined at line~\ref{distinct:1} so it is
initially typed as @Array<UQ,T>@.
At lines~\ref{distinct:4}--\ref{distinct:7} it is iterated over,
so in order to prove the access at line~\ref{distinct:5} safe, we
need to treat @results@ as an immutable array.
However, later on at line~\ref{distinct:9} the code pushes an
element onto @results@, an operation that requires a mutable
receiver.
Our system cannot handle the interleaving of these two kinds of
operations that (in addition) appear in a tight loop
(lines~\ref{distinct:2}--\ref{distinct:11}).
The alias tracking section in the related work
(\S\ref{sec:related}) includes
approaches that could allow support for such cases.

\mypara{Annotations per Function Overload }
A weakness of \rsc,
that stems from the use of Two-Phased Typing~\cite{rsc-ecoop15}
in handling intersection types,
is cases where type checking requires annotations
under a specific signature overload.
Consider for example the following code,
which is a variation of the @reduce@ function presented
in \S\ref{sec:overview}:
\begin{codewithnumbers}
/*@ <A>  (a:A[]`{\color{mygreen}$^+$}`,f:(A,A,idx<a>)=>A) => A`\label{code:reduce:0}`
    <A,B>(a:A[]`{\phantom{\color{mygreen}$^+$}}`,f:(B,A,idx<a>)=>B,x:B) => B`\label{code:reduce:1}`
 */                                          `\label{code:reduce:2}`
function reduce(a, f, x) {                   `\label{code:reduce:3}`
    var res, s;                              `\label{code:reduce:4}`
    if (arguments.length === 3) {            `\label{code:reduce:6}`
      res = x;                               `\label{code:reduce:7}`
      s   = 0;                               `\label{code:reduce:8}`
    } else {                                 `\label{code:reduce:9}`
      res = a[0];                            `\label{code:reduce:10}`
      s   = 1;                               `\label{code:reduce:11}`
    }                                        `\label{code:reduce:12}`
    for (var i = s; i < a.length; i++)       `\label{code:reduce:14}`
      res = f(res , a[i], i);                `\label{code:reduce:15}`
    return res;                              `\label{code:reduce:17}`
}                                            `\label{code:reduce:18}`
\end{codewithnumbers}

Checking the function body for the second overload
(line~\ref{code:reduce:1}) is problematic: without a user
type annotation on @res@, the inferred type after joining
the environments of each conditional branch will be
@res: B + (A + @{\color{mygreen}\texttt{undefined}}@)@
(as @res@ is collecting values from @x@ and
@a[0]@, at lines \ref{code:reduce:7} and \ref{code:reduce:10},
respectively),
instead of the intended @res: B@.
This causes an error when @res@ is passed to function @f@ at
line~\ref{code:reduce:15}, expected to have type @B@, which 
cannot be overcome even with refinement checking, since this 
code is no longer executed under the check on the length of 
the @arguemnts@ variable (line~\ref{code:reduce:6}).
A solution to this issue
would be for the user to annotate the type of @res@ as @B@ at its
definition at line~\ref{code:reduce:4}, but only for the specific
(second) overload. The assignment at line~\ref{code:reduce:10}
will be invalid, but this is acceptable since that branch is
provably (by the refinement checking phase~\cite{rsc-ecoop15})
dead.
This option, however,  is currently not available.

\section{Related Work}\label{sec:related}

\rsc\ is related to several distinct lines of work.

\mypara{Types for Dynamic Languages}
Original approaches incorporate \emph{flow analysis} in the
type system, using mechanisms to track aliasing and
flow-sensitive updates~\cite{Thiemann05,Drossopoulou05}.
Typed Racket's \emph{occurrence} typing
narrows the type of unions based on
control dominating type tests, and its
\emph{latent predicates} lift the results
of tests across higher order
functions~\cite{typedracket}.
DRuby~\cite{Fur09a} uses intersection
types to \emph{represent} summaries for
overloaded functions.
TeJaS~\cite{Lerner13} combines occurrence
typing with flow analysis to analyze
\js~\cite{Lerner13}.
Unlike \rsc\ none of the above
reason about relationships \emph{between}
values of multiple program variables,
which is needed to account for
value-overloading and richer
program safety properties.

\mypara{Program Logics}
At the other extreme, one can encode types as
formulas in a logic, and use SMT solvers for
all the analysis (subtyping).
DMinor explores this idea in a first-order functional
language with type tests~\cite{dminor}.
The idea can be scaled to higher-order languages
by embedding the typing relation inside the
logic~\cite{NestedPOPL12}.
DJS combines nested refinements with alias types
\cite{AliasTypes}, a restricted separation logic,
to account for aliasing and flow-sensitive heap
updates to obtain a static type system for a large
portion of \js~\cite{Chugh2012}.
DJS proved to be extremely difficult
to use.
First, the programmer had to spend a lot of effort
on manual heap related annotations; a task that
became especially cumbersome in the presence
of higher order functions.
Second, nested refinements precluded the
possibility of refinement inference,
further increasing the burden on the user.
In contrast, mutability modifiers have proven
to be lightweight~\cite{Zibin2007} and two-phase
typing lets \toolname\ use liquid refinement
inference~\cite{LiquidPLDI08},
yielding a system that is more practical for
real world programs.
%
%
\emph{Extended Static Checking}~\cite{Flanagan2002}
uses Floyd-Hoare style first-order contracts (pre-, post-conditions
and loop invariants) to generate verification conditions discharged
by an SMT solver.
Refinement types can be viewed as a generalization of
Floyd-Hoare logics that uses types to compositionally
account for polymorphic higher-order functions and
containers that are ubiquitous in modern languages
like \ts.

{X10}~\cite{Nystrom2008} is a language that extends an
object-oriented type system with \emph{constraints} on
the immutable state of classes.
Compared to {X10}, in \rsc: 
(a)~we make mutability parametric~\cite{Zibin2007},
and extend the refinement system accordingly,
(b)~we crucially obtain flow-sensitivity
via SSA transformation, and path-sensitivity
by incorporating branch conditions,
(c)~we account for reflection by
encoding tags in refinements and two-phase
typing~\cite{rsc-ecoop15}, and 
(d)~our design ensures that we can use
liquid type inference~\cite{LiquidPLDI08}
to automatically synthesize refinements.
%

%

\mypara{Analyzing \tsc}
Feldthaus \etal present a hybrid analysis to find discrepancies between
\ts\ interfaces~\cite{definitelytyped} and their
\js\ implementations~\cite{MollerOOPSLA14}, and Rastogi \etal
extend \ts with an efficient gradual type system that mitigates the
unsoundness of \ts's type system~\cite{Rastogi2015}.


\mypara{Object and Reference Immutability}
\toolname\ builds on existing methods for statically enforcing
immutability.
In particular, we build on Immutability Generic Java
(IGJ) which encodes object  and reference immutability
using Java generics~\cite{Zibin2007}.
Subsequent work extends these ideas to allow
(1) richer \emph{ownership} patterns for creating immutable
    cyclic structures~\cite{Zibin2010},
(2) \emph{unique} references, and ways to recover
    immutability after violating uniqueness,
    without requiring an alias analysis~\cite{Gordon2012}.

Reference immutability has recently been combined
with rely-guarantee logics (originally used to reason about 
thread interference), to allow refinement type reasoning.
Gordon \etal~\cite{Gordon2013} treat references 
to shared objects like threads in
rely-guarantee logics, and so multiple aliases to an 
object are allowed only if the guarantee
condition of each alias implies the rely 
condition for all other aliases. Their approach allows 
refinement types over mutable data, but resolving 
their proof obligations depends on theorem-proving,
which hinders automation.
Milit{\~a}o \etal~\cite{Militao2014} 
present Rely-Guarantee Protocols that 
can model complex aliasing interactions, and, 
compared to Gordon's work, 
allow temporary inconsistencies, 
can recover from shared state via ownership tracking, and 
resort to more lightweight proving mechanisms.

The above extensions are orthogonal to \toolname;
in the future, it would be interesting to see if
they offer practical ways for accounting for
(im)mutability in \ts\ programs.

\mypara{Object Initialization}
A key challenge in ensuring immutability is accounting
for the construction phase where fields are \emph{initialized}.
%
We limit our attention to \emph{lightweight} approaches
\ie\ those that do not require tracking aliases, capabilities
or separation logic~\cite{AliasTypes,Gardner2012}.
Haack and Poll~\cite{Haack2009} describe a flexible
initialization schema that uses secret tokens, known
only to \emph{stack-local} regions, to initialize
all members of cyclic structures.
Once initialization is complete the tokens are converted to
global ones. Their analysis is able to infer the points where
new tokens need to be introduced and committed.
The \emph{Masked Types} approach tracks, within the
type system, the set of fields that remain to be
initialized~\cite{MaskedTypes}.
{X10}'s \emph{hardhat} flow-analysis based approach to
initialization~\cite{Zibin2012} and \emph{Freedom
Before Commitment}~\cite{Summers2011} are perhaps the most
permissive of the lightweight methods, allowing,
unlike \toolname, method dispatches or field accesses
in constructors.

\section{Conclusions and Future Work}\label{sec:conclusion}

We have presented \rsc\ which brings SMT-based
modular and extensible analysis to dynamic, imperative,
class-based languages by harmoniously integrating
several techniques.
First, we restrict refinements to immutable
variables and fields (cf. {X10} \cite{Tardieu2012}).
Second, we make mutability parametric (cf. IGJ~\cite{Zibin2007})
and recover path- and flow-sensitivity via SSA.
Third, we account for reflection and value overloading
via two-phase typing~\cite{rsc-ecoop15}.
Finally, our design ensures that we can use liquid
type inference~\cite{LiquidPLDI08} to automatically
synthesize refinements.
Consequently, we have shown how \toolname\ can
verify a variety of properties with a modest
annotation overhead similar to \ts.
Finally, our experience points to several avenues
for future work, including:
(1)~more permissive but lightweight techniques for
    object initialization~\cite{Zibin2012},
(2)~automatic inference of trivial types
    via flow analysis~\cite{FirefoxTI},
(3)~verification of security properties, \eg
    access-control policies in \js\ browser
    extensions~\cite{Guha2011Oakland}.

\bibliographystyle{abbrvnat}
\bibliography{main}

\begin{thebibliography}{43}
\providecommand{\natexlab}[1]{#1}
\providecommand{\url}[1]{\texttt{#1}}
\expandafter\ifx\csname urlstyle\endcsname\relax
  \providecommand{\doi}[1]{doi: #1}\else
  \providecommand{\doi}{doi: \begingroup \urlstyle{rm}\Url}\fi

\bibitem[Anderson et~al.(2005)Anderson, Giannini, and
  Drossopoulou]{Drossopoulou05}
C.~Anderson, P.~Giannini, and S.~Drossopoulou.
\newblock {Towards Type Inference for Javascript}.
\newblock In \emph{Proceedings of the 19th European Conference on
  Object-Oriented Programming}, 2005.

\bibitem[Bierman et~al.(2010)Bierman, Gordon, Hri\c{t}cu, and
  Langworthy]{dminor}
G.~M. Bierman, A.~D. Gordon, C.~Hri\c{t}cu, and D.~Langworthy.
\newblock {Semantic Subtyping with an SMT Solver}.
\newblock In \emph{Proceedings of the 15th ACM SIGPLAN International Conference
  on Functional Programming}, 2010.

\bibitem[Bierman et~al.(2014)Bierman, Abadi, and Torgersen]{Bier14}
G.~M. Bierman, M.~Abadi, and M.~Torgersen.
\newblock Understanding typescript.
\newblock In \emph{{ECOOP} 2014 - Object-Oriented Programming - 28th European
  Conference, Uppsala, Sweden, July 28 - August 1, 2014. Proceedings}, pages
  257--281, 2014.

\bibitem[Bostock()]{d3}
M.~Bostock.
\newblock \url{http://d3js.org/}.

\bibitem[Chugh et~al.(2012{\natexlab{a}})Chugh, Herman, and Jhala]{Chugh2012}
R.~Chugh, D.~Herman, and R.~Jhala.
\newblock Dependent types for javascript.
\newblock In \emph{{Proceedings of the ACM International Conference on Object
  Oriented Programming Systems Languages and Applications}}, OOPSLA '12, pages
  587--606, New York, NY, USA, 2012{\natexlab{a}}. ACM.

\bibitem[Chugh et~al.(2012{\natexlab{b}})Chugh, Rondon, and
  Jhala]{NestedPOPL12}
R.~Chugh, P.~M. Rondon, and R.~Jhala.
\newblock {Nested Refinements: A Logic for Duck Typing}.
\newblock In \emph{Proceedings of the 39th Annual ACM SIGPLAN-SIGACT Symposium
  on Principles of Programming Languages}, 2012{\natexlab{b}}.

\bibitem[{Cognitect Labs}()]{transducers}
{Cognitect Labs}.
\newblock \url{https://github.com/cognitect-labs/transducers-js}.

\bibitem[Feldthaus and M{\o}ller(2014)]{MollerOOPSLA14}
A.~Feldthaus and A.~M{\o}ller.
\newblock {Checking Correctness of TypeScript Interfaces for JavaScript
  Libraries}.
\newblock In \emph{Proceedings of the ACM International Conference on Object
  Oriented Programming Systems Language and Applications}, 2014.

\bibitem[Flanagan et~al.(2002)Flanagan, Leino, Lillibridge, Nelson, Saxe, and
  Stata]{Flanagan2002}
C.~Flanagan, K.~R.~M. Leino, M.~Lillibridge, G.~Nelson, J.~B. Saxe, and
  R.~Stata.
\newblock Extended static checking for java.
\newblock In \emph{Proceedings of the ACM SIGPLAN 2002 Conference on
  Programming Language Design and Implementation}, PLDI '02, pages 234--245,
  New York, NY, USA, 2002. ACM.
\newblock ISBN 1-58113-463-0.

\bibitem[Furr et~al.(2009)Furr, An, Foster, and Hicks]{Fur09a}
M.~Furr, J.-h.~D. An, J.~S. Foster, and M.~Hicks.
\newblock {Static Type Inference for Ruby}.
\newblock In \emph{Proceedings of the 2009 ACM Symposium on Applied Computing},
  2009.

\bibitem[Gardner et~al.(2012)Gardner, Maffeis, and Smith]{Gardner2012}
P.~Gardner, S.~Maffeis, and G.~D. Smith.
\newblock Towards a program logic for javascript.
\newblock In \emph{POPL}, pages 31--44, 2012.

\bibitem[{Google Developers}()]{octane}
{Google Developers}.
\newblock \url{https://developers.google.com/octane/}.

\bibitem[Gordon et~al.(2012)Gordon, Parkinson, Parsons, Bromfield, and
  Duffy]{Gordon2012}
C.~S. Gordon, M.~J. Parkinson, J.~Parsons, A.~Bromfield, and J.~Duffy.
\newblock {Uniqueness and Reference Immutability for Safe Parallelism}.
\newblock In \emph{OOPSLA}, 2012.

\bibitem[Gordon et~al.(2013)Gordon, Ernst, and Grossman]{Gordon2013}
C.~S. Gordon, M.~D. Ernst, and D.~Grossman.
\newblock {Rely-guarantee References for Refinement Types over Aliased Mutable
  Data}.
\newblock In \emph{Proceedings of the 34th ACM SIGPLAN Conference on
  Programming Language Design and Implementation}, PLDI '13, pages 73--84, New
  York, NY, USA, 2013. ACM.
\newblock ISBN 978-1-4503-2014-6.

\bibitem[Guha et~al.(2011)Guha, Fredrikson, Livshits, and
  Swamy]{Guha2011Oakland}
A.~Guha, M.~Fredrikson, B.~Livshits, and N.~Swamy.
\newblock Verified security for browser extensions.
\newblock In \emph{Proceedings of the 2011 IEEE Symposium on Security and
  Privacy}, SP '11, pages 115--130, Washington, DC, USA, 2011. IEEE Computer
  Society.

\bibitem[Guo and Hackett(2012)]{FirefoxTI}
S.~Guo and B.~Hackett.
\newblock {Fast and Precise Hybrid Type Inference for JavaScript}.
\newblock In \emph{PLDI}, 2012.

\bibitem[Haack and Poll(2009)]{Haack2009}
C.~Haack and E.~Poll.
\newblock {Type-Based Object Immutability with Flexible Initialization}.
\newblock In \emph{ECOOP}, Berlin, Heidelberg, 2009.

\bibitem[Igarashi et~al.(2001)Igarashi, Pierce, and Wadler]{Igarashi2001}
A.~Igarashi, B.~C. Pierce, and P.~Wadler.
\newblock {Featherweight Java: A Minimal Core Calculus for Java and GJ}.
\newblock \emph{ACM Trans. Program. Lang. Syst.}, 23\penalty0 (3):\penalty0
  396--450, May 2001.
\newblock ISSN 0164-0925.

\bibitem[Knowles and Flanagan(2010)]{Knowles10}
K.~Knowles and C.~Flanagan.
\newblock {Hybrid Type Checking}.
\newblock \emph{ACM Trans. Program. Lang. Syst.}, 32\penalty0 (2), Feb. 2010.

\bibitem[Knowles and Flanagan(2008)]{Knowles2009}
K.~Knowles and C.~Flanagan.
\newblock Compositional reasoning and decidable checking for dependent contract
  types.
\newblock In \emph{Proceedings of the 3rd Workshop on Programming Languages
  Meets Program Verification}, PLPV '09, pages 27--38, New York, NY, USA, 2008.
  ACM.
\newblock ISBN 978-1-60558-330-3.

\bibitem[Lerner et~al.(2013)Lerner, Politz, Guha, and Krishnamurthi]{Lerner13}
B.~S. Lerner, J.~G. Politz, A.~Guha, and S.~Krishnamurthi.
\newblock {TeJaS: Retrofitting Type Systems for JavaScript}.
\newblock In \emph{Proceedings of the 9th Symposium on Dynamic Languages},
  2013.

\bibitem[{Microsoft Corporation}()]{TypeScript}
{Microsoft Corporation}.
\newblock {{TypeScript} v1.4}.
\newblock \url{http://www.typescriptlang.org/}.

\bibitem[Milit{\~a}o et~al.(2014)Milit{\~a}o, Aldrich, and Caires]{Militao2014}
F.~Milit{\~a}o, J.~Aldrich, and L.~Caires.
\newblock \emph{ECOOP 2014 -- Object-Oriented Programming: 28th European
  Conference, Uppsala, Sweden, July 28 -- August 1, 2014. Proceedings}, chapter
  Rely-Guarantee Protocols, pages 334--359.
\newblock Springer Berlin Heidelberg, Berlin, Heidelberg, 2014.

\bibitem[Nelson(1981)]{Nelson81}
G.~Nelson.
\newblock Techniques for program verification.
\newblock Technical Report CSL81-10, Xerox Palo Alto Research Center, 1981.

\bibitem[Nystrom et~al.(2008)Nystrom, Saraswat, Palsberg, and
  Grothoff]{Nystrom2008}
N.~Nystrom, V.~Saraswat, J.~Palsberg, and C.~Grothoff.
\newblock {Constrained Types for Object-oriented Languages}.
\newblock In \emph{Proceedings of the 23rd ACM SIGPLAN Conference on
  Object-oriented Programming Systems Languages and Applications}, OOPSLA '08,
  pages 457--474, New York, NY, USA, 2008. ACM.

\bibitem[Qi and Myers(2009)]{MaskedTypes}
X.~Qi and A.~C. Myers.
\newblock Masked types for sound object initialization.
\newblock In \emph{Proceedings of the 36th Annual ACM SIGPLAN-SIGACT Symposium
  on Principles of Programming Languages}, POPL '09, pages 53--65, New York,
  NY, USA, 2009. ACM.

\bibitem[Rastogi et~al.(2015)Rastogi, Swamy, Fournet, Bierman, and
  Vekris]{Rastogi2015}
A.~Rastogi, N.~Swamy, C.~Fournet, G.~Bierman, and P.~Vekris.
\newblock Safe \& efficient gradual typing for typescript.
\newblock In \emph{Proceedings of the 42Nd Annual ACM SIGPLAN-SIGACT Symposium
  on Principles of Programming Languages}, POPL '15, pages 167--180, New York,
  NY, USA, 2015. ACM.
\newblock ISBN 978-1-4503-3300-9.

\bibitem[Rondon et~al.(2008)Rondon, Kawaguci, and Jhala]{LiquidPLDI08}
P.~M. Rondon, M.~Kawaguci, and R.~Jhala.
\newblock {Liquid Types}.
\newblock In \emph{Proceedings of the ACM SIGPLAN Conference on Programming
  Language Design and Implementation}, 2008.

\bibitem[Rushby et~al.(1998)Rushby, Owre, and Shankar]{Rushby98}
J.~Rushby, S.~Owre, and N.~Shankar.
\newblock {Subtypes for Specifications: Predicate Subtyping in {PVS}}.
\newblock \emph{IEEE TSE}, 1998.

\bibitem[Seidel et~al.(2015)Seidel, Vazou, and Jhala]{Seidel14}
E.~L. Seidel, N.~Vazou, and R.~Jhala.
\newblock Type targeted testing.
\newblock In \emph{Proceedings of the 24th European Symposium on Programming on
  Programming Languages and Systems - Volume 9032}, pages 812--836, New York,
  NY, USA, 2015. Springer-Verlag New York, Inc.
\newblock ISBN 978-3-662-46668-1.

\bibitem[Smith et~al.(1999)Smith, Walker, and Morrisett]{AliasTypes}
F.~Smith, D.~Walker, and G.~Morrisett.
\newblock {Alias Types}.
\newblock In \emph{In European Symposium on Programming}, pages 366--381.
  Springer-Verlag, 1999.

\bibitem[Summers and Mueller(2011)]{Summers2011}
A.~J. Summers and P.~Mueller.
\newblock Freedom before commitment: A lightweight type system for object
  initialisation.
\newblock In \emph{Proceedings of the 2011 ACM International Conference on
  Object Oriented Programming Systems Languages and Applications}, OOPSLA '11,
  pages 1013--1032, New York, NY, USA, 2011. ACM.
\newblock ISBN 978-1-4503-0940-0.

\bibitem[Swamy et~al.(2011)Swamy, Chen, Fournet, Strub, Bhargavan, and
  Yang]{Swamy2011}
N.~Swamy, J.~Chen, C.~Fournet, P.-Y. Strub, K.~Bhargavan, and J.~Yang.
\newblock Secure distributed programming with value-dependent types.
\newblock In \emph{Proceedings of the 16th ACM SIGPLAN International Conference
  on Functional Programming}, ICFP '11, pages 266--278, New York, NY, USA,
  2011. ACM.

\bibitem[Tardieu et~al.(2012)Tardieu, Nystrom, Peshansky, and
  Saraswat]{Tardieu2012}
O.~Tardieu, N.~Nystrom, I.~Peshansky, and V.~Saraswat.
\newblock Constrained kinds.
\newblock In \emph{Proceedings of the ACM International Conference on Object
  Oriented Programming Systems Languages and Applications}, OOPSLA '12, pages
  811--830, New York, NY, USA, 2012. ACM.

\bibitem[Thiemann(2005)]{Thiemann05}
P.~Thiemann.
\newblock {Towards a Type System for Analyzing Javascript Programs}.
\newblock In \emph{Proceedings of the 14th European Conference on Programming
  Languages and Systems}, 2005.

\bibitem[Tobin-Hochstadt and Felleisen(2010)]{typedracket}
S.~Tobin-Hochstadt and M.~Felleisen.
\newblock {Logical Types for Untyped Languages}.
\newblock In \emph{Proceedings of the 15th ACM SIGPLAN International Conference
  on Functional Programming}, 2010.

\bibitem[Vazou et~al.(2014)Vazou, Seidel, Jhala, Vytiniotis, and
  Peyton-Jones]{Vazou2014}
N.~Vazou, E.~L. Seidel, R.~Jhala, D.~Vytiniotis, and S.~Peyton-Jones.
\newblock {Refinement Types for Haskell}.
\newblock In \emph{Proceedings of the 19th ACM SIGPLAN International Conference
  on Functional Programming}, 2014.

\bibitem[Vekris et~al.(2015)Vekris, Cosman, and Jhala]{rsc-ecoop15}
P.~Vekris, B.~Cosman, and R.~Jhala.
\newblock Trust, but verify: Two-phase typing for dynamic languages.
\newblock In \emph{29th European Conference on Object-Oriented Programming,
  {ECOOP} 2015, July 5-10, 2015, Prague, Czech Republic}, pages 52--75, 2015.

\bibitem[Xi and Pfenning(1999)]{XiPfenning99}
H.~Xi and F.~Pfenning.
\newblock {Dependent Types in Practical Programming}.
\newblock In \emph{Proceedings of the 26th ACM SIGPLAN-SIGACT Symposium on
  Principles of Programming Languages}, 1999.

\bibitem[Yankov()]{definitelytyped}
B.~Yankov.
\newblock \url{http://definitelytyped.org}.

\bibitem[Zibin et~al.(2007)Zibin, Potanin, Ali, Artzi, Kiezun, and
  Ernst]{Zibin2007}
Y.~Zibin, A.~Potanin, M.~Ali, S.~Artzi, A.~Kiezun, and M.~D. Ernst.
\newblock {Object and Reference Immutability Using Java Generics}.
\newblock In \emph{Proceedings of the the 6th Joint Meeting of the European
  Software Engineering Conference and the ACM SIGSOFT Symposium on The
  Foundations of Software Engineering}, 2007.

\bibitem[Zibin et~al.(2010)Zibin, Potanin, Li, Ali, and Ernst]{Zibin2010}
Y.~Zibin, A.~Potanin, P.~Li, M.~Ali, and M.~D. Ernst.
\newblock {Ownership and Immutability in Generic {Java}}.
\newblock In \emph{OOPSLA}, 2010.

\bibitem[Zibin et~al.(2012)Zibin, Cunningham, Peshansky, and
  Saraswat]{Zibin2012}
Y.~Zibin, D.~Cunningham, I.~Peshansky, and V.~Saraswat.
\newblock Object initialization in x10.
\newblock In \emph{Proceedings of the 26th European Conference on
  Object-Oriented Programming}, ECOOP'12, pages 207--231, Berlin, Heidelberg,
  2012. Springer-Verlag.

\end{thebibliography}

\appendix
\onecolumn

\section{Full System}

In this section we present the full type system for the core language of
\S3 of the main paper.

\subsection{Formal Languages}

\mypara{\srcLang}
Figure~\ref{fig:syntax-imp} shows the full syntax for the input language.
The type language is the same as described in the main paper.
The operational semantics, shown in Figure~\ref{fig:imp:opsem}, 
is borrowed from Safe TypeScript~\cite{Rastogi2015}, 
with certain simplifications since the
language we are dealing with is simpler 
than the one used there.
We use evaluation contexts $\srcEvalCtx$, with a left to right evaluation order.

\def\skipOne{\\\\}
\def\skipTwo{\\\\\\}

\begin{figure*}[ht]
\[
\begin{array}{lrcl}
\multicolumn{4}{c}{\text{Syntax}}
\skipOne
%
%
  \text{Expressions} &
  \srcExpr        & \prod  &            \srcEvar
                    \spmid              \srcVconst
                    \spmid              \srcThis
                    \spmid              \dotref{\srcExpr}{\srcFieldname}
                    \spmid              \srcMethcall{\srcExpr}{\srcMethodname}{\many{\srcExpr}}
                    \spmid              \srcNewexpr{\cname}{\many{\srcExpr}}
                    \spmid              \srcCast{\rtype}{\srcExpr}
\skipOne
%
%
  \text{Statements} &
  \srcStmt        & \prod   & 
                              \srcVarDecl{\srcEvar}{\srcExpr}                   \spmid
                              \srcDotassign{\srcExpr}{\srcFieldname}{\srcExpr}  \spmid
                              \srcAssign{\srcEvar}{\srcExpr}                    \spmid
                              \srcIte{\srcExpr}{\srcStmt}{\srcStmt}             \spmid
                              \srcSeq{\srcStmt}{\srcStmt}                       \spmid
                              \srcSkip
\skipOne
  \text{Field Decl.} &
  \srcFieldsSym  & \prod   & \cdot                                         \spmid 
                              \srcImmfieldbinding{\srcFieldname}{\rtype}       \spmid 
                              \srcMutfieldbinding{\srcFieldname}{\rtype}       \spmid
                              \concatenate{\srcFieldsSym_1}{\srcFieldsSym_2}
\skipOne
  \text{Method Body} & 
  \srcBody  & \prod   & \srcSeq{\srcStmt}{\srcReturn{\srcExpr}}
\skipOne
  \text{Expr. or Body} & 
  \srcExprOrBd & \prod & \srcExpr \spmid \srcBody
\skipOne
  \text{Method Decl.} &
  \srcMethodsSym & \prod   & \cdot                                         
                              \spmid
                              \srcMethsig{\srcMethodname}{\srcEvar}{\rtype}{\pred}{\rtype}
                              \spmid
                              \concatenate{\srcMethodsSym_1}{\srcMethodsSym_2}
\skipOne
  \text{Field Def.} &
  \srcFieldDefsSym  & \prod & \cdot                                         \spmid 
                              \srcFieldDef{\srcFieldname}{\srcVal}            \spmid 
                              \concatenate{\srcFieldDefsSym_1}{\srcFieldDefsSym_2}
\skipOne
  \text{Method Def.} &
  \srcMethodDefsSym & \prod & \cdot                                         
                              \spmid
                              \srcMethdef{\srcMethodname}{\srcEvar}{\rtype}{\pred}{\rtype}
                                         {\srcBody}
                              \spmid
                              \concatenate{\srcMethodDefsSym_1}{\srcMethodDefsSym_2}
\skipOne
  \text{Class Def.} &
  \srcCldeclname  & \prod   & \srcClassdecl{\cname}{\pred}{\rtypec}
                                           {\srcFieldsSym}{\srcMethodDefsSym}
\skipOne
  \text{Signature}  &
  \srcSignatures   & \prod   & \cdot \spmid \srcCldeclname \spmid
                                \concatenate{\srcSignatures_1}{\srcSignatures_2}
\skipOne
  \text{Program}    &
  \srcProg          & \prod   & \concatenate{\srcSignatures}{\srcBody}
\skipTwo
\multicolumn{4}{c}{\text{Runtime Configuration}}
\skipOne
  \text{Evaluation Context} & 
  \srcEvalCtx & \prod   & \srcEmpEvalctx                            \spmid          
                          \srcDotref{\srcEvalCtx}{\srcFieldname}    \spmid          
                          \srcMethcall{\srcEvalCtx}{\srcMethodname}{\many{\srcExpr}}        \spmid          
                          \srcMethcall{\srcVal}{\srcMethodname}{\many{\srcVal}, \srcEvalCtx, \many{\srcExpr}} \spmid          
                          \srcNewexpr{\cname}{\many{\srcVal}, \srcEvalCtx, \many{\srcExpr}} \spmid 
                          \srcCast{\rtype}{\srcEvalCtx}             \spmid          
                          \srcVarDecl{\srcEvar}{\srcEvalCtx}           \spmid          
\\
            & &         & 
                          \srcDotassign{\srcEvalCtx}{\srcFieldname}{\srcExpr}               \spmid          
                          \srcDotassign{\srcVal}{\srcFieldname}{\srcEvalCtx}                \spmid          
                          \srcAssign{\srcEvar}{\srcEvalCtx}         \spmid       
                          \srcIte{\srcEvalCtx}{\srcStmt}{\srcStmt}  \spmid 
                          \srcReturn{\srcEvalCtx}                   \spmid
                          \srcSeq{\srcEvalCtx}{\srcStmt}            \spmid
                          \srcSeq{\srcEvalCtx}{\srcReturn{\srcExpr}}
\skipOne
  \text{Runtime Conf.} & \srcRtConf & \prod & \pair{\srcRtState}{\srcStmt}
\skipOne
  \text{State} & \srcRtState & \prod & \quadruple{\srcSignatures}{\srcStore}{\srcStack}{\srcHeap}
\skipOne
  \text{Store}      & 
  \srcStore         & \prod   & \cdot \spmid \mapping{\srcEvar}{\srcVal} 
                                \spmid \concatenate{\srcStore_1}{\srcStore_2}
\skipOne
  \text{Value}      & 
  \srcVal           & \prod   & \srcLoc \spmid \srcVconst
\skipOne
  \text{Stack}      & 
  \srcStack         & \prod   & \cdot \spmid \stackCons{\srcStack}{\toStack{\srcStore}{\srcEvalCtx}}
\skipOne
  \text{Heap}       & 
  \srcHeap          & \prod   & \cdot \spmid \mapping{\srcLoc}{\srcObject}
                                \spmid \concatenate{\srcHeap_1}{\srcHeap_2}
\skipOne
  \text{Object}     & 
\srcObject        & \prod   & \heapObject{\srcLoc}{\srcFieldDefsSym}
                    \spmid    \heapClassObject{\cname}{\srcLoc}{\srcMethodDefsSym}
\end{array}
\]
\caption{\srcLang: syntax and runtime configuration}
\label{fig:syntax-imp}
\end{figure*}

\begin{figure*}[!t]
  \judgementHeadNameOnly{Operational Semantics for \srcLang}

  \judgementRelOnly{\stepsFour{\srcRtState}{\srcExprOrBd}{\srcRtState'}{\srcExprOrBd'}}
\begin{mathpar}
%
%
\inferrule[\srcRcEvalCtx]
{
  \stepsFour{\quadruple{\srcSignatures}{\srcStore }{\cdot}{\srcHeap }}{\srcExpr}
            {\quadruple{\srcSignatures}{\srcStore'}{\cdot}{\srcHeap'}}{\srcExpr'}
}
{
  \stepsFour{\quadruple{\srcSignatures}{\srcStore }{\srcStack}{\srcHeap }}{\idxEctx{\srcEvalCtx}{\srcExpr}}
            {\quadruple{\srcSignatures}{\srcStore'}{\srcStack}{\srcHeap'}}{\idxEctx{\srcEvalCtx}{\srcExpr'}}
}
\and
\inferrule[\srcRcVar]
{}
{
  \stepsFour{\srcRtState}{\srcEvar}
            {\srcRtState}{\idxMapping{\accessState{\srcRtState}{\srcStore}}{\srcEvar}}
}
\and
\inferrule[\srcRcDotRef]
{
  \idxMapping{\accessState{\srcRtState}{\srcHeap}}{\srcLoc} = \heapObject{\srcLoc'}{\srcFieldDefsSym}
  \\\\
  \srcFieldDef{\srcFieldname}{\srcVal} \in \srcFieldDefsSym 
}
{
  \stepsFour{\srcRtState}{\srcDotref{\srcLoc}{\srcFieldname}}
            {\srcRtState}{\srcVal}
}
\and
\inferrule[\srcRcNew]
{
  \idxMapping{\srcHeap}{\srcLoc_0} = \heapClassObject{\cname}{\srcLoc_0'}{\srcMethodDefsSym}
  \\\\
  \srcRtFields{\srcSignatures}{\cname} = \srcFieldbindings{\srcFieldname}{\rtype}
  \\\\
  \srcObject = \heapObject{\srcLoc_0}{\srcFieldDefs{\srcFieldname}{\srcVal}}
  \\\\
  \srcHeap' = \updMapping{\srcHeap}{\srcLoc}{\srcObject}
  \\
  \fresh{\srcLoc}
}
{
  \stepsFour{\quadruple{\srcSignatures}{\srcStore}{\srcStack}{\srcHeap}}
            {\srcNewexpr{\cname}{\many{\srcVal}}}
            {\quadruple{\srcSignatures}{\srcStore}{\srcStack}{\srcHeap'}}
            {\srcLoc}
}
\and
\inferrule[\srcRcCall]
{
  \srcResolveMeth{\srcHeap}{\srcLoc}{\srcMethodname}{
  \srcMethdefUnty{\srcMethodname}{\srcEvar}{\srcSeq{\srcStmt}{\srcReturn{\srcExpr}}}}
  \\\\
  \srcStore' = \concatenate{\mappings{\srcEvar}{\srcVal}}{\mapping{\srcThis}{\srcLoc}}
  \\
  \srcStack' = \concatenate{\srcStack}{\srcStore,\srcEvalCtx}
}
{
  \stepsFour{\quadruple{\srcSignatures}{\srcStore}{\srcStack}{\srcHeap}}
            {\idxEctx{\srcEvalCtx}{\srcMethcall{\srcLoc}{\srcMethodname}{\many{\srcVal}}}}
            {\quadruple{\srcSignatures}{\srcStore'}{\srcStack'}{\srcHeap}}
            {\srcSeq{\srcStmt}{\srcReturn{\srcExpr}}}
}
\and
\inferrule[\srcRcCast]
{}
{
  \stepsFour{\srcRtState}{\srcCast{\rtype}{\srcExpr}}
            {\srcRtState}{\srcExpr} 
}
\end{mathpar}
\judgementRelOnly{\stepsFour{\srcRtState}{\srcStmt}{\srcRtState'}{\srcStmt'}}
\begin{mathpar}
\inferrule[\srcRcSkip]
{}
{\stepsFour{\srcRtState}{\srcSeq{\srcSkip}{\srcStmt}}{\srcRtState}{\srcStmt}}
\and
%
\inferrule[\srcRcVarDecl]
{
  \srcStore' = \updMapping{\accessState{\srcRtState}{\srcStore}}{\srcEvar}{\srcVal}
}
{
  \stepsFour{\srcRtState}{\srcVarDecl{\srcEvar}{\srcVal}}
            {\updState{\srcRtState}{\srcStore'}}{\srcVal}
}
\and
\inferrule[\srcRcDotAsgn]
{
  \srcHeap' = \updMapping{\accessState{\srcRtState}{\srcHeap}}{\srcLoc}
                         {\updMapping{\idxMapping{\accessState{\srcRtState}{\srcHeap}}{\srcLoc}}
                                     {\srcFieldname}{\srcVal}}
}
{
  \stepsFour{\srcRtState}{\srcDotassign{\srcLoc}{\srcFieldname}{\srcVal}}
            {\updState{\srcRtState}{\srcHeap'}}
            {\srcVal}
}
\and
\inferrule[\srcRcAssign]
{
  \srcStore' = \updMapping{\accessState{\srcRtState}{\srcStore}}{\srcEvar}{\srcVal}
}
{
  \stepsFour{\srcRtState}{\srcAssign{\srcEvar}{\srcVal}}
            {\updState{\srcRtState}{\srcStore'}}
            {\srcVal}
}
\and
\inferrule[\srcRcIte]
{
  \srcVconst = \vtrue  \imp \index = 1  \\\\
  \srcVconst = \vfalse \imp \index = 2
}
{
  \stepsFour{\srcRtState}{\srcIte{\srcVconst}{\srcStmt_1}{\srcStmt_2}}
            {\srcRtState}{\srcStmt_{\index}}
}
\and
\inferrule[\srcRcRet]
{
  \accessState{\srcRtState}{\srcStack} = \concatenate{\srcStack'}{\srcStore,\srcEvalCtx}
}
{\stepsFour{\srcRtState}
           {\srcReturn{\srcVal}}
           {\updState{\srcRtState}{\srcStack', \srcStore}}
           {\idxEctx{\srcEvalCtx}{\srcVal}}
}
\end{mathpar}
\caption{Reduction Rules for \srcLang (adapted from Safe TypeScript~\cite{Rastogi2015})}
\label{fig:imp:opsem}
\end{figure*}


\mypara{\ssaLang}
Figure~\ref{fig:syntax-ssa} shows the full syntax for the SSA 
transformed language.
The reduction rules of the operational semantics 
for language \ssaLang are shown in Figure~\ref{fig:opsem}.
We use evaluation contexts $\ssaEvalCtx$, with a left to right evaluation order.

\begin{figure*}[ht]
\[
\begin{array}{lrcl}
\multicolumn{4}{c}{\text{Syntax}}
\skipOne
%
%
  \text{Expression} &
  \expr
              & \prod   &   \evar                                     \spmid
                            \vconst                                   \spmid
                            \ethis                                    \spmid
                            \dotref{\expr}{\fieldname}                \spmid
                            \methcall{\expr}{\methodname}{\many{\expr}} \spmid
                            \newexpr{\cname}{\many{\expr}}            \spmid
                            \cast{\rtype}{\expr}                      \spmid    
                            \dotassign{\expr_1}{\fieldname}{\expr_2}  \spmid
                            \ssactxidx{\ssactx}{\expr}
\skipOne
  \text{SSA context} & 
  \ssactx     & \prod   &   \hole \spmid 
                            \letin{\evar}{\expr}{\hole} \spmid
                            \letifshort
                              {\many{\phi}}
                              {\expr}
                              {\ssactx_1}{\ssactx_2}{\hole}
\skipOne
  \text{Term} & 
  \ssaterm    & \prod   &   \expr \spmid \ssactx
\skipOne
  \text{$\Phi$-Vars}    &
  \phi        & \prod   &   \ptriplet{\evar}{\evar_1}{\evar_2}

\skipOne
  \text{Field Decl.} &
  \fieldsSym  & \prod   &   \cdot                                         \spmid 
                            \immfieldbinding{\fieldname}{\rtype}       \spmid 
                            \mutfieldbinding{\fieldname}{\rtype}       \spmid
                            \concatenate{\fieldsSym_1}{\fieldsSym_2}
\skipOne
  \text{Method Decl.} &
  \methodsSym  & \prod &   \cdot                                         
                            \spmid
                            \methsig{\methodname}{\evar}{\rtype}{\pred}{\rtype}
                            \spmid
                            \concatenate{\methodsSym_1}{\methodsSym_2}
\skipOne
  \text{Field Def.} &
  \fieldDefsSym  & \prod  & \cdot                                     \spmid 
                            \fieldDef{\fieldname}{\val}             \spmid 
                            \concatenate{\fieldDefsSym_1}{\fieldDefsSym_2}
\skipOne
  \text{Method Def.} &
  \methodDefsSym  & \prod &   \cdot                                         
                            \spmid
                            \methdef{\methodname}{\evar}{\rtype}
                                    {\pred}{\rtype}{\expr}
                            \spmid
                            \concatenate{\methodDefsSym_1}{\methodDefsSym_2}
\skipOne

  \text{Class Def.} &
  \ssaCldeclSym & \prod   &   \classdecl{\cname}{\pred}{\rtypec}{\fieldsSym}{\methodDefsSym}
\skipOne
  \text{Signature}  &
  \ssaSignatures    & \prod   & \cdot \spmid \ssaCldeclSym \spmid
                                \concatenate{\ssaSignatures_1}{\ssaSignatures_2}
\skipOne
  \text{Program} &
  \ssaProg    & \prod   &   \concatenate{\ssaSignatures}{\expr}
\skipTwo
\multicolumn{4}{c}{\text{Runtime Configuration}}
\skipOne
  \text{Evaluation Context} &
  \ssaEvalCtx  & \prod &    \empevalctx             \spmid  
                            \dotref{\ssaEvalCtx}{\fieldname}  \spmid  
                            \methcall{\ssaEvalCtx}{\methodname}{\many{\expr}} \spmid  
                            \methcall{\val}{\methodname}{\many{\val}, \ssaEvalCtx,  \many{\expr}} \spmid  
                            \newexpr{\cname}{\many{\val}, \ssaEvalCtx, \many{\expr}} \spmid  
                            \cast{\rtype}{\ssaEvalCtx} \spmid  
\\
 &  & & 
                            \letin{\evar}{\ssaEvalCtx}{\expr} \spmid  
                            \dotassign{\ssaEvalCtx}{\fieldname}{\expr} \spmid  
                            \dotassign{\val}{\fieldname}{\ssaEvalCtx} \spmid  
                            \letifshort
                              {\many{\phi}}
                              {\ssaEvalCtx}{\expr}{\expr}{\expr}
\skipOne
  \text{SSA Eval. Context} &
  \ssaEvalCtxSsa  & \prod & \letin{\evar}{\ssaEvalCtx}{\hole} \spmid
                            \letifshort
                              {\many{\phi}}
                              {\ssaEvalCtx}
                              {\ssactx_1}{\ssactx_2}{\hole}
\skipOne
  \text{Term Eval. Context} & 
  \ssaEvalCtxTerm & \prod & \ssaEvalCtx \spmid \ssaEvalCtxSsa
\skipOne
  \text{Runtime Conf.} & \ssaRtConf & \prod & \pair{\ssaRtState}{\expr}
\skipOne
  \text{State} & \ssaRtState & \prod & \pair{\ssaSignatures}{\ssaHeap}
\skipOne
  \text{Heap}       & 
  \ssaHeap          & \prod   & \cdot \spmid \mapping{\loc}{\ssaObject}
                                \spmid \concatenate{\ssaHeap_1}{\ssaHeap_2}
\skipOne
  \text{Store}      & 
  \ssaStore         & \prod   & \cdot \spmid \mapping{\evar}{\val}
                                \spmid \concatenate{\ssaStore_1}{\ssaStore_2}
\skipOne
  \text{Value}      & 
  \val              & \prod   & \loc \spmid \vconst
\skipOne
  \text{Object}     & 
  \ssaObject        & \prod   & \heapObject{\loc}{\fieldDefsSym}
                      \spmid    \heapClassObject{\cname}{\loc}{\methodDefsSym}
\end{array}
\]
\caption{\ssaLang: syntax and runtime configuration}
\label{fig:syntax-ssa}
\end{figure*}

\begin{figure*}[ht]
  \judgementHead{Operational Semantics for \ssaLang}
                {\stepsFour{\ssaRtState}{\expr}{\newstuff{\ssaRtState'}}{\expr'}}
\begin{mathpar}
\inferrule[\rcevalctx]
{
  \stepsFour{\ssaRtState}{\expr}{\ssaRtState'}{\expr'}
}
{
  \stepsFour{\ssaRtState}{\idxEctx{\ssaEvalCtx}{\expr}}{\ssaRtState'}{\idxEctx{\ssaEvalCtx}{\expr'}}
}
\and
\inferrule[\rfield]
{
  \idxMapping{\accessState{\ssaRtState}{\ssaHeap}}{\loc} = 
    \heapObject{\loc'}{\fieldDefsSym}
  \\\\
  \fieldDef{\fieldname}{\val} \in \fieldDefsSym 
}
{
  \stepsFour{\ssaRtState}
        {\dotref{\loc}{\fieldname}}
        {\ssaRtState}
        {\val}
}
\and
\inferrule[\rinvoke]
{
  \resolveMeth{\ssaHeap}{\loc}{\methodname}
              {\parens{\methdef{\methodname}{\evar}{\rtypeb}{\pred}{\rtype}{\expr}}}
  \\\\
  \eval{\appsubst{\tsubsttwo{\many{\val}}{\many{\evar}}{\loc}{\ethis}}{\pred}} = \ssaTrue
%
}
{
  \stepsFour{\ssaRtState}
        {\methcall{\loc}{\methodname}{\many{\val}}}
        {\ssaRtState}
        {\appsubst{\tsubsttwo{\many{\val}}{\many{\evar}}{\loc}{\ethis}}{\expr}}
}
\and
\inferrule[\rcast]
{
  \wfconstr{\env}{\varbinding{\idxMapping{\ssaRtState}{\loc}}{\rtypeb}; \rtypeb\leq\rtype}
}
{
  \stepsFour{\ssaRtState}{\cast{\rtype}{\loc}}{\ssaRtState}{\loc}
}
\and
\inferrule[\rnew]
{
  \idxMapping{\srcHeap}{\loc_0} = \heapClassObject{\cname}{\loc_0'}{\methodDefsSym}
  \\\\
  \rtFields{\ssaSignatures}{\cname} = \fieldbindings{\fieldname}{\rtype}
  \\\\
  \ssaObject = \heapObject{\loc_0}{\fieldDefs{\fieldname}{\val}}
  \\\\
  \ssaHeap' = \updMapping{\ssaHeap}{\loc}{\ssaObject}
  \\
  \fresh{\loc}
}
{
  \stepsFour{\pair{\ssaSignatures}{\ssaHeap}}
            {\newexpr{\cname}{\many{\val}}}
            {\pair{\ssaSignatures}{\ssaHeap'}}
            {\loc}
}
\and
\inferrule[\rletin]
{}
{
  \stepsFour{\ssaRtState}{\letin{\evar}{\val}{\expr}}
            {\ssaRtState}
            {\appsubst{\esubst{\val}{\evar}}{\expr}}
}
\and
%
\inferrule[\rdotasgn]
{
  \ssaHeap' = \updMapping{\accessState{\ssaRtState}{\ssaHeap}}{\loc}
                         {\updMapping{\idxMapping{\accessState{\ssaRtState}{\ssaHeap}}{\loc}}
                                     {\fieldname}{\val}}
}
{
  \stepsFour{\ssaRtState}{\dotassign{\loc}{\fieldname}{\val}}
            {\updState{\ssaRtState}{\ssaHeap'}}{\val}
}
\and
%
\inferrule[\rletif] {
  \vconst = \ssaTrue  \imp \index = 1 
  \\
  \vconst = \ssaFalse \imp \index = 2
}
{
  \stepsFour{\ssaRtState}{\letif
      {\many{\evar}}{\many{\evar}_1}{\many{\evar}_2}
      {\vconst}{\ssactx_{1}}{\ssactx_{2}}{\expr}
    }{\ssaRtState}{\ssactxidx{\ssactx_{\index}}
     {\appsubst{\esubst{\many{\evar_{\index}}}{\many{\evar}}}{\expr}}}
}
\end{mathpar}
\caption{Reduction Rules for \ssaLang}
\label{fig:opsem}
\end{figure*}

\subsection{SSA Transformation}

Section~3 of the main paper describes the SSA transformation 
from \srcLang to \ssaLang.
This section 
provides more details and
extends the transformation to runtime 
configurations, to enable the statement and proof of 
our consistency theorem.

\subsubsection{Static Tranformation}

Figure~\ref{fig:ssa:extra} includes some additional
transformation rules that supplement the rules of Figure~3 
of the main paper.
The main program transformation judgment is:
$$
\tossaprog{\srcProg}{\ssaProg}{\ssaenvs}
$$
A global SSA enviornment $\ssaenvs$ is the result of the  
translation of the entire program $\srcProg$ to $\ssaProg$.
In particular, in a program translation tree:
\begin{itemize}
  \item 
each expression node introduces a single binding to the 
relevant SSA environment
$$
\tossaexpr{\ssaenv}{\srcExpr}{\ssaexpr}
\qquad\text{produces binding}\qquad
\mapping{\srcExpr}{\ssaenv}
$$
\item 
each statement introduces two bindings, one for the 
input environment and one for the output
(we use the notation $\stmtPre{\cdot}$ 
and $\stmtPost{\cdot}$, respectively):
$$
\tossastmt{\ssaenv_0}{\srcStmt}{\ssactx}{\ssaenv_1}
\qquad
\text{produces bindings}
\qquad
\mapping{\stmtPre{\srcStmt}}{\ssaenv_0}
\qquad
\mapping{\stmtPost{\srcStmt}}{\ssaenv_1}
$$
\end{itemize}
We assume all AST nodes are uniquely identified.

\begin{figure*}[t]
\relDescription{SSA Transformation}

\vspace{1em}
\judgementDescrForm{\tossaprog{\srcProg}{\ssaProg}{\ssaenvs}}{Program Translation}
\begin{mathpar}
\inferrule[]
{
  \tossasigs{\srcSignatures}{\ssaSignatures} 
  \quad \text{produces}\;\ssaenvs_1
  \\\\
  \tossabody{\cdot}{\srcBody}{\expr}
  \quad \text{produces}\;\ssaenvs_2
}
{
  \tossaprog{\concatenate{\srcSignatures}{\srcBody}}
            {\concatenate{\ssaSignatures}{\expr}}
            {\ssaenvs_1 \cup \ssaenvs_2}
}
\end{mathpar}
\judgementDescrForm{\tossasigs{\srcSignatures}{\ssaSignatures}}
                   {Signature Translation}
\begin{mathpar}
\inferrule[\ssasigsemp]{}{\tossasigs{\cdot}{\cdot}}
\and
\inferrule[\ssasigsbnd]
{
  \tossameth{\srcMethodDefsSym}{\methodDefsSym}
  \\
  \tossafields{\srcFieldDefsSym}{\fieldsSym}
}
{
  \tossasigs
    {\srcClassdecl{\cname}{\pred}{\rtypec}{\srcFieldsSym}{\srcMethodDefsSym}}
    {\classdecl{\cname}{\pred}{\rtypec}{\fieldsSym}{\methodDefsSym}}
}
\and
\inferrule[\ssasigscons]
{
  \tossasigs{\srcSignatures_1}{\ssaSignatures_1}  \\
  \tossasigs{\srcSignatures_2}{\ssaSignatures_2}
}
{ 
  \tossastore{\ssaenv}{\concatenate{\srcSignatures_1}{\srcSignatures_2}}
                      {\concatenate{\ssaSignatures_1}{\ssaSignatures_2}}
}
\end{mathpar}
\judgementDescrFormTwo
  {\tossaexpr{\ssaenv}{\srcExpr}{\ssaexpr}}
  {\tossastmt{\ssaenv}{\srcStmt}{\ssactx}{\ssaenv'}}
  {Expression and Statement Translations (selected)}
\begin{mathpar}
\inferrule[\ssaconst]{}{\tossaconst{\srcVconst}{\toValue{\srcVconst}}}
\and
\inferrule[\ssainvoke]
{
  \tossaexpr{\ssaenv}{\srcExpr}{\expr}
  \\
  \tossaexpr{\ssaenv}{\srcExpr_{\index}}{\many{\expr}_{\index}}
  \\\\
  \toString{\srcMethodname} = \toString{\methodname}
  \\
  \fresh{\methodname}
}
{
  \tossaexpr{\ssaenv}
            {\srcMethcall{\srcExpr}{\srcMethodname}{\many{\srcExpr}_{\index}}} 
            {\methcall{\expr}{\methodname}{\many{\expr}_{\index}}}
}
\end{mathpar}
\caption{Additional SSA Transformation Rules}
\label{fig:ssa:extra}
\end{figure*}

\subsubsection{Runtime Configuration Tranformation}

Figure~\ref{fig:ssa:rt} includes the rules for translating 
runtime configurations. The main judgment is of the form: 
$$
\tossartconf{\ssaenvs}{\srcRtState}{\srcExprOrBd}{\ssaRtState}{\expr}
$$
This assumes that the program containing expression (or body) $\srcExprOrBd$ 
was SSA-translated producing a global SSA environment $\ssaenvs$.
Rule \ssartconf translates a term $\srcExprOrBd$ 
under a state $\srcRtState$. 
This process gets factored into the translation of:
\begin{itemize} 
  \item the signatures $\accessState{\srcRtState}{\srcSignatures}$, 
    which is straight-forward (same as in static translation),
  \item the heap  $\accessState{\srcRtState}{\srcHeap}$, 
    which is described in Figure~\ref{fig:ssa:rt:two}, and 
  \item term $\srcExprOrBd$ under a local store
    $\accessState{\srcRtState}{\srcStore}$ and a stack 
    $\accessState{\srcRtState}{\srcStack}$.
\end{itemize}

The last part breaks down into rules that expose the structure of the stack. 
Rule \ssaexpstackemp translates configurations involving an empty stack, 
which are delegated to the judgment 
$\tossartterm{\srcHeap}{\ssaenvs}{\srcStore}{\srcExprOrBd}{\expr}$, 
and rule \ssastackcons separately translates the top of the stack 
and the rest of the stack frames, and then composes them into a single 
target expression.

Finally, judgments of the forms 
$\tossartstackterm{\srcHeap}{\ssaenvs}{\srcStore}{\srcStack}{\srcExprOrBd}{\expr}$ and 
$\tossartstackterm{\srcHeap}{\ssaenvs}{\srcStore}{\srcStack}{\srcEvalCtx}{\ssaEvalCtxTerm}$ 
translate expressions and statements under a local store $\srcStore$. 
The rules here are similar to their static counterparts. 
The \emph{key difference} stems from the fact that 
in \ssaLang variable are replaced with the respective values 
as soon as they come into scope. 
On the contrary, in \srcLang variables are only instantiated 
with the matching (in the store) value 
when they get into an evaluation position.
To wit, rule \ssartvarref performs the necessary substitution $\subst$
on the translated variable, which we calculate though 
the meta-function $\toSubstname$, defined as follows:
\[
  \begin{array}{rcl}
  \toSubst{\ssaenv}{\srcStore}{\srcHeap} & \defeq &
  \begin{cases}
    \braces{\esubst{\val}{\evar}\,|\,
      \mapping{\srcEvar}{\evar}   \in \ssaenv,\,
      \mapping{\srcEvar}{\srcVal} \in \srcStore,\,
      \tossaval{\srcHeap}{\srcVal}{\val}
    }
    & 
    \text{if } \envdom{\ssaenv} = \envdom{\srcStore}
  \\
  \impossible & \text{otherwise}
  \end{cases}
%
\end{array}
\]

\begin{figure*}[t]
\relDescription{SSA Transformation for Runtime Configurations}

\vspace{1em}
\judgementDescrFormTwo
  {\tossartconf{\ssaenvs}{\srcRtState}{\srcExprOrBd}{\ssaRtState}{\expr}}
  {\tossartconf{\ssaenvs}{\srcRtState}{\srcStmt}{\ssaRtState}{\ssactx}}
  {Runtime Configuration Translation}
\begin{mathpar}
\inferrule[\ssartconf]
{
  \tossartsigs{\ssaenvs}{\accessState{\srcRtState}{\srcSignatures}}{\ssaSignatures}
  \\
  \tossaheap{\srcRtState}{\accessState{\srcRtState}{\srcHeap}}{\ssaHeap}
  \\
  \tossartstackterm{\accessState{\srcRtState}{\srcHeap}}{\ssaenvs}
    {\accessState{\srcRtState}{\srcStore}}
    {\accessState{\srcRtState}{\srcStack}}
    {\srcExprOrBd}{\expr}
}
{
  \tossartconf{\ssaenvs}
              {\srcRtState}{\srcExprOrBd}
              {\pair{\ssaSignatures}{\ssaHeap}}{\expr}
}
\and
\inferrule[\ssastmtrtconf]
{
  \tossartsigs{\ssaenvs}{\accessState{\srcRtState}{\srcSignatures}}{\ssaSignatures}
  \\
  \tossaheap{\srcRtState}{\accessState{\srcRtState}{\srcHeap}}{\ssaHeap}
  \\
  \tossartstackterm{\accessState{\srcRtState}{\srcHeap}}{\ssaenvs}
    {\accessState{\srcRtState}{\srcStore}}
    {\accessState{\srcRtState}{\srcStack}}
    {\srcStmt}{\ssactx}
}
{
  \tossartconf{\ssaenvs}
              {\srcRtState}{\srcStmt}
              {\pair{\ssaSignatures}{\ssaHeap}}{\ssactx}
}
\end{mathpar}
\\

\judgementDescrFormTwo
  {\tossartstackterm{\srcHeap}{\ssaenvs}{\srcStore}{\srcStack}{\srcExprOrBd}{\expr}}
  {\tossartstackterm{\srcHeap}{\ssaenvs}{\srcStore}{\srcStack}{\srcEvalCtx}{\ssaEvalCtxTerm}}
  {Runtime Stack Translation}
\begin{mathpar}
%
\inferrule[\ssaexpstackemp]
{
  %
  \tossartterm{\srcHeap}{\ssaenvs}{\srcStore}{\srcExprOrBd}{\expr}
}
{
  \tossartstackterm{\srcHeap}{\ssaenvs}
                   {\srcStore}{\cdot}
                   {\srcExprOrBd}{\expr}
}
\and
%
\inferrule[\ssaectxstackemp]
{
  \tossaevalctx{\srcHeap}{\ssaenvs}{\srcStore}{\srcEvalCtx}{\ssaEvalCtxTerm}
}
{
  \tossartstackterm{\srcHeap}{\ssaenvs}
              {\srcStore}{\cdot}
              {\srcEvalCtx}{\ssaEvalCtxTerm}
}
\and
\inferrule[\ssastackcons]
{
  \tossartstackterm{\srcHeap}{\ssaenvs}{\srcStore_0}{\cdot}{\srcExprOrBd}{\expr_0}
  \\
  \tossartstackterm{\srcHeap}{\ssaenvs}
              {\srcStore}{\srcStack}{\srcEvalCtx}
              {\ssaEvalCtx}
}
{
  \tossartstackterm{\srcHeap}{\ssaenvs}
              {\srcStore_0} 
              {\parens{\stackCons{\srcStack}{\toStack{\srcStore}{\srcEvalCtx}}}}
              {\srcExprOrBd}
              {\idxEctx{\ssaEvalCtx}{\expr_0}}
}
\and
\inferrule[\ssaectxstackcons]
{
  \tossartstackterm{\srcHeap}{\ssaenvs}
              {\srcStore_0}{\cdot}{\srcEvalCtx_0}
              {\ssaEvalCtxTerm_0}
  \\
  \tossartstackterm{\srcHeap}{\ssaenvs}
              {\srcStore}{\srcStack}{\srcEvalCtx}
              {\ssaEvalCtx}
}
{
  \tossartstackterm
    {\srcHeap}{\ssaenvs}{\srcStore_0}
    {\parens{\stackCons{\srcStack}{\toStack{\srcStore}{\srcEvalCtx}}}}
    {\srcEvalCtx_0}
    {\idxEctx{\ssaEvalCtx}{\ssaEvalCtxTerm_0}}
}
\\

\end{mathpar}
\judgementDescrFormThree
  {\tossartsigs{\ssaenvs}{\srcMethodDefsSym}{\methodDefsSym}}
  {\tossartterm{\srcHeap}{\ssaenvs}{\srcStore}{\srcExprOrBd}{\expr}}
  {\tossartterm{\srcHeap}{\ssaenvs}{\srcStore}{\srcStmt}{\ssactx}}
  {Runtime Term Translation (selected rules)}
\begin{mathpar}
\inferrule[\ssartmethdecl]
{
  \tossartterm{\cdot}{\ssaenvs}{\cdot}{\srcBody}{\expr}
}
{
  \tossartsigs{\ssaenvs}
    {\srcMethdefUnty{\srcMethodname}{\srcEvar}{\srcBody}}
    {\methdefUnty{\methodname}{\evar}{\expr}}
}
\and
\inferrule[\ssartval]
{
  \tossaval{\srcHeap}{\srcVal}{\val}
}
{
  \tossartterm{\srcHeap}{\ssaenvs}{\srcStore}{\srcVal}{\val}
}
\and
\inferrule[\ssartvarref]
{
  \tossaexpr{\idxMapping{\ssaenvs}{\srcEvar}}{\srcEvar}{\evar}
  \\\\
  \subst = \toSubst{\idxMapping{\ssaenvs}{\srcEvar}}{\srcStore}{\srcHeap}
}
{
  \tossartterm{\srcHeap}{\ssaenvs}{\srcStore}{\srcEvar}{\appsubst{\subst}{\evar}}
}
\and
\inferrule[\ssartmethcall]
{
  \tossartterm{\srcHeap}{\ssaenvs}{\srcStore}{\srcExpr}{\expr}
  \\
  \tossartterm{\srcHeap}{\ssaenvs}{\srcStore}{\many{\srcExpr}}{\many{\expr}}
  \\\\
  \toString{\srcMethodname} = \toString{\methodname}
}
{
  \tossartterm{\srcHeap}{\ssaenvs}{\srcStore}
              {\srcMethcall{\srcExpr}{\srcMethodname}{\many{\srcExpr}}}
              {\methcall{\expr}{\methodname}{\many{\expr}}}
}
\and
\inferrule[\ssartmethbody]
{
  \tossartterm{\srcHeap}{\ssaenvs}{\srcStore}{\srcStmt}{\ssactx}
  \\
  \ssaenvs' = \updMapping{\ssaenvs}{\srcExpr}{\idxMappingPost{\ssaenvs}{\srcStmt}}
  \\
  \tossartterm{\srcHeap}{\ssaenvs'}{\srcStore}{\srcExpr}{\expr}
}
{
  \tossartterm{\srcHeap}{\ssaenvs}{\srcStore}
    {\srcSeq{\srcStmt}{\srcReturn{\srcExpr}}}
    {\ssactxidx{\ssactx}{\expr}}
}
\and
\inferrule[\ssartvardecl]
{
  \mapping{\srcEvar}{\evar} \in
    \idxMappingPost{\ssaenvs}{\srcVarDecl{\srcEvar}{\srcExpr}}
  \\
  \tossartterm{\srcHeap}{\ssaenvs}{\srcStore}{\srcExpr}{\expr}
}
{
  \tossartterm{\srcHeap}{\ssaenvs}{\srcStore}
              {\srcVarDecl{\srcEvar}{\srcExpr}}
              {\letin{\evar}{\ssaexpr}{\hole}}
}
\and
\inferrule[\ssartite]
{
  \tossartterm{\srcHeap}{\ssaenvs}{\srcStore}
    {\srcExpr}{\expr}
  \\
  \tossartterm{\srcHeap}{\ssaenvs}{\srcStore}
    {\srcStmt_1}{\ssactx_1}
  \\
  \tossartterm{\srcHeap}{\ssaenvs}{\srcStore}
    {\srcStmt_2}{\ssactx_2}
  \\\\
  \ptriplet{\many{\srcEvar}}{\many{\evar}_1}{\many{\evar}_2} 
    = \envDiff{\idxMappingPost{\ssaenvs}{\srcStmt_1}} 
              {\idxMappingPost{\ssaenvs}{\srcStmt_2}} 
  \\
  \many{\evar} = \idxMapping
    {\idxMappingPost{\ssaenvs}{\srcIte{\srcExpr}{\srcStmt_1}{\srcStmt_2}}}
    {\many{\srcEvar}}
}
{
  \tossartterm{\srcHeap}{\ssaenvs}{\srcStore}
    {\srcIte{\srcExpr}{\srcStmt_1}{\srcStmt_2}}
    {\letif{\many{\evar}}{\many{\evar}_1}{\many{\evar}_2}
           {\ssaexpr}{\ssactx_1}{\ssactx_2}{\hole}
    }
}
\and
\inferrule[\ssartassign]
{
  \mapping{\srcEvar}{\evar} \in
    \idxMappingPost{\ssaenvs}{\srcAssign{\srcEvar}{\srcExpr}}
  \\
  \tossartterm{\srcHeap}{\ssaenvs}{\srcStore}{\srcExpr}{\expr}
}
{
  \tossartterm{\srcHeap}{\ssaenvs}{\srcStore}
      {\srcAssign{\srcEvar}{\srcExpr}}
      {\letin{\evar}{\ssaexpr}{\hole}}
}
\end{mathpar}
\\

\judgementDescrForm
  {\tossaevalctx{\srcHeap}{\ssaenvs}{\srcStore}{\srcEvalCtx}{\ssaEvalCtxTerm}}
  {Evaluation Context Translation (selected rules)}
\begin{mathpar}
\inferrule[]{}
{\tossaevalctx{\srcHeap}{\ssaenvs}{\srcStore}{\srcEmpEvalctx}{\empevalctx}}  
\and
\inferrule[]
{
  \tossaevalctx{\srcHeap}{\ssaenvs}{\srcStore}{\srcEvalCtx}{\ssaEvalCtx}
  \\
  \fresh{\fieldname}
  \\\\
  \toString{\srcFieldname} = \toString{\fieldname}
}
{
  \tossaevalctx{\srcHeap}{\ssaenvs}{\srcStore}
               {\srcDotref{\srcEvalCtx}{\srcFieldname}}
               {\dotref{\ssaEvalCtx}{\fieldname}}
}  
\and
\inferrule[]
{
  \tossaevalctx{\srcHeap}{\ssaenvs}{\srcStore}{\srcEvalCtx}{\ssaEvalCtx}
  \\
  \fresh{\methodname}
  \\\\
  \toString{\srcMethodname} = \toString{\methodname}
  \\
  \tossartstackterm{\srcHeap}{\ssaenvs}{\srcStore}{\cdot}{\many{\srcExpr}}{\many{\expr}}
}
{
  \tossaevalctx{\srcHeap}{\ssaenvs}{\srcStore}
             {\srcMethcall{\srcEvalCtx}{\srcMethodname}{\many{\srcExpr}}}
             {\methcall{\ssaEvalCtx}{\methodname}{\many{\expr}}}
}  
\and
\inferrule[]
{
  \tossartstackterm{\srcHeap}{\ssaenvs}{\srcStore}{\cdot}{\srcEvar}{\evar}
  \\
  \tossaevalctx{\srcHeap}{\ssaenvs}{\srcStore}{\srcEvalCtx}{\ssaEvalCtx}
}
{
  \tossaevalctx{\srcHeap}{\ssaenvs}{\srcStore}
             {\srcVarDecl{\srcEvar}{\srcEvalCtx}}
             {\letin{\evar}{\ssaEvalCtx}{\hole}}
}  
\and
\inferrule[]
{  
  \tossaevalctx{\srcHeap}{\ssaenvs}{\srcStore}{\srcEvalCtx}{\ssaEvalCtxSsa}
  \\
  \tossaevalctx{\srcHeap}{\ssaenvs}{\srcStore}{\srcStmt}{\ssactx}
}
{
  \tossaevalctx{\srcHeap}{\ssaenvs}
             {\srcStore}
             {\srcSeq{\srcEvalCtx}{\srcStmt}}
             {\ssactxidx{\ssaEvalCtxSsa}{\ssactx}}
}
\end{mathpar}
%
\caption{SSA Transformation Rules for Runtime Configurations}
\label{fig:ssa:rt}
\end{figure*}

\begin{figure*}[t]
\judgementDescrFormTwo
  {\tossaheap{\srcRtState}{\srcHeap}{\ssaHeap}}
  {\tossaval{\srcHeap}{\srcVal}{\val}}
  {Heap Translation}
\begin{mathpar}
\inferrule[\ssaheapemp]{}{\tossaheap{\srcRtState}{\cdot}{\cdot}}
\and
\inferrule[\ssaheapbnd]
{
  \tossaobj{\accessState{\srcRtState}{\srcHeap}}{\srcObject}{\ssaObject}
  \\
  \fresh{\loc}
}
{
  \tossaheap{\srcRtState}
            {\parens{\mapping{\srcLoc}{\srcObject}}}
            {\parens{\mapping{\loc}{\ssaObject}}}
}
\and
\inferrule[\ssaheapcons]
{
  \tossaheap{\srcRtState}{\srcHeap_1}{\ssaHeap_1}  \\
  \tossaheap{\srcRtState}{\srcHeap_2}{\ssaHeap_2}
}
{
  \tossaheap{\srcRtState}
            {\parens{\concatenate{\srcHeap_1}{\srcHeap_2}}}
            {\concatenate{\ssaHeap_1}{\ssaHeap_2}}
}
\and
\inferrule[\ssaloc]
{
  \mapping{\srcLoc}{\srcObject} \in \srcHeap
  \\ 
  \tossaheap{\srcHeap}
    {\parens{\mapping{\srcLoc}{\srcObject}}}
    {\parens{\mapping{\loc}{\ssaObject}}}
}
{
  \tossaval{\srcHeap}{\srcLoc}{\loc}
}
\and
\inferrule[\ssaconst]
{
  \toValue{\srcVconst} = \toValue{\vconst}
  \\
  \fresh{\vconst}
}
{
  \tossaval{\srcHeap}{\srcVconst}{\vconst}
}
\end{mathpar}
\judgementDescrForm{\tossaobj{\srcHeap}{\srcObject}{\ssaObject}}{Heap Object Translation}
\begin{mathpar}
\inferrule[]
{
  \tossaval{\srcHeap}{\srcLoc}{\loc}
  \\
  \tossafield{\srcHeap}{\srcFieldsSym}{\fieldDefsSym}
}
{
  \tossaobj{\srcHeap}
    {\heapObject{\srcLoc}{\srcFieldDefsSym}}
    {\heapObject{\loc}{\fieldDefsSym}}
}
\and
\inferrule[]
{
  \tossaval{\srcHeap}{\srcLoc}{\loc}
  \\
  \tossagen{\srcMethodDefsSym}{\methodDefsSym}
}
{
  \tossaobj{\srcHeap}
  {\heapClassObject{\cname}{\srcLoc}{\srcMethodDefsSym}}
  {\heapClassObject{\cname}{\loc}{\methodDefsSym}}
}
\end{mathpar}
\caption{SSA Transformation Rules for Heaps and Objects}
\label{fig:ssa:rt:two}
\end{figure*}

\subsection{Object Constraint System}

Our system leverages the idea introduced in the formall core of
X10~\cite{Nystrom2008} to extend a base constraint system \constraintsystem
with a larger constraint system \objectsystemidx{\constraintsystem}, built on
top of \constraintsystem.
The original system \constraintsystem comprises formulas
taken from a decidable SMT logic~\cite{Nelson81}, including, for example, linear
arithmetic constraints and uninterpreted predicates.
The Object Constraint System \objectsystemidx{\constraintsystem} introduces the
constraints:
\begin{itemize}
  \item \isclass{\cname}, which it true for all classes
    \cname defined in the program;
  \item \hasimm{\evar}{\srcFieldname}, to denote that the
    \emph{immutable} field \srcFieldname is accessible from
    variable \evar;

  \item \hasmut{\evar}{\srcFieldname}, to denote that the
    \emph{mutable} field \srcFieldname is accessible from
    variable \evar; and
  \item \fieldsOfVar{\evar} = \srcFieldsSym,
    to expose all fields available to \evar.
\end{itemize}

Figure~\ref{fig:structural:constraints} shows the 
constraint system as ported from CFG~\cite{Nystrom2008}. 
We refer the reader to that work for details.
The main differences are syntactic changes to
account for our notion of \emph{strengthening}.
Also the \scfield rule accounts
now for both immutable and mutable fields.
The main judgment here is of the form: 
$$
\sconstraint{\ssaSignatures}{\env}{\pred}
$$
where \ssaSignatures is the set of classes defined 
in the program.
Substitutions and strengthening operations on field 
declarations are performed on the types of the declared
fields (\eg \scfieldi, \scfieldc).

\begin{figure*}[ht]
\judgementHead{Structural Constraints}{\sconstraint{\ssaSignatures}{\env}{\pred}}
\begin{mathpar}
%
\inferrule[\scclass]
{
  \classdecl{\cname}{\pred}{\rtypec}{\fieldsSym}{\methodDefsSym}
  \in \ssaSignatures
}
{
  \sconstraint{\ssaSignatures}{\env}{\isclass{\cname}}
}
\and
%
\inferrule[\scinv]
{
  \sconstraint{\ssaSignatures}{\env}{\mconcatenate
    {\varbinding{\evar}{\cname}}
    {\isclass{\cname}}
  }
}
{
  \sconstraint{\ssaSignatures}{\env}{\classinv{\cname}{\evar}}
}
\and
%
\inferrule[\scfield]
{
  \sconstraint{\ssaSignatures}{\env}
              {\fieldsOfVar{\evar} = \mconcatenate
                {\immfieldbinding{\many{\fieldname_{\index }}}{\many{\rtype_{\index} }}}
                {\mutfieldbinding{\many{\fieldnameb_{\index}}}{\many{\rtypeb_{\index}}}}
              }
}
{
  \sconstraint{\ssaSignatures}{\env}
    {\hasimm{\evar}{\varbinding{\fieldname_{\index}}{\rtype_{\index}}}}
  \\\\
  \sconstraint{\ssaSignatures}{\env}
    {\hasmut{\evar}{\varbinding{\fieldnameb_{\index}}{\rtypeb_{\index}}}}
}
\and
%
\inferrule[\scobject]
{}
{
  \sconstraint{\ssaSignatures}{\envbinding{\evar}{\texttt{Object}}}{\fieldsOfVar{\evar}{\emp}}
}
\and
%
\inferrule[\scfieldi]
{
  \sconstraint{\ssaSignatures}{\envext{\env}{\evar}{\cnameb}}
              {\fieldsOfVar{\evar} = \fieldsSym
              }
  \\\\
  \classdecl{\cname}{\pred}{\rtypec}{\fieldsSym'}{\methodDefsSym}
  \in \ssaSignatures
}
{
  \sconstraint{\ssaSignatures}{\envext{\env}{\evar}{\cnameb}}
    {\fieldsOfVar{\evar} = \mconcatenate
      {\fieldsSym}
      {\appsubst{\tsubst{\evar}{\this}}{\fieldsSym'}}
      }
}
\and
\inferrule[\scfieldc]
{
  \sconstraint{\ssaSignatures}{\envext{\env}{\evar}{\cname}}
              {\fieldsOfVar{\evar} = \fieldsSym}
}
{
  \sconstraint{\ssaSignatures}{\envext{\env}{\evar}{\reftp{\cname}{\pred}}}
              {\fieldsOfVar{\evar} = \strengthen{\fieldsSym}{\appsubst{\pred}{\tsubst{\evar}{\vv}}}}
}
\and
%
\inferrule[\scmethb]
{
  \sconstraint{\ssaSignatures}{\env}{\isclass{\cname}}
  \\
  \subst = \tsubst{\evar}{\this}
  \\\\
  \methdef{\methodname}{\evar}{\rtype}{\pred}{\rtype}{\expr}  \in \cname
}
{
  \sconstraint{\ssaSignatures}{\envext{\env}{\evar}{\cname}}{
    \has{\evar}{\parens{
      \methdef{\methodname}{\evar}
              {\appsubst{\subst}{\rtype}}
              {\appsubst{\subst}{\pred}}
              {\appsubst{\subst}{\rtype}}
              {\expr}
      }}
  }
}
\and
%
\inferrule[\scmethi]
{
  \sconstraint{\ssaSignatures}{\envext{\env}{\evar}{\cnameb}}{\has{\evar}{\parens{\methdef{\methodname}{\evar}{\rtype}{\pred}{\rtype}{\expr}}}}
  \\\\
  \classdecl{\cname}{\pred}{\cnameb}{\fieldsSym}{\methodDefsSym}
  \in \ssaSignatures
  \\
  \methodname \notin \methodDefsSym
}
{
  \sconstraint{\ssaSignatures}{\envext{\env}{\evar}{\cname}}{\has{\evar}{\parens{\methdef{\methodname}{\evar}{\rtype}{\pred}{\rtype}{\expr}}}}
}
\and
%
\inferrule[\scmethc]
{
  \sconstraint{\ssaSignatures}{\envext{\env}{\evar}{\cname}}{\has{\evar}{\parens{\methdef{\methodname}{\evar}{\rtype}{\pred_0}{\rtype}{\expr}}}}
}
{
  \sconstraint{\ssaSignatures}{\envext{\env}{\evar}{\reftp{\cname}{\pred}}}
              {\has{\evar}{\parens{\methdef{\methodname}{\evar}{\rtype}{\pred_0}
              {\strengthen{\rtype}{\appsubst{\tsubst{\evar}{\this}}{\pred}}}{\expr}}}}
}
\end{mathpar}
\caption{Structural Constraints (adapted from \cite{Nystrom2008})}
\label{fig:structural:constraints}
\end{figure*}

\subsection{Well-formedness Constraints}

The well-formedness rules for predicates, terms, types and heaps can
be found in Figure~\ref{fig:wf}. The majority of these rules are
routine.

The judgment for term well-formedness assigns a \emph{sort}
to each term \cterm, which can be thought of as a base type. The
judgment $\wfpredspecial{\env}{\qname}{\many{\cterm}}$
is used as a shortcut for any further constraints that the $\fname$
operator might impose on its arguments  $\many{\cterm}$.
For example if $\fname$ is the equality operator then the two
arguments are required to have types that are related via subtyping,
\ie if $\envbinding{\cterm_1}{\btype_1}$ and
$\envbinding{\cterm_2}{\btype_2}$, it needs to be the case that
$\btype_1 \subt \btype_2$ or
$\btype_2 \subt \btype_1$.

Type well-formedness is typical among similar refinement
types~\cite{Knowles2009}.

\begin{figure*}[ht]

\judgementHead{Well-Formed Predicates}{\wfpred{\env}{\pred}}
\begin{mathpar}
\inferrule[\wppand]
{
  \wfpred{\env}{\pred_1}
  \\
  \wfpred{\env}{\pred_2}
}
{
  \wfpred{\env}{\pand{\pred_1}{\pred_2}}
}
\and
\inferrule[\wppnot]
{
  \wfpred{\env}{\pred}
}
{
  \wfpred{\env}{\pnot{\pred}}
}
\and
\inferrule[\wpterm]
{
  \wfterm{\env}{\cterm}{\tbool}
}
{
  \wfpred{\env}{\cterm}
}
\end{mathpar}
\judgementHead{Well-Formed Terms}{\wfterm{\env}{\cterm}{\btype}}
\begin{mathpar}
\inferrule[\wfvar]
{
  \envbinding{\evar}{\rtype} \in \env
}
{
  \wfterm{\env}{\evar}{\basetype{\rtype}}
}
\and
\inferrule[\wfconst]
{}
{
  \wfterm{\env}{\vconst}{\basetype{\tconst{\vconst}}}
}
\and
\inferrule[\wffield]
{
  \wfterm{\env}{\cterm}{\btype}
  \\
  \typecheck{\envext{\env}{\evar}{\btype}}
            {\hasimm{\evar}{\fieldname_{\index}}}
            {\rtype_{\index}}
}
{
  \wfterm{\env}{\dotref{\cterm}{\fieldname_{\index}}}{\basetype{\rtype_{\index}}}
}
\and
\inferrule[\wffun]
{
  \wfterm{\env}{\fname}{\arrtype{\many{\btype}}{\btype'}}
  \\
  \wfpredspecial{\env}{\qname}{\many{\cterm}}
}
{
  \wfterm{\env}{\funcall{\fname}{\many{\cterm}}}{\btype'}
}
\end{mathpar}
\judgementHead{Well-Formed Types}{\wftype{\env}{\rtype}}
\begin{mathpar}
\inferrule[\wtbase]
{
  \wfpred{\envext{\env}{\vv}{\btype}}{\pred}
}
{
  \wftype{\env}{\reftp{\btype}{\pred}}
}
\and
\inferrule[\wtexist]
{
  \wftype{\env}{\rtype_1}
  \\
  \wftype{\envext{\env}{\evar}{\rtype_1}}{\rtype_2}
}
{
  \wftype{\env}{\texist{\evar}{\rtype_1}{\rtype_2}}
}
\end{mathpar}
%
%
%
%
\judgementHead{Well-Formed Heaps}{\newwfstore{\storety}{\ssaHeap}}
\begin{mathpar}
\inferrule[\hemp]{}
{
  \newwfstore{\storety}{\cdot}
}
\and
%
\inferrule[\hbndinst]
{
  \ssaObject \defeq \heapObject{\loc'}{\fieldDefsSym}
  \\
  \fieldDefsSym \defeq \mconcatenate
    {\immFieldDefs{\fieldname }{\val}_{\immindex}}
    {\mutFieldDefs{\fieldnameb}{\val}_{\mutindex}}
  \\
  \basetype{\idxMapping{\storety}{\loc}} = \cname
  \\\\
  \wfconstr{\envext{\env}{\evarb}{\cname}}
           {\fieldsOfVar{\evarb} = 
              \mconcatenate{\immfieldbindings{\fieldname}{\rtypec}}
                           {\mutfieldbindings{\fieldnameb}{\rtyped}}
           }
  \\
  \rttypecheckval{\storety}{\many{\val}_{\immindex}}{\many{\rtype}_{\immindex}}
  \\
  \rttypecheckval{\storety}{\many{\val}_{\mutindex}}{\many{\rtype}_{\mutindex}}
  \\\\
  \wfconstr{
      \envextex
        {\envext{\env}{\evarb}{\cname}}
        {\many{\evarb}_{\immindex}}
      {\singleton{\many{\rtype}_{\immindex}}{\tdotref{\evarb}{\many{\fieldname}}}}
    }
    {
      \mconcatenate{
        \mconcatenate{\many{\rtype}_{\immindex} \subt \many{\rtypec}}
        {\many{\rtype}_{\mutindex} \subt \many{\rtyped}}
      }
      {\classinv{\cname}{\evarb}}
    }
}
{
  \newwfstore{\storety}{\mapping{\loc}{\ssaObject}}
}
\and
%
%
\inferrule[\hcons]
{
  \newwfstore{\storety}{\ssaHeap_1}
  \\
  \newwfstore{\storety}{\ssaHeap_2}
}
{
  \newwfstore{\storety}{\concatenate{\ssaHeap_1}{\ssaHeap_2}}
}
\end{mathpar}
%
\caption{Well-Formedness Rules}
\label{fig:wf}
\end{figure*}

\subsection{Subtyping}

Figure~\ref{fig:subtyping} presents the full set of sybtyping rules,
which borrows ideas from similar systems~\cite{Knowles2009,LiquidPLDI08}.

\begin{figure*}[ht]

\judgementHead{Subtyping}{\issubtype{\env}{\rtype}{\rtype'}}
\begin{mathpar}
\inferrule[{\subrefl}]
{}
{
  \issubtype{\env}{\rtype}{\rtype}
}
\and
\inferrule[{\subtrans}]
{
  \issubtype{\env}{\rtype_1}{\rtype_2}
  \\
  \issubtype{\env}{\rtype_2}{\rtype_3}
}
{
  \issubtype{\env}{\rtype_1}{\rtype_3}
}
\and
\inferrule[{\subext}]
{
  \classdecl{\cname}{\pred}{\cnameb}{\fieldsSym}{\methodDefsSym}
}
{
  \issubtype{\env}{\cname}{\cnameb}
}
\and
\inferrule[{\subbase}]
{
  \issubtype{\env}{\btype}{\btype'}
  \\\\
  \validimp{\env}{\pred}{\pred'}
}
{
  \issubtype{\env}{\reftp{\btype}{\pred}}{\reftp{\btype'}{\pred'}}
}
\and
\inferrule[{\subwitness}]
{
  \typecheck{\env}{\expr}{\rtypeb}
  \\
  \issubtype{\env}{\rtype}{\appsubst{\tsubst{\expr}{\evar}}{\rtype'}}
}
{
  \issubtype{\env}{\rtype}{\texist{\evar}{\rtypeb}{\rtype'}}
}
\and
\inferrule[{\subbind}]
{
  \issubtype{\envextex{\env}{\evar}{\rtypeb}}{\rtype}{\rtype'}
  \\
  \evar \notin \freevars{\rtype'}
}
{
  \issubtype{\env}{\texist{\evar}{\rtypeb}{\rtype}}{\rtype'}
}
\end{mathpar}
\caption{Subtyping Rules}
\label{fig:subtyping}
\end{figure*}

\begin{figure*}[ht]
\judgementHeadTwo
  {Runtime Typing Rules}
  {\rttypecheckval{\storety}{\val}{\rtype}}
  {\rttypecheck{\storety}{\ssaHeap}{\ssaObject}{\rtype}}
\begin{mathpar}
\inferrule[\rtchkloc]
{
  \idxMapping{\storety}{\loc} = \rtype
}
{
  \rttypecheckval{\storety}{\loc}{\rtype}
}
\and
\inferrule[\rtchkconst]{}
{
  \rttypecheckval{\storety}{\val}{\tconst{\vconst}}
}
\and
\inferrule[\rtchkobj]
{
  \basetype{\idxMapping{\storety}{\loc}} = \cname
  \\
  \rtFieldDefs{\ssaHeap}{\loc} = \mconcatenate
{\immFieldDefs{\fieldname }{\val}_{\immindex}}
    {\mutFieldDefs{\fieldnameb}{\val}_{\mutindex}}
  \\
\rttypecheckval{\storety}{\many{\val}_{\immindex}}{\many{\rtype}_{\immindex}}
}
{
  \rttypecheck{\storety}{\ssaHeap}
    {
      \heapObject{\loc}{\fieldDefsSym}
    }
    {
      \texist
        {\many{\evarb}_{\immindex}}
        {\many{\rtype}_{\immindex}}
        {
          \reftp{\cname}
                {
                  \concatpreds
                    {\tdotref{\vv}{\many{\fieldname}} = \many{\evarb}_{\immindex}}
                    {\classinv{\cname}{\vv}}
                }
        }     
    }
}
\end{mathpar}
\caption{Typing Runtime Configurations for \ssaLang}
\label{fig:rt:typing}
\end{figure*}

\clearpage
\section{Proofs}

\counterwithin{equation}{enumi}
\counterwithin{enumii}{enumi}
\counterwithin{enumiii}{enumii}

\renewcommand{\labelenumi}{$\bullet$}
\renewcommand{\labelenumii}{$\triangleright$}
\renewcommand{\labelenumiii}{--}

\setlist[itemize]{itemsep=5mm}

\setcounter{theorem}{0}

The main results in this section are: 
\begin{itemize}
  \setlength\itemsep{1pt}

  \item Program Consistency Lemma
      (Lemma~\ref{lemma:ssa:expr:consistency:full},
      page~\pageref{lemma:ssa:expr:consistency:full})

  \item Forward Simulation Theorem 
      (Theorem~\ref{theorem:consistency:appendix}, page~\pageref{theorem:consistency:appendix})

  \item Subject Reduction Theorem
      (Theorem~\ref{theorem:subj:reduc}, page~\pageref{theorem:subj:reduc})

  \item Progress Theorem
      (Theorem~\ref{theorem:progress}, page~\pageref{theorem:progress})
\end{itemize}

\subsection{SSA Translation}


\begin{definition}[Environment Substitution]{\label{def:ssa:env:subst}}
%
$$
\esubst{\ssaenv_1}{\ssaenv_2}
\defeq 
\esubst{\many{\evar}_1}{\many{\evar}_2}
\quad\text{where}\quad
\ptriplet{\many{\srcEvar}}{\many{\evar}_1}{\many{\evar}_2} = 
  \envDiff{\ssaenv_1}{\ssaenv_2}
  $$
\end{definition}

\begin{definition}[Valid Configuration]{\label{def:valid:conf}}
$$
\validConf{\pair{\srcRtState}{\srcExprOrBd}} \defeq
\begin{cases}
  \mathit{true} &  \text{if}\;
          \parens{\accessState{\srcRtState}{\srcStack} = \cdot }
          \imp
          \exists\,\srcBody$ \st $\srcExprOrBd \equiv \srcBody
  \\
  \mathit{false} & \text{otherwise}
\end{cases}
$$
\end{definition}

\begin{assumption}[Stack Form]\label{assum:stack:form}
Let stack $\srcStack = \stackCons{\srcStack_0}{\toStack{\srcStore}{\srcEvalCtx}}$. 
Evaluation context $\srcEvalCtx$ is of one of the following forms:
  \begin{itemize}
    \setlength\itemsep{1pt}
    \item {\normalfont $\srcSeq{\srcEvalCtx_0}{\srcReturn{\srcExpr}}$}
    \item {\normalfont $\srcReturn{\srcEvalCtx_0}$}
  \end{itemize}
\end{assumption}


%
%
\begin{lemma}[Global Environment Substitution]\label{lemma:ssa:stmt:subst}
If 
{\normalfont$
    \tossartterm{\srcHeap}{\ssaenvs}{\srcStore}{\srcExpr}{\expr}
  $},
then 
  {\normalfont$
    \tossartterm{\srcHeap}{\ssaenvs'}{\srcStore}{\srcExpr}
        {\appsubst{\esubst{\idxMapping{\ssaenvs'}{\srcExpr}}
                          {\idxMapping{\ssaenvs}{\srcExpr}}}
                  {\expr}}
  $}
\end{lemma}



\begin{lemma}[Evaluation Context]\label{lemma:ssa:eval:ctx}
If
\begin{enumerate}
  \item [] $\tossartterm{\srcHeap}{\ssaenvs}{\srcStore}
             {\srcExprOrBd} 
             {\idxEctx{\ssaEvalCtx}{\expr}}$
\end{enumerate}
then there exist \srcEvalCtx and \srcExpr \st: 
  \begin{itemize}
      \setlength\itemsep{1pt}
      \item $\srcExprOrBd \equiv \idxEctx{\srcEvalCtx}{\srcExpr}$
      \item $\tossaevalctx{\srcHeap}{\ssaenvs}{\srcStore}{\srcEvalCtx}{\ssaEvalCtx}$
      \item $\tossartterm{\srcHeap}{\ssaenvs}{\srcStore}{\srcExpr}{\expr}$
    \end{itemize}
\end{lemma}
\begin{proof}
  \emph{By induction on the derivation of the input transformation.}
\end{proof}

\begin{lemma}[Translation under Store]\label{lemma:ssa:store:translation}
If
$\tossartterm{\cdot}{\ssaenvs}{\cdot}{\srcBody}{\expr}$,
then
$\tossartterm{\srcHeap}{\ssaenvs}{\srcStore}{\srcBody}{\appsubst{\subst}{\expr}}$,
where
\subst = \toSubst{\idxMapping{\ssaenvs}{\srcBody}}{\srcStore}{\srcHeap}.
\end{lemma}
\begin{proof}
  \emph{By induction on the structure of the input translation.}
\end{proof}

\begin{lemma}[Canonical Forms]\label{lemma:canonical:forms}
  \
\begin{enumerate}[label=(\alph*), ref={\thelemma (\alph*)}]
  \item If \label{lemma:canon:forms:a}
    {\normalfont$
      \tossartterm{\srcHeap}{\ssaenvs}{\srcStore}
                  {\srcExprOrBd}{\vconst}$},
    then 
    {\normalfont$
      \srcExprOrBd \equiv \srcVconst
    $}

  \item If \label{lemma:canon:forms:b}
    {\normalfont$
      \tossartterm{\srcHeap}{\ssaenvs}{\srcStore}
                  {\srcExprOrBd}
                  {\methcall{\loc}{\methodname}{\many{\val}}}
    $},
    then 
    {\normalfont$
      \srcExprOrBd \equiv \srcMethcall{\srcLoc}{\srcMethodname}{\many{\srcVal}}
    $}

  \item If \label{lemma:canon:forms:c}
    {\normalfont$
      \tossartterm{\srcHeap}{\ssaenvs}{\srcStore}
                  {\srcExprOrBd}
                  {\letifshort{\many{\phi}}
                              {\expr}
                            {\ssactx_1}{\ssactx_2}{\expr'}}
    $},
    then
    {\normalfont$
      \srcExprOrBd \equiv \srcSeq{\srcIte{\srcExpr}{\srcStmt_1}{\srcStmt_2}}{\srcReturn{\srcExpr'}}
    $}

  \item If \label{lemma:canon:forms:d}
    {\normalfont $\tossameth{\srcMethodDefsSym}{\methdefUnty{\methodname}{\evar}{\expr_0}}$},
    then
    {\normalfont $\srcMethodDefsSym \equiv \srcMethdefUnty{\srcMethodname}{\srcEvar}{\srcBody}$}

\end{enumerate}
\end{lemma}


\begin{lemma}[Translation Closed under Evaluation Context Composition]\label{lemma:ssa:append:ectx}
If 
\begin{enumerate}[label=(\alph*)]
  \item $\tossaevalctx{\srcHeap}{\ssaenvs}{\srcStore}{\srcEvalCtx_0}{\ssaEvalCtx_0}$
  \item $\tossartstackterm{\srcHeap}{\ssaenvs}{\srcStore'}
             {\parens{\concatenate{\srcStore}{\srcEvalCtx_1}}}
             {\srcBody}{\expr}$ 
\end{enumerate}
then
$\tossartstackterm{\srcHeap}{\ssaenvs}{\srcStore'}
                       {\parens{\concatenate{\srcStore}{\idxEctx{\srcEvalCtx_0}{\srcEvalCtx_1}}}}
                       {\srcBody}{\idxEctx{\ssaEvalCtx_0}{\expr}}$
\end{lemma}

\begin{lemma}[Heap and Store Weakening]\label{lemma:ssa:ectx:heap:weaken}
If 
\begin{enumerate}[label=(\alph*)]
  \item[] $\tossartstackterm{\srcHeap}{\ssaenvs}{\srcStore}{\srcStack}{\srcEvalCtx}{\ssaEvalCtxTerm}$
\end{enumerate}
then $\forall\;\srcHeap', \srcStore'$ \st 
  $\srcHeap' \supseteq \srcHeap$ and
  $\srcStore' \supseteq \srcStore$, it holds that 
  $\tossartstackterm{\srcHeap'}{\ssaenvs}{\srcStore'}{\srcStack}{\srcEvalCtx}{\ssaEvalCtxTerm}$
\end{lemma}

\begin{lemma}[Translation Closed under Stack Extension]\label{lemma:ssa:ectx:app:gen}
If 
\begin{enumerate}[label=(\alph*)]
  \item
    $\tossartstackterm{\srcHeap}{\ssaenvs}{\srcStore_0}{\srcStack_0}{\srcEvalCtx_0}{\ssaEvalCtx_0}$
    \label{lemma:ssa:ectx:app:gen:a}
  \item 
    $\tossartstackterm{\srcHeap}{\ssaenvs}{\srcStore_1}{\srcStack_1}{\srcBody_1}{\expr_1}$
    \label{lemma:ssa:ectx:app:gen:b}
\end{enumerate}
then
$\tossartstackterm{\srcHeap}{\ssaenvs}
  {\srcStore_1}
  {\parens{\stackCons{\stackCons{\srcStack_0}{\toStack{\srcStore_0}{\srcEvalCtx_0}}}{\srcStack_1}}}
  {\srcBody_1}
  {\idxEctx{\ssaEvalCtx_0}{\expr_1}}$
\end{lemma}
\begin{proof}
  We proceed by induction on the structure of derivation 
    \ref{lemma:ssa:ectx:app:gen:b}:
  \begin{itemize}
  
    \item \brackets{\ssaexpstackemp}: Fact \ref{lemma:ssa:ectx:app:gen:b} has the form:
      \la{
        \tossartstackterm{\srcHeap}{\ssaenvs}{\srcStore_1}{\cdot}{\srcBody_1}{\expr_1}
        \label{lemma:ssa:ectx:app:gen:0}
      }

      By applying Rule~\ssastackcons on \ref{lemma:ssa:ectx:app:gen:0} and 
      \ref{lemma:ssa:ectx:app:gen:a}:
      \la{
        \tossartstackterm{\srcHeap}{\ssaenvs}
              {\srcStore_1} 
              {\parens{\stackCons{\srcStack_0}{\toStack{\srcStore_0}{\srcEvalCtx_0}}}}
              {\srcBody_1}
              {\idxEctx{\ssaEvalCtx_0}{\expr_1}}
      }

      Which proves the wanted result.

    \item \brackets{\ssastackcons}: 
      Fact \ref{lemma:ssa:ectx:app:gen:b} has the form:
      \la{
        \tossartstackterm{\srcHeap}{\ssaenvs}
              {\srcStore_1} 
              {\parens{\stackCons{\srcStack}{\toStack{\srcStore}{\srcEvalCtx}}}}
              {\srcBody_1}
              {\idxEctx{\ssaEvalCtx}{\expr_{1.1}}}
        \label{lemma:ssa:ectx:app:gen:1}
      }

      By inverting Rule \ssastackcons on \ref{lemma:ssa:ectx:app:gen:1}:
      \la{
        \tossartstackterm{\srcHeap}{\ssaenvs}{\srcStore_1}{\cdot}{\srcBody_1}{\expr_{1.1}}
        \label{lemma:ssa:ectx:app:gen:2}
        \\
        \tossartstackterm{\srcHeap}{\ssaenvs}
              {\srcStore}{\srcStack}{\srcEvalCtx}
              {\ssaEvalCtx}
        \label{lemma:ssa:ectx:app:gen:3}
      }
  
      By induction hypothesis on 
      \ref{lemma:ssa:ectx:app:gen:a} and \ref{lemma:ssa:ectx:app:gen:3}
      (the lemma can easily be extended to evaluation contexts):
      \la{
        \tossartstackterm{\srcHeap}{\ssaenvs}
        {\srcStore}{\parens{\stackCons{\srcStack_0}{\stackCons{\toStack{\srcStore_0}{\srcEvalCtx_0}}{\srcStack}}}}
          {\srcEvalCtx}{\idxEctx{\ssaEvalCtx_0}{\ssaEvalCtx}}
        \label{lemma:ssa:ectx:app:gen:4}
      }

      By applying Rule \ssaectxstackcons on 
      \ref{lemma:ssa:ectx:app:gen:2}
      and 
      \ref{lemma:ssa:ectx:app:gen:4}:
      \la{
        \tossartstackterm{\srcHeap}{\ssaenvs}{\srcStore_1}
        {\parens{
            \stackCons{\srcStack_0}{
              \stackCons{\toStack{\srcStore_0}{\srcEvalCtx_0}}
                        {\stackCons{\srcStack}{\toStack{\srcStore}{\srcEvalCtx}}}}}}
                        {\srcBody_1}{\idxEctx{\ssaEvalCtx_0}{\idxEctx{\ssaEvalCtx}{\expr_{1.1}}}}
      }

      Which proves the wanted result.

  \end{itemize}

\end{proof}

\begin{lemma}[Translation Closed under Evaluation Context Application]\label{lemma:ssa:ectx:app}
If 
\begin{enumerate}[label=(\alph*)]
  \item
    $\tossartstackterm{\srcHeap}{\ssaenvs}{\srcStore}{\srcStack}{\srcEvalCtx}{\ssaEvalCtxTerm}$
    \label{lemma:ssa:ectx:app:a}
  \item $\tossartterm{\srcHeap}{\ssaenvs}{\srcStore}{\srcExpr}{\expr}$
\end{enumerate}
then
$\tossartstackterm{\srcHeap}{\ssaenvs}{\srcStore}{\srcStack}
    {\idxEctx{\srcEvalCtx}{\srcExpr}}
    {\idxEctx{\ssaEvalCtxTerm}{\expr}}$
\end{lemma}
\begin{proof}
  \emph{By induction on the derivation of 
  \ref{lemma:ssa:ectx:app:a}.}
\end{proof}


\begin{lemma}[Method Resolution]\label{lemma:ssa:method:res}
If 
\begin{enumerate}[label=(\alph*)]
  \item \tossaheap{\srcRtState}{\srcHeap}{\ssaHeap}
  \item \tossaval{\srcHeap}{\srcLoc}{\loc}
  \item \toString{\srcMethodname} = \toString{\methodname}
  \item \resolveMeth{\ssaHeap}{\loc}{\methodname}{\methodDefsSym}
\end{enumerate}
then:
\begin{enumerate}[label=(\alph*), resume]
  \item {\normalfont \srcResolveMeth{\srcHeap}{\srcLoc}{\srcMethodname}{\srcMethodDefsSym}} 
  \item \tossameth{\srcMethodDefsSym}{\methodDefsSym}
\end{enumerate}
\end{lemma}

\begin{lemma}[Value Monotonicity]\label{lemma:value:monotonicity}

If    
\begin{enumerate}[label=(\alph*)]
  \item $\validConf{\pair{\srcRtState}{\srcExprOrBd}}$ 
  \item $\tossartconf{\ssaenvs}{\srcRtState}{\srcExprOrBd}{\ssaRtState}{\val}$
        \label{lemma:value:monotonicity:b}
\end{enumerate}
then there exist $\srcStore'$ and $\srcExprOrBd'$ \st:
\begin{enumerate}[label=(\alph*), resume]
  \item {\normalfont $
      \stepsStarFour{\srcRtState}{\srcExprOrBd}{\srcRtState'}{\srcExprOrBd'}
        $}
  \item $\tossartconf{\ssaenvs}
             {\srcRtState'}{\srcExprOrBd'}{\ssaRtState}{\val}$
  \item {\normalfont $\srcExprOrBd' \equiv 
    \begin{cases}
      \srcReturn{\srcVal} & \emph{if}\quad\srcExprOrBd \equiv \srcBody  \\
      \srcVal             & \emph{otherwise}
  \end{cases}$}
  \item If $\accessState{\srcRtState}{\srcStack} = \cdot$ 
        then $\accessState{\srcRtState'}{\srcStore} = \accessState{\srcRtState}{\srcStore}$
\end{enumerate}
where
$\srcRtState' \equiv 
\quadruple{\accessState{\srcRtState}{\srcSignatures}}{\srcStore'}
          {\cdot}{\accessState{\srcRtState}{\srcHeap}}$
\end{lemma}
\begin{proof}
  \emph{By induction on the structure of the derivation \ref{lemma:value:monotonicity:b}.}
\end{proof}

\begin{lemma}[Top-Level Reduction]\label{lemma:top:level:reduction}
If 
\begin{itemize}
\item []
$\stepsFour{\quadruple{\srcSignatures}{\srcStore }{\srcStack}{\srcHeap}}{\srcExprOrBd}
          {\quadruple{\srcSignatures}{\srcStore'}{\srcStack'}{\srcHeap'}}{\srcExprOrBd'}$
\end{itemize}
then for a stack $\srcStack_0$ it holds that:
\begin{itemize}
\item []
  $\stepsFour{\quadruple{\srcSignatures}{\srcStore}
                        {\parens{\stackCons{\srcStack_0}{\srcStack}}}{\srcHeap}}
             {\srcExprOrBd}
             {\quadruple{\srcSignatures}{\srcStore'}
                        {\parens{\stackCons{\srcStack_0}{\srcStack'}}}{\srcHeap'}}
             {\srcExprOrBd'}$
\end{itemize}
\end{lemma}
\begin{proof}
  \emph{By induction on the structure of the input reduction.}
\end{proof}

\begin{lemma}[Empty Stack Consistency]\label{lemma:ssa:empty:stack:cons:full}
If
\begin{enumerate}[label=(\alph*)]
  \item $\tossartconf{\ssaenvs}{\srcRtState}{\srcExprOrBd}{\ssaRtState}{\expr}$
    \label{ssa:es:a}
  \item $\accessState{\srcRtState}{\srcStack} = \cdot$
    \label{ssa:es:b}
  \item \stepsFour{\ssaRtState}{\ssaexpr}{\ssaRtState'}{\ssaexpr'} 
    \label{ssa:es:c}
\end{enumerate}
then there exist 
$\srcRtState'$
and
$\srcExprOrBd'$
\st:
\begin{enumerate}[label=(\alph*), resume]
  \item \stepsStarFour{\srcRtState}{\srcExprOrBd}{\srcRtState'}{\srcExprOrBd'}, 
      \label{ssa:es:d}
  \item $\tossartconf{\ssaenvs}{\srcRtState'}{\srcExprOrBd'}{\ssaRtState'}{\expr'}$
      \label{ssa:es:e}
  \item \label{ssa:es:f} 
    \begin{enumerate}
    \item If 
      {\normalfont $\srcExprOrBd \equiv 
        \idxEctx{\srcEvalCtx}{\srcMethcall{\srcLoc}{\srcMethodname}{\many{\srcVal}}}$}
        then: \begin{enumerate}
          \item $\accessState{\srcRtState'}{\srcStack} = 
                 \toStack{\accessState{\srcRtState}{\srcStore}}{\srcEvalCtx}$
          \item $\accessState{\srcRtState'}{\srcHeap} = \accessState{\srcRtState}{\srcHeap}$
          \item $\exists \srcBody'$ \st $\srcExprOrBd' \equiv \srcBody'$
          \item $\ssaRtState' = \ssaRtState$
        \end{enumerate}
    \item Otherwise: \begin{enumerate}
        \item $\accessState{\srcRtState'}{\srcStack} = \cdot$
        \item $\accessState{\srcRtState'}{\srcHeap} \supseteq \accessState{\srcRtState}{\srcHeap}$
        \item $\accessState{\srcRtState'}{\srcStore} \supseteq \accessState{\srcRtState}{\srcStore}$
        \item If   $\exists \srcExpr $ \st $\srcExprOrBd  \equiv \srcExpr $ 
              then $\exists \srcExpr'$ \st $\srcExprOrBd' \equiv \srcExpr'$
        \item If   $\exists \srcBody $ \st $\srcExprOrBd  \equiv \srcBody $ 
              then $\exists \srcBody'$ \st $\srcExprOrBd' \equiv \srcBody'$
      \end{enumerate}
  \end{enumerate}
\end{enumerate}
\end{lemma}

\begin{proof}

  Fact~\ref{ssa:es:a} has the form:
  \la{
    \tossartconf
      {\ssaenvs}
      {\srcRtState}{\srcExprOrBd}
      {\pair{\ssaSignatures}{\ssaHeap}}{\expr}
      \label{ssa:es:29}
  }

  Because of fact~\ref{ssa:es:b}:
  \la{
    \srcRtState \equiv \quadruple{\srcSignatures}{\srcStore}{\cdot}{\srcHeap}
  }

  By inverting Rule~\ssartconf on \ref{ssa:es:29}:
  \la{
    \tossartsigs{\ssaenvs}{\srcSignatures}{\ssaSignatures}
    \label{ssa:es:38}
    \\
    \tossaheap{\srcRtState}{\srcHeap}{\ssaHeap}
    \label{ssa:es:39}
    \\
    \tossartstackterm{\srcHeap}{\ssaenvs}
                {\srcStore}{\cdot}{\srcExprOrBd}{\expr}
    \label{ssa:es:23}
  }

  By inverting \ssaexpstackemp on \ref{ssa:es:82}:
  \la{
    \tossartterm{\srcHeap}{\ssaenvs}
      {\srcStore}{\srcExprOrBd}
      {\expr}
    \label{ssa:es:110}
  }

  Suppose $\srcExprOrBd$ is a value. By Rules \ssaconst and \ssaloc, 
  $\expr$ is also a value: a contradiction because of \ref{ssa:es:c}. 
  Hence:
  \la{
    \notaval{\srcExprOrBd}
    \label{ssa:es:84}
  }
  



  We proceed by induction on the structure of reduction \ref{ssa:es:c}:
  \begin{itemize}
    
    %
    %
    \item \brackets{\rcevalctx}

      \la{
        \stepsFour{\ssaRtState}
                  {\idxEctx{\ssaEvalCtx_0}{\expr_0}}
                  {\ssaRtState'}
                  {\idxEctx{\ssaEvalCtx_0}{\expr_0'}}
        \label{ssa:es:28}
      }

      By inverting \rcevalctx on \ref{ssa:es:28}:
      \la{
        \stepsFour{\ssaRtState}{\expr_0}{\ssaRtState'}{\expr_0'}
        \label{ssa:es:24}
      }

      Fact~\ref{ssa:es:110} is of the form:
      \la{      
        \tossartterm{\srcHeap}{\ssaenvs}
          {\srcStore}{\srcExprOrBd}{\idxEctx{\ssaEvalCtx_0}{\expr_0}}
        \label{ssa:es:82}
      }

      By Lemma~\ref{lemma:ssa:eval:ctx} on \ref{ssa:es:82}:
      \la{
        \srcExprOrBd \equiv \idxEctx{\srcEvalCtx_0}{\srcExpr_0}
        \label{ssa:es:83}
        \\
        \tossaevalctx{\srcHeap}{\ssaenvs}{\srcStore}{\srcEvalCtx_0}{\ssaEvalCtx_0}
        \label{ssa:es:49}
        \\
        \tossartterm{\srcHeap}{\ssaenvs}{\srcStore}{\srcExpr_0}{\expr_0}
        \label{ssa:es:31}
      }

      By applying Rule~\ssaexpstackemp on \ref{ssa:es:31}:
      \la{
        \tossartstackterm{\srcHeap}{\ssaenvs}{\srcStore}{\cdot}{\srcExpr_0}{\expr_0}  
        \label{ssa:es:111}
      }

      By applying Rule \ssartconf on \ref{ssa:es:38}, \ref{ssa:es:39} and \ref{ssa:es:111}:
      \la{
        \tossartconf{\ssaenvs}{\srcRtState}{\srcExpr_0}{\ssaRtState}{\expr_0}
        \label{ssa:es:112}
      }

      By induction hypothesis using \ref{ssa:es:112}, \ref{ssa:es:b} and \ref{ssa:es:24}:
      \la{
        \stepsFour
          {\quadruple{\srcSignatures}{\srcStore}{\cdot}{\srcHeap}}
          {\srcExpr_0}
          {\quadruple{\srcSignatures}{\srcStore'}{\srcStack'}{\srcHeap'}}
          {\srcExprOrBd_0'}
          \label{ssa:es:41}
        \\
        \tossartconf
          {\ssaenvs}
          {\quadruple{\srcSignatures}{\srcStore'}{\srcStack'}{\srcHeap'}}
          {\srcExprOrBd_0'}
          {\ssaRtState'}
          {\expr_0'}
          \label{ssa:es:42}
      }

      We examine cases on the form of $\srcExpr_0$:
      \begin{itemize}

        \item Case {\normalfont 
            $\srcExpr_0 \equiv
              \idxEctx{\srcEvalCtx_1}{\srcMethcall{\srcLoc}{\srcMethodname}{\many{\srcVal}}}$
            }:
          \la{
            \srcStack'    & = \concatenate{\srcStore}{\srcEvalCtx_1}
            \label{ssa:es:43}
            \\
            \srcHeap'     & = \srcHeap
            \label{ssa:es:99}
            \\
            \srcExprOrBd_0' & = \srcBody'
            \label{ssa:es:44}
            \\
            \ssaRtState' & = \ssaRtState
            \label{ssa:es:86}
          }
          For some method body $\srcBody'$.
          So \ref{ssa:es:42} becomes:
          \la{
            \tossartconf
              {\ssaenvs}
              {\quadruple{\srcSignatures}{\srcStore'}
                         {\parens{\concatenate{\srcStore}{\srcEvalCtx_1}}}
                         {\srcHeap}}
              {\srcBody'}
              {\ssaRtState}{\expr_0'}
              \label{ssa:es:66}
          }

          By inverting rule \srcRcCall on \ref{ssa:es:41}:
          \la{
            \srcResolveMeth
              {\srcHeap}{\srcLoc}{\srcMethodname}
              {\srcMethdefUnty{\srcMethodname}{\srcEvar}{\srcBody'}}
            \label{ssa:es:87}
            \\
            \srcStore' = \concatenate{\mappings{\srcEvar}{\srcVal}}{\mapping{\srcThis}{\srcLoc}}
            \label{ssa:es:88}
            \\
            \srcStack_0' = \toStack{\srcStore}{\srcEvalCtx_1}
            \label{ssa:es:89}
          }

          By applying rule \srcRcCall using 
          \ref{ssa:es:87}, 
          \ref{ssa:es:88} and 
          $\srcStack' = \toStack{\srcStore}{\idxEctx{\srcEvalCtx_0}{\srcEvalCtx_1}}$
          on 
          $\pair{\srcRtState}{\srcExprOrBd} \equiv 
            \pair{\quadruple{\srcSignatures}{\srcStore}{\cdot}{\srcHeap}}
                  {\idxEctx{\parens{\idxEctx{\srcEvalCtx_0}{\srcEvalCtx_1}}}
                           {\srcMethcall{\srcLoc}{\srcMethodname}{\many{\srcVal}}}}$:
          \la{
            \stepsFour
              {\quadruple{\srcSignatures}{\srcStore}{\cdot}{\srcHeap}}
              {\idxEctx{\parens{\idxEctx{\srcEvalCtx_0}{\srcEvalCtx_1}}}
                       {\srcMethcall{\srcLoc}{\srcMethodname}{\many{\srcVal}}}}
              {\quadruple{\srcSignatures}
                         {\srcStore'}
                         {\parens{\concatenate{\srcStore}{\idxEctx{\srcEvalCtx_0}{\srcEvalCtx_1}}}}
                         {\srcHeap}
              }
              {\srcBody'}              
            \label{ssa:es:94}
          }
          Which proves \ref{ssa:es:d}.
          By inverting Rule \ssartconf on \ref{ssa:es:66}:
          \la{
            \tossaheap{\srcRtState'}{\srcHeap}{\ssaHeap}
            \label{ssa:es:93}
            \\
            \tossartstackterm{\srcHeap}{\ssaenvs}{\srcStore'}
                        {\parens{\concatenate{\srcStore}{\srcEvalCtx_1}}}
                        {\srcBody'}{\expr_0'}
            \label{ssa:es:90}
          }

          From Lemma~\ref{lemma:ssa:append:ectx} on \ref{ssa:es:49} and \ref{ssa:es:90}:
          \la{
            \tossartstackterm{\srcHeap}{\ssaenvs}{\srcStore'}
                        {\parens{\concatenate{\srcStore}{\idxEctx{\srcEvalCtx_0}{\srcEvalCtx_1}}}}
                        {\srcBody'}{\idxEctx{\ssaEvalCtx_0}{\expr_0'}}
            \label{ssa:es:91}
          }

          By applying rule \ssartconf using \ref{ssa:es:38}, \ref{ssa:es:93} and \ref{ssa:es:91}:
          \la{
            \tossartconf
              {\ssaenvs}
              {\quadruple{\srcSignatures}
                         {\srcStore'}
                         {\parens{\concatenate{\srcStore}{\idxEctx{\srcEvalCtx_0}{\srcEvalCtx_1}}}}
                         {\srcHeap}}
              {\srcBody'}
              {\ssaRtState}{\idxEctx{\ssaEvalCtx_0}{\expr_0'}}
              \label{ssa:es:45}
          }
          Which proves \ref{ssa:es:e}.
          By \ref{ssa:es:83} and the current case:
          \la{
            \srcExprOrBd \equiv 
            \idxEctx{\parens{\idxEctx{\srcEvalCtx_0}{\srcEvalCtx_1}}}
                    {\srcMethcall{\srcLoc}{\srcMethodname}{\many{\srcVal}}}
            \label{ssa:es:95}
          }

          By \ref{ssa:es:94} and \ref{ssa:es:45}:
          \la{
            \accessState{\srcRtState'}{\srcStack} = 
              \concatenate{\srcStore}{\idxEctx{\srcEvalCtx_0}{\srcEvalCtx_1}}
            \label{ssa:es:96}
            \\
            \srcExprOrBd' = \srcBody'
            \label{ssa:es:97}
            \\
            \ssaRtState' = \ssaRtState
            \label{ssa:es:98}
          }
          
          By \ref{ssa:es:96}, \ref{ssa:es:99}, 
          \ref{ssa:es:97} and 
          \ref{ssa:es:98} we prove \ref{ssa:es:f}.

        \item All remaining cases:
          \la{
            \srcStack'   \equiv \cdot
            \label{ssa:es:46}
            \\
            \srcHeap' \supseteq \srcHeap
            \label{ssa:es:101}
            \\
            \srcStore' \supseteq \srcStore
            \label{ssa:es:105}
            \\
            \srcExprOrBd_0' \equiv \srcExpr_0'
            \label{ssa:es:47}
          }

          So \ref{ssa:es:41} and \ref{ssa:es:42} become:
          \la{
            \stepsFour
              {\quadruple{\srcSignatures}{\srcStore}{\cdot}{\srcHeap}}
              {\srcExpr_0}
              {\quadruple{\srcSignatures}{\srcStore'}{\cdot}{\srcHeap'}}
              {\srcExpr_0'}
            \label{ssa:es:100}
            \\         
            \tossartconf
              {\ssaenvs}
              {\quadruple{\srcSignatures}{\srcStore'}{\cdot}{\srcHeap'}}
              {\srcExpr_0'}
              {\ssaRtState'}{\expr_0'}
              \label{ssa:es:103}
          }

          By applying Rule~\srcRcEvalCtx using \ref{ssa:es:100}:
          \la{
            \stepsFour
              {\quadruple{\srcSignatures}{\srcStore}{\cdot}{\srcHeap}}
              {\idxEctx{\srcEvalCtx_0}{\srcExpr_0}}
              {\quadruple{\srcSignatures}{\srcStore'}{\cdot}{\srcHeap'}}
              {\idxEctx{\srcEvalCtx_0}{\srcExpr_0'}}
          }
          Which proves \ref{ssa:es:d} and \ref{ssa:es:f}.
          By inverting Rules~\ssartconf and \ssaexpstackemp on \ref{ssa:es:103}:
          \la{
            \tossartterm{\srcHeap'}{\ssaenvs}{\srcStore'}{\srcExpr_0}{\expr_0}
            \label{ssa:es:104}
          }
          
          From Lemma~\ref{lemma:ssa:ectx:heap:weaken} using 
          \ref{ssa:es:49}, \ref{ssa:es:101} and \ref{ssa:es:105}:
          \la{
            \tossaevalctx{\srcHeap'}{\ssaenvs}{\srcStore'}{\srcEvalCtx_0}{\ssaEvalCtx_0}
            \label{ssa:es:106}
          }

          From Lemma~\ref{lemma:ssa:ectx:app} on \ref{ssa:es:104} and \ref{ssa:es:106}:
          \la{
            \tossartterm{\srcHeap'}{\ssaenvs}{\srcStore'}
              {\idxEctx{\srcEvalCtx_0}{\srcExpr_0}}
              {\idxEctx{\ssaEvalCtx_0}{\expr_0}}
            \label{ssa:es:107}
          }

          By inverting rule \ssartconf on \ref{ssa:es:103}:
          \la{
            \tossaheap{\srcRtState'}{\srcHeap'}{\ssaHeap'}
            \label{ssa:es:108}
          }

          By Rule \ssartconf using \ref{ssa:es:38}, \ref{ssa:es:107} and \ref{ssa:es:108}:
          \la{
            \tossartconf
              {\ssaenvs}
              {\quadruple{\srcSignatures}{\srcStore'}{\cdot}{\srcHeap'}}
              {\idxEctx{\srcEvalCtx_0}{\srcExpr_0'}}
              {\pair{\ssaSignatures}{\ssaHeap'}}
              {\idxEctx{\ssaEvalCtx_0}{\expr_0'}}
          }
          Which proves \ref{ssa:es:e}.

      \end{itemize}

    %
    %
    
    \item \brackets{\rinvoke}:
      \la{
        \stepsFour
          {\ssaRtState}
          {\methcall{\loc}{\methodname}{\many{\val}}}
          {\ssaRtState}
          {\appsubst{\tsubsttwo{\many{\val}}{\many{\evar}}{\loc}{\ethis}}{\expr_0}}
          \label{ssa:es:11}
      }

      Where by inverting \rinvoke on \ref{ssa:es:11}:
      \la{
        \resolveMeth{\ssaHeap}{\loc}{\methodname}
                    {\parens{\methdefUnty{\methodname}{\evar}{\expr_0}}}  
        \label{ssa:es:12}
      }

      Fact~\ref{ssa:es:23} is of the form:
      \la{
        \tossartstackterm{\srcHeap}{\ssaenvs}
          {\srcStore}{\cdot}{\srcExprOrBd}
          {\methcall{\loc}{\methodname}{\many{\val}}}
        \label{ssa:es:call:1}
      }

      By Lemma~\ref{lemma:canon:forms:b} on \ref{ssa:es:call:1}:
      \la{
        \srcExprOrBd \equiv \srcMethcall{\srcLoc}{\srcMethodname}{\many{\srcVal}}
        \label{ssa:es:call:2}
      }

      So \ref{ssa:es:call:1} becomes:
      \la{
        \tossartstackterm{\srcHeap}{\ssaenvs}
          {\srcStore}{\cdot}{\srcMethcall{\srcLoc}{\srcMethodname}{\many{\srcVal}}}
          {\methcall{\loc}{\methodname}{\many{\val}}}
        \label{ssa:es:call:6}
      }

      By inverting Rule~\ssaexpstackemp on \ref{ssa:es:call:6}:
      \la{
        \tossartterm{\srcHeap}{\ssaenvs}{\srcStore}
          {\srcMethcall{\srcLoc}{\srcMethodname}{\many{\srcVal}}}
          {\methcall{\loc}{\methodname}{\many{\val}}}
        \label{ssa:es:call:7}
      }

      By inverting Rule~\ssartmethcall on \ref{ssa:es:call:7}:
      \la{ 
        \tossartterm{\srcHeap}{\ssaenvs}{\srcStore}{\srcLoc}{\loc}
        \label{ssa:es:call:8}
        \\
        \tossartterm{\srcHeap}{\ssaenvs}{\srcStore}{\many{\srcVal}}{\many{\val}}
        \label{ssa:es:call:9}
        \\
        \toString{\srcMethodname} = \toString{\methodname}
        \label{ssa:es:call:10}
      }

      By inverting \ssartval on \ref{ssa:es:call:8} and \ref{ssa:es:call:9}:
      \la{
        \tossaval{\srcHeap}{\srcLoc}{\loc}
        \label{ssa:es:call:11}
        \\
        \tossaval{\srcHeap}{\many{\srcVal}}{\many{\val}}
        \label{ssa:es:call:21}
      }

      By Lemma~\ref{lemma:ssa:method:res} on 
      \ref{ssa:es:39}, 
      \ref{ssa:es:call:11}, 
      \ref{ssa:es:call:10} and 
      \ref{ssa:es:12}:
      \la{
        \srcResolveMeth{\srcHeap}{\srcLoc}{\srcMethodname}{\srcMethodDefsSym}  
        \label{ssa:es:call:12}
        \\
        \tossartsigs{\ssaenvs}{\srcMethodDefsSym}{\methdefUnty{\methodname}{\evar}{\expr_0}}
        \label{ssa:es:call:13}
      }
       
      By Lemma~\ref{lemma:canon:forms:d} on \ref{ssa:es:call:13}:
      \la{
        \srcMethodDefsSym \equiv 
          \srcMethdefUnty{\srcMethodname}{\srcEvar}{\srcBody}                                          
        \label{ssa:es:call:14}
      }

      By applying Rule~\srcRcCall using 
      \ref{ssa:es:call:12},
      \ref{ssa:es:call:17},
      \ref{ssa:es:call:18} and $\srcEvalCtx \equiv \empevalctx$:
      \la{
        \stepsFour
            {\quadruple{\srcSignatures}{\srcStore}{\srcStack}{\srcHeap}}
            {\srcMethcall{\srcLoc}{\srcMethodname}{\many{\srcVal}}}
            {\quadruple{\srcSignatures}{\srcStore'}{\srcStack'}{\srcHeap}}
            {\srcBody}
        \label{ssa:es:call:22}
      }
      Which proves \ref{ssa:es:d}.
      By inverting rule \ssartmethdecl on \ref{ssa:es:call:13}:
      \la{
        \tossartterm{\cdot}{\ssaenvs}{\cdot}{\srcBody}{\expr}
        \label{ssa:es:call:15}
      }

      Let a store $\srcStore'$ and a stack $\srcStack'$ \st:
      \la{
        \srcStore' \equiv \concatenate{\mappings{\srcEvar}{\srcVal}}{\mapping{\srcThis}{\srcLoc}}
        \label{ssa:es:call:17}
        \\
        \srcStack' \equiv \toStack{\srcStore}{\empevalctx}
        \label{ssa:es:call:18}
      }

      By applying Lemma~\ref{lemma:ssa:store:translation} on \ref{ssa:es:call:15} 
      \la{
        \tossartterm{\srcHeap}{\ssaenvs}{\srcStore'}
                    {\srcBody}{\appsubst{\subst}{\expr_0}}
        \label{ssa:es:call:23}
      }

      Where:
      \la{
        \subst & \doteq 
        \toSubst{\idxMapping{\ssaenvs}{\srcBody}}{\srcStore'}{\srcHeap} 
        \notag
        \\ & = 
          \braces{\esubst{\val}{\evar}\,|\,
            \mapping{\srcEvar}{\evar}   \in \idxMapping{\ssaenvs}{\srcBody},\,
            \mapping{\srcEvar}{\srcVal} \in \srcStore',\,
            \tossaval{\srcHeap}{\srcVal}{\val}
          }
        \notag
        \\ & = 
          \tsubsttwo{\many{\val}}{\many{\evar}}{\loc}{\ethis}
        \label{ssa:es:call:20}
      }

      We pick:
      \la{
        \srcExprOrBd' \equiv \srcBody
        \label{ssa:es:call:24}
      }

      By applying Rule \ssaexpstackemp using \ref{ssa:es:call:23}:
      \la{
        \tossartstackterm{\srcHeap}{\ssaenvs}
                         {\srcStore'}{\cdot}
                         {\srcBody}
                         {\appsubst{\subst}{\expr_0}}
        \label{ssa:es:call:25}
      }

      It holds that: 
      \la{
        \tossartstackterm{\srcHeap}{\ssaenvs}
              {\srcStore}{\cdot}{\empevalctx}
              {\empevalctx}
        \label{ssa:es:call:26}
      }

      By Rule \ssastackcons on \ref{ssa:es:call:25} and 
      \ref{ssa:es:call:26}:
      \la{
        \tossartstackterm{\srcHeap}{\ssaenvs}
              {\srcStore'}{\parens{\toStack{\srcStore}{\empevalctx}}}{\srcBody}
              {\appsubst{\subst}{\expr_0}}
        \label{ssa:es:call:27}
      }

      By Rule \ssartconf using \ref{ssa:es:38}, \ref{ssa:es:39} and \ref{ssa:es:call:27}:
      \la{
        \tossartconf
          {\ssaenvs}
          {\quadruple{\srcSignatures}{\srcStore'}{\srcStack'}{\srcHeap}}
          {\srcBody}
          {\pair{\ssaSignatures}{\ssaHeap}}
          {\appsubst{\subst}{\expr_0}}
      }
      Which proves \ref{ssa:es:e}.
      From 
      \ref{ssa:es:call:18}, 
      \ref{ssa:es:call:22}, 
      \ref{ssa:es:call:24} and 
      \ref{ssa:es:call:11} 
      we prove \ref{ssa:es:f}.

    \item \brackets{\rletif}: 
      \la{
        \stepsFour{\ssaRtState}{
          \letif{\many{\evar}}{\many{\evar}_1}{\many{\evar}_2}
                {\vconst}{\ssactx_{1}}{\ssactx_{2}}{\expr_0}
          }
          {\ssaRtState}
          {
            \ssactxidx{\ssactx_{\index}}
            {\appsubst{\esubst{\many{\evar}_{\index}}{\many{\evar}}}{\expr_0}}
          }  
        \label{ssa:es:ite:0}  
        \\
        \vconst = \ssaTrue  \imp \index = 1 
        \\
        \vconst = \ssaFalse \imp \index = 2
      }    

      Let:
      \la{
        \vconst = \ssaTrue
      }
      
      The case for $\ssaFalse$ is symmetrical. 
      Facts \ref{ssa:es:ite:0} and \ref{ssa:es:110} become:
      \la{
        \stepsFour{\ssaRtState}{
          \letif{\many{\evar}}{\many{\evar}_1}{\many{\evar}_2}
                {\ssaTrue}{\ssactx_{1}}{\ssactx_{2}}{\expr_0}
          }
          {\ssaRtState}
          {
            \ssactxidx{\ssactx_1}
                      {\appsubst{\esubst{\many{\evar}_1}{\many{\evar}}}{\expr_0}}
          }
        \label{ssa:es:ite:1} 
        \\
        \tossartterm{\srcHeap}{\ssaenvs}
          {\srcStore}{\srcExprOrBd}
          {\letif{\many{\evar}}{\many{\evar}_1}{\many{\evar}_2}
          {\ssaTrue}{\ssactx_1}{\ssactx_2}{\expr_0}} 
        \label{ssa:es:ite:2} 
      }    
      
      By Lemma~\ref{lemma:canon:forms:c} on \ref{ssa:es:ite:2}:
      \la{
        \srcExprOrBd \equiv
          \srcSeq{\srcIte{\srcExpr_c}{\srcStmt_1}{\srcStmt_2}}{\srcReturn{\srcExpr_0}}
        \label{ssa:es:ite:3}  
      }

      So \ref{ssa:es:ite:2} becomes:
      \la{
        \tossartterm{\srcHeap}{\ssaenvs}
          {\srcStore}
          {\srcSeq{\srcIte{\srcExpr_c}{\srcStmt_1}{\srcStmt_2}}{\srcReturn{\srcExpr_0}}}
          {\letif{\many{\evar}}{\many{\evar}_1}{\many{\evar}_2}
          {\ssaTrue}{\ssactx_1}{\ssactx_2}{\expr_0}} 
        \label{ssa:es:ite:40} 
      }

      By inverting Rule \ssartmethbody on \ref{ssa:es:ite:40}:
      \la{
        \tossartterm{\srcHeap}{\ssaenvs}
          {\srcStore}
          {\srcIte{\srcExpr_c}{\srcStmt_1}{\srcStmt_2}}
          {\letif{\many{\evar}}{\many{\evar}_1}{\many{\evar}_2}
          {\ssaTrue}{\ssactx_1}{\ssactx_2}{\hole}} 
        \label{ssa:es:ite:5} 
        \\
        \ssaenvs' = \updMapping{\ssaenvs}{\srcExpr_0}{\idxMappingPost{\ssaenvs}
          {\srcIte{\srcExpr_c}{\srcStmt_1}{\srcStmt_2}}}
        \label{ssa:es:ite:6} 
        \\
        \tossartterm{\srcHeap}{\ssaenvs'}{\srcStore}{\srcExpr_0}{\expr_0}
        \label{ssa:es:ite:7} 
      }
    
      By inverting Rule \ssartite on \ref{ssa:es:ite:5}:
      \la{
        \tossartterm{\srcHeap}{\ssaenvs}{\srcStore}
          {\srcExpr_c}{\ssaTrue}
        \label{ssa:es:ite:8} 
        \\
        \tossartterm{\srcHeap}{\ssaenvs}{\srcStore}
          {\srcStmt_1}{\ssactx_1}
        \label{ssa:es:ite:9} 
        \\
        \tossartterm{\srcHeap}{\ssaenvs}{\srcStore}
          {\srcStmt_2}{\ssactx_2}
        \label{ssa:es:ite:10} 
        \\
        \ptriplet{\many{\srcEvar}}{\many{\evar}_1}{\many{\evar}_2} 
          = \envDiff{\idxMappingPost{\ssaenvs}{\srcStmt_1}} 
                    {\idxMappingPost{\ssaenvs}{\srcStmt_2}} 
        \label{ssa:es:ite:11} 
        \\
        \many{\evar} = \idxMapping
          {\idxMappingPost{\ssaenvs}{\srcIte{\srcExpr_c}{\srcStmt_1}{\srcStmt_2}}}
          {\many{\srcEvar}}
        \label{ssa:es:ite:12} 
      }     
      
      By Lemma~\ref{lemma:canonical:forms} on \ref{ssa:es:ite:8} we get:
      \la{
        \srcExpr_c \equiv \vtrue    
        \label{ssa:es:ite:4}
      }
     
      By Rules~\srcRcEvalCtx and \srcRcIte we get:
      \la{
        \stepsFour{\srcRtState}
                  {\srcSeq{\srcIte{\vtrue}{\srcStmt_1}{\srcStmt_2}}{\srcReturn{\srcExpr_0}}}
                  {\srcRtState}
                  {\srcSeq{\srcStmt_1}{\srcReturn{\srcExpr_0}}}
      }
      Which proves \ref{ssa:es:d}. Let:
      \la{
        \ssaenvs'' \equiv 
          \updMapping{\ssaenvs'}{\srcExpr_0}
                     {\idxMappingPost{\ssaenvs}
                                     {\srcStmt_1}}
        \label{ssa:es:ite:13}
      }

      By Lemma~\ref{lemma:ssa:stmt:subst} on \ref{ssa:es:ite:7} using \ref{ssa:es:ite:13}:
      \la{
        \tossartterm{\srcHeap}{\ssaenvs''}{\srcStore}{\srcExpr_0}
          {\appsubst{\esubst{\idxMapping{\ssaenvs''}{\srcExpr_0}}
                            {\idxMapping{\ssaenvs'}{\srcExpr_0}}}
                    {\expr_0}}
        \label{ssa:es:ite:14}
      }

      From \ref{ssa:es:ite:6} and \ref{ssa:es:ite:13} it holds that:
      \la{
        \idxMapping{\ssaenvs'}{\srcExpr_0} 
        & = \idxMappingPost{\ssaenvs}{\srcIte{\vtrue}{\srcStmt_1}{\srcStmt_2}}
        \label{ssa:es:ite:18}
        \\
        \idxMapping{\ssaenvs''}{\srcExpr_0} 
        & = \idxMappingPost{\ssaenvs}{\srcStmt_1} 
        \label{ssa:es:ite:19}
      }

      So:
      \la{
        \envDiff{\idxMapping{\ssaenvs'}{\srcExpr_0}}{\idxMapping{\ssaenvs''}{\srcExpr_0}}
          = \ptriplet{\many{\srcEvar}}{\many{\evar}_1}{\many{\evar}} 
          \label{ssa:es:ite:15}
      }

      By Definition~\ref{def:ssa:env:subst}:
      \la{
        \esubst{\idxMapping{\ssaenvs''}{\srcExpr_0}}
               {\idxMapping{\ssaenvs'}{\srcExpr_0}} = 
        \esubst{\many{\evar}_1}{\many{\evar}}  
          \label{ssa:es:ite:16}
      }
      So \ref{ssa:es:ite:14} becomes:
      \la{
        \tossartterm{\srcHeap}{\ssaenvs''}{\srcStore}{\srcExpr_0}
          {\appsubst{\esubst{\many{\evar}_1}{\many{\evar}}}{\expr_0}}
        \label{ssa:es:ite:17}
      }

      By applying Rule~\ssartmethbody on \ref{ssa:es:ite:9}, \ref{ssa:es:ite:19} 
      and \ref{ssa:es:ite:17}, using \ref{ssa:es:ite:16}:
      \la{
        \tossartterm{\srcHeap}{\ssaenvs}{\srcStore}
          {\srcSeq{\srcStmt_1}{\srcReturn{\srcExpr_0}}}
          {\ssactxidx{\ssactx_1}
            {\appsubst{\esubst{\many{\evar}_1}{\many{\evar}}}{\expr_0}}}
      }

      Which, using \ssartconf and \ssaexpstackemp, 
      prove \ref{ssa:es:e} and \ref{ssa:es:f}.

    \item \brackets{\rcast}, \brackets{\rnew}, \brackets{\rletin}, 
          \brackets{\rdotasgn}, \brackets{\rfield}: 
          \emph{Cases handled in similar fashion as before.}

  \end{itemize}

\end{proof}

\begin{corollary}[Empty Stack Valid Configuration]\label{cor:empty:stack:valid}
If
\begin{enumerate}[label=(\alph*)]
  \item $\tossartconf{\ssaenvs}{\srcRtState}{\srcExprOrBd}{\ssaRtState}{\expr}$
    \label{ssa:vc:a}
  \item $\accessState{\srcRtState}{\srcStack} = \cdot$
    \label{ssa:vc:b}
  \item \stepsFour{\ssaRtState}{\ssaexpr}{\ssaRtState'}{\ssaexpr'} 
    \label{ssa:vc:c}
\end{enumerate}
then $\stepsStarFour{\srcRtState}{\srcExprOrBd}{\srcRtState'}{\srcExprOrBd'}$
with
$\validConf{\pair{\srcRtState'}{\srcExprOrBd'}}$.
\end{corollary}
\begin{proof}
\emph{Examine all cases of result \ref{ssa:es:f} of Lemma~\ref{lemma:ssa:empty:stack:cons:full}.}
\end{proof}

\begin{lemma}[Consistency]\label{lemma:ssa:expr:consistency:full}
If
\begin{enumerate}[label=(\alph*)]
  \item $\tossartconf{\ssaenvs}{\srcRtState}{\srcExprOrBd}{\ssaRtState}{\expr}$
    \label{ssa:sec:a}
  \item \stepsFour{\ssaRtState}{\ssaexpr}{\ssaRtState'}{\ssaexpr'} 
    \label{ssa:sec:b}
  \item \validConf{\pair{\srcRtState}{\srcExprOrBd}} 
    \label{ssa:sec:c}
\end{enumerate}
then there exist 
$\srcRtState'$
and
$\srcExprOrBd'$
\st:
\begin{enumerate}[label=(\alph*), resume]
  \item \stepsStarFour{\srcRtState}{\srcExprOrBd}{\srcRtState'}{\srcExprOrBd'}, 
      \label{ssa:sec:d}
  \item $\tossartconf{\ssaenvs}{\srcRtState'}{\srcExprOrBd'}{\ssaRtState'}{\expr'}$
      \label{ssa:sec:e}
  \item \validConf{\pair{\srcRtState'}{\srcExprOrBd'}} 
    \label{ssa:sec:f}
\end{enumerate}
\end{lemma}

\begin{proof}

Let:
\la{
  \srcRtState \equiv \quadruple{\srcSignatures}{\srcStore}{\srcStack}{\srcHeap}
  \label{ssa:sec:0}
}

By inverting Rule~\ssartconf on \ref{ssa:sec:a}:
\la{
  \tossartsigs{\ssaenvs}{\srcSignatures}{\ssaSignatures}
  \label{ssa:sec:1}
  \\
  \tossaheap{\srcRtState}{\srcHeap}{\ssaHeap}
  \label{ssa:sec:2}
  \\
  \tossartstackterm{\srcHeap}{\ssaenvs}{\srcStore}{\srcStack}{\srcExprOrBd}{\expr}
  \label{ssa:sec:3}
}

We proceed by induction on the derivation \ref{ssa:sec:3}:
\begin{itemize}

  \item \brackets{\ssaexpstackemp}:
    \la{
      \tossartstackterm{\srcHeap}{\ssaenvs}{\srcStore}{\cdot}{\srcExprOrBd}{\expr}
      \label{ssa:sec:4}
    }

    By Lemma~\ref{lemma:ssa:empty:stack:cons:full} using \ref{ssa:sec:a} and \ref{ssa:sec:b}
    there exist $\srcExprOrBd'$ and $\srcRtState'$ \st:
    \la{
      \stepsStarFour{\srcRtState}{\srcExprOrBd}{\srcRtState'}{\srcExprOrBd'}
      \label{ssa:sec:5}
      \\
      \tossartconf{\ssaenvs}{\srcRtState'}{\srcExprOrBd'}{\ssaRtState'}{\expr'}
      \label{ssa:sec:6}
    }

    From Corollary~\ref{cor:empty:stack:valid} using 
    \ref{ssa:sec:a},
    \ref{ssa:sec:b} and 
    \ref{ssa:sec:c} we get:
    \la{
      \validConf{\pair{\srcRtState'}{\srcExprOrBd'}}
      \label{ssa:sec:57}
    }

    We prove \ref{ssa:sec:d}, \ref{ssa:sec:e} and \ref{ssa:sec:f} by 
    \ref{ssa:sec:5}, \ref{ssa:sec:6} and \ref{ssa:sec:57}, respectively.

  \item \brackets{\ssastackcons}:
    \la{
      \tossartstackterm{\srcHeap}{\ssaenvs}{\srcStore}
        {\parens{\stackCons{\srcStack_0}{\toStack{\srcStore_0}{\srcEvalCtx_0}}}}
        {\srcExprOrBd}
        {\idxEctx{\ssaEvalCtx_0}{\expr_0}}
      \label{ssa:sec:7}
    }

    Where:
    \la{
      \srcStack \equiv \stackCons{\srcStack_0}{\toStack{\srcStore_0}{\srcEvalCtx_0}}
    }

    By \ref{ssa:sec:c} and the definition of a \emph{valid configuration}, 
    there exists a $\srcBody_0$ \st:
    \la{
      \srcExprOrBd \equiv \srcBody_0
      \label{ssa:sec:25}
    }

    By inverting Rule~\ssastackcons on \ref{ssa:sec:7} using \ref{ssa:sec:25}:
    \la{
      \tossartstackterm{\srcHeap}{\ssaenvs}{\srcStore}{\cdot}{\srcBody_0}{\expr_0}
      \label{ssa:sec:9}
      \\
      \tossartstackterm{\srcHeap}{\ssaenvs}
                  {\srcStore_0}{\srcStack_0}{\srcEvalCtx_0}
                  {\ssaEvalCtx_0}
      \label{ssa:sec:10}
    }

    By applying rule \ssartconf on \ref{ssa:sec:1}, \ref{ssa:sec:2} and \ref{ssa:sec:9}: 
    \la{
      \tossartconf{\ssaenvs}
          {\quadruple{\srcSignatures}{\srcStore}{\cdot}{\srcHeap}}{\srcBody_0}
          {\pair{\ssaSignatures}{\ssaHeap}}{\expr_0}
      \label{ssa:sec:13}
    }

    We examine cases on the configuration of $\pair{\ssaRtState}{\expr_0}$:
    \begin{itemize}

      \item 
        Case $\pair{\ssaRtState}{\expr_0}$ is a \emph{terminal} configuration,
        so there exists $\val$ \st: 
        \la{
          \expr_0 \equiv \val
        }

        Fact \ref{ssa:sec:13} becomes:
        \la{
          \tossartconf{\ssaenvs}
              {\quadruple{\srcSignatures}{\srcStore}{\cdot}{\srcHeap}}{\srcBody_0}
              {\pair{\ssaSignatures}{\ssaHeap}}{\val}
          \label{ssa:sec:12}
        }

        By Lemma~\ref{lemma:value:monotonicity} on \ref{ssa:sec:12}:
        \la{
          \stepsStarFour{\quadruple{\srcSignatures}{\srcStore}{\cdot}{\srcHeap}}
                        {\srcBody_0}
                        {\quadruple{\srcSignatures}{\srcStore}{\cdot}{\srcHeap}}
                        {\srcReturn{\srcVal}}
          \label{ssa:sec:14}
          \\
          \tossartconf{\ssaenvs}
                      {\quadruple{\srcSignatures}{\srcStore}{\cdot}{\srcHeap}}
                      {\srcReturn{\srcVal}}
                      {\ssaRtState}{\val}
          \label{ssa:sec:15}
        }

        By Lemma~\ref{lemma:top:level:reduction} on \ref{ssa:sec:14}:
        \la{
          \stepsStarFour{\quadruple{\srcSignatures}{\srcStore}{\srcStack}{\srcHeap}}
                        {\srcBody_0}
                        {\quadruple{\srcSignatures}{\srcStore}{\srcStack}{\srcHeap}}
                        {\srcReturn{\srcVal}}
          \label{ssa:sec:45}
        }

        By inverting Rule \ssartconf on \ref{ssa:sec:15}:
        \la{
          \tossartstackterm{\srcHeap}{\ssaenvs}{\srcStore}{\cdot}{\srcReturn{\srcVal}}{\val}
          \label{ssa:sec:18}
        }

        By applying Rule \ssastackcons on \ref{ssa:sec:18} and \ref{ssa:sec:10}:
        \la{
          \tossartstackterm{\srcHeap}{\ssaenvs}{\srcStore}
              {\parens{\stackCons{\srcStack_0}{\toStack{\srcStore_0}{\srcEvalCtx_0}}}}
              {\srcReturn{\srcVal}}{\idxEctx{\ssaEvalCtx_0}{\val}}
          \label{ssa:sec:19}
        }

        By applying Rule \ssartconf on \ref{ssa:sec:1}, \ref{ssa:sec:2} and \ref{ssa:sec:19}: 
        \la{
          \tossartconf{\ssaenvs}
            {\quadruple{\srcSignatures}{\srcStore}
              {\parens{\stackCons{\srcStack_0}{\toStack{\srcStore_0}{\srcEvalCtx_0}}}}
                {\srcHeap}}
            {\srcReturn{\srcVal}}
            {\pair{\ssaSignatures}{\ssaHeap}}{\idxEctx{\ssaEvalCtx_0}{\val}}
          \label{ssa:sec:21}
        }

        By applying Rule \srcRcRet on on THe left-hand side of \ref{ssa:sec:21}:
        \la{
          \stepsFour
            {\quadruple
              {\srcSignatures}{\srcStore}
              {\parens{\stackCons{\srcStack_0}{\toStack{\srcStore_0}{\srcEvalCtx_0}}}}
              {\srcHeap}}
            {\srcReturn{\srcVal}}
            {\quadruple{\srcSignatures}{\srcStore_0}{\srcStack_0}{\srcHeap}}
            {\idxEctx{\srcEvalCtx_0}{\srcVal}}
          \label{ssa:sec:20}
        }

        By inverting \ssaexpstackemp and \ssartmethbody on \ref{ssa:sec:18}:
        \la{
          \tossartterm{\srcHeap}{\ssaenvs}{\srcStore}{\srcVal}{\val}
          \label{ssa:sec:26}
        }

        By inverting Rule~\ssartval on \ref{ssa:sec:26}:
        \la{
          \tossaval{\srcHeap}{\srcVal}{\val}
          \label{ssa:sec:59}
        }
        
        By applying Rule~\ssartval on \ref{ssa:sec:59} using $\srcStore_0$:
        \la{
          \tossartterm{\srcHeap}{\ssaenvs}{\srcStore_0}{\srcVal}{\val}
          \label{ssa:sec:22}
        }

        By applying Lemma~\ref{lemma:ssa:ectx:app} on \ref{ssa:sec:10} and \ref{ssa:sec:22}:
        \la{
          \tossartstackterm{\srcHeap}{\ssaenvs}{\srcStore_0}{\srcStack_0}
            {\idxEctx{\srcEvalCtx_0}{\srcVal}}
            {\idxEctx{\ssaEvalCtx_0}{\val}}
          \label{ssa:sec:27}
        }

        By applying Rule \ssartconf on \ref{ssa:sec:1}, \ref{ssa:sec:2} and \ref{ssa:sec:27}: 
        \la{
          \tossartconf{\ssaenvs}
            {\quadruple{\srcSignatures}{\srcStore_0}{\srcStack_0}{\srcHeap}}
            {\idxEctx{\srcEvalCtx_0}{\srcVal}}
            {\pair{\ssaSignatures}{\ssaHeap}}
            {\idxEctx{\ssaEvalCtx_0}{\val}}
          \label{ssa:sec:38}
        }

        Because of \ref{ssa:sec:25}:
        \la{
          \validConf{\pair{\quadruple{\srcSignatures}{\srcStore_0}{\srcStack_0}{\srcHeap}}
                          {\idxEctx{\srcEvalCtx_0}{\srcVal}}}
          \label{ssa:sec:28}
        }

        By induction hypothesis using \ref{ssa:sec:38}, \ref{ssa:sec:b} and \ref{ssa:sec:28}:
        \la{
          \stepsStarFour
            {\quadruple{\srcSignatures}{\srcStore_0}{\srcStack_0}{\srcHeap}}
            {\idxEctx{\srcEvalCtx_0}{\srcVal}}
            {\srcRtState'}{\srcExprOrBd'}
          \label{ssa:sec:23}
          \\
          \tossartconf{\ssaenvs}{\srcRtState'}{\srcExprOrBd'}{\ssaRtState'}{\expr'}
          \label{ssa:sec:24}
          \\
          \validConf{\pair{\srcRtState'}{\srcExprOrBd'}}
          \label{ssa:sec:29}
        }

        We prove \ref{ssa:sec:d} by \ref{ssa:sec:45}, \ref{ssa:sec:20} and \ref{ssa:sec:30}; 
                 \ref{ssa:sec:e} by \ref{ssa:sec:24}; and 
                 \ref{ssa:sec:f} by \ref{ssa:sec:29}.

      \item Case $\pair{\ssaRtState}{\expr_0}$ is a \emph{non-terminal} configuration,
        so there exists $\expr_0'$ \st: 
        \la{
          \stepsFour{\ssaRtState}{\expr_0}{\ssaRtState'}{\expr_0'}
          \label{ssa:sec:30}
        }

        By Rule~\rcevalctx using \ref{ssa:sec:30}: 
        \la{
          \stepsFour{\ssaRtState}{\idxEctx{\ssaEvalCtx_0}{\expr_0}}
                    {\ssaRtState'}{\idxEctx{\ssaEvalCtx_0}{\expr_0'}}
          \label{ssa:sec:35}
        }

        By Lemma~\ref{lemma:ssa:empty:stack:cons:full} using 
        \ref{ssa:sec:13} and \ref{ssa:sec:30}:
        \la{
          \stepsStarFour{\quadruple{\srcSignatures}{\srcStore}{\cdot}{\srcHeap}} 
                        {\srcBody_0}{\srcRtState'}{\srcExprOrBd'}
          \label{ssa:sec:31}
          \\
          \tossartconf{\ssaenvs}{\srcRtState'}{\srcExprOrBd'}{\ssaRtState'}{\expr_0'}
          \label{ssa:sec:32}
        }

        And we examine cases on the form of $\srcBody_0$ for the last result of the above lemma:
        \begin{itemize}

          \item Case $\srcBody_0 \equiv 
                      \idxEctx{\srcEvalCtx}{\srcMethcall{\srcLoc}{\srcMethodname}{\many{\srcVal}}}$.
            It holds that: 
            \la{
              \pair{\srcRtState'}{\srcExprOrBd'} \equiv 
              \pair{\quadruple{\srcSignatures}{\srcStore_1}{\parens{\toStack{\srcStore}{\srcEvalCtx}}}{\srcHeap}}{\srcBody_1}
              \label{ssa:sec:47}
            }

            So \ref{ssa:sec:32} becomes:
            \la{
              \tossartconf{\ssaenvs}
                {\quadruple{\srcSignatures}{\srcStore_1}{\parens{\toStack{\srcStore}{\srcEvalCtx}}}{\srcHeap}}
                {\srcBody_1}{\ssaRtState'}{\expr_0'}
              \label{ssa:sec:33}
            }

            By inverting \ssartconf on \ref{ssa:sec:33}:
            \la{
              \tossartstackterm{\srcHeap}{\ssaenvs}{\srcStore_1}
                {\parens{\toStack{\srcStore}{\srcEvalCtx}}}
                {\srcBody_1}{\expr_0'}
              \label{ssa:sec:39}
            }

            By Lemma \ref{lemma:ssa:ectx:app:gen}
            using \ref{ssa:sec:10} and \ref{ssa:sec:39}:
            \la{
              \tossartstackterm{\srcHeap}{\ssaenvs}{\srcStore_1}
                {
                  \parens{
                    \stackCons
                      {\stackCons{\srcStack_0}{\toStack{\srcStore_0}{\srcEvalCtx_0}}}
                      {\toStack{\srcStore}{\srcEvalCtx}}
                  }
                }
                {\srcBody_1}{\idxEctx{\ssaEvalCtx_0}{\expr_0'}}
              \label{ssa:sec:40}
            }

            Let:
            \la{
              \srcStack' \equiv 
                    \stackCons
                      {\stackCons{\srcStack_0}{\toStack{\srcStore_0}{\srcEvalCtx_0}}}
                      {\toStack{\srcStore}{\srcEvalCtx}}
              \label{ssa:sec:41}
            }

            By applying Rule \ssartconf on \ref{ssa:sec:1}, \ref{ssa:sec:2} and \ref{ssa:sec:40}: 
            \la{
              \tossartconf{\ssaenvs}
                {\quadruple{\srcSignatures}{\srcStore_1}{\srcStack'}{\srcHeap}}
                {\srcBody_1}
                {\ssaRtState'}
                {\idxEctx{\ssaEvalCtx_0}{\expr_0'}}
              \label{ssa:sec:42}
            }

            By Lemma~\ref{lemma:top:level:reduction} on \ref{ssa:sec:31}:
            \la{
              \stepsStarFour{\quadruple{\srcSignatures}{\srcStore}
                                       {\srcStack}
                                       {\srcHeap}}{\srcBody_0}
                            {\quadruple{\srcSignatures}{\srcStore_1}
                                       {\srcStack'}
                                       {\srcHeap}}{\srcBody_1}
              \label{ssa:sec:46}
            }
        
            We prove \ref{ssa:sec:d}, \ref{ssa:sec:e} and \ref{ssa:sec:f} by 
            \ref{ssa:sec:46}, \ref{ssa:sec:42} and \ref{ssa:sec:47}, respectively.

          \item For all remaining cases on $\srcBody_0$:
            \la{
              \srcHeap' \supseteq \srcHeap
              \label{ssa:sec:48}
              \\
              \srcStore' \supseteq \srcStore
              \label{ssa:sec:49}
            }
            Because of \ref{ssa:sec:25}, it holds that: 
            \la{
              \pair{\srcRtState'}{\srcExprOrBd'} \equiv 
              \pair{\quadruple{\srcSignatures}{\srcStore'}{\cdot}{\srcHeap'}}{\srcBody'}
              \label{ssa:sec:55}
            }

            By inverting Rule \ssartconf on \ref{ssa:sec:32}:
            \la{
              \tossaheap{\srcRtState'}{\srcHeap'}{\ssaHeap'}
              \label{ssa:sec:52}
            }

            By Lemma~\ref{lemma:top:level:reduction} on \ref{ssa:sec:31}:
            \la{
              \stepsStarFour{\quadruple{\srcSignatures}{\srcStore}{\srcStack}{\srcHeap}} 
                            {\srcBody_0}
                            {\quadruple{\srcSignatures}{\srcStore'}{\srcStack}{\srcHeap'}}
                            {\srcBody'}
              \label{ssa:sec:53}
            }

            Fact \ref{ssa:sec:32} becomes:
            \la{
              \tossartconf{\ssaenvs}
                {\quadruple{\srcSignatures}{\srcStore'}{\cdot}{\srcHeap'}}
                {\srcBody'}{\ssaRtState'}{\expr_0'}
              \label{ssa:sec:36}
            }

            By inverting \ssartconf on \ref{ssa:sec:36}:
            \la{
              \tossartstackterm
              {\srcHeap'}
              {\ssaenvs}
              {\srcStore'}{\cdot}{\srcBody'}{\expr_0'}
              \label{ssa:sec:37}
            }

            By applying Lemma~\ref{lemma:ssa:ectx:heap:weaken} on
            \ref{ssa:sec:10} using \ref{ssa:sec:48}:
            \la{
              \tossartstackterm{\srcHeap'}{\ssaenvs}
                  {\srcStore_0}{\srcStack_0}{\srcEvalCtx_0}
                  {\ssaEvalCtx_0}
              \label{ssa:sec:50}
            }

            By applying rule \ssastackcons on \ref{ssa:sec:10} and \ref{ssa:sec:37}:
            \la{
              \tossartstackterm{\srcHeap'}{\ssaenvs}
                {\srcStore'}
                {\parens{\stackCons{\srcStack_0}{\toStack{\srcStore_0}{\srcEvalCtx_0}}}}
                {\srcBody'}{\idxEctx{\ssaEvalCtx_0}{\expr_0'}}
              \label{ssa:sec:51}
            }

            By applying rule \ssartconf on \ref{ssa:sec:1}, \ref{ssa:sec:52} and
            \ref{ssa:sec:51}: 
            \la{
              \tossartconf
                {\ssaenvs}
                {\quadruple{\srcSignatures}{\srcStore'}{\srcStack}{\srcHeap'}}
                {\srcBody'}
                {\ssaRtState'}
                {\idxEctx{\ssaEvalCtx_0}{\expr_0'}}
              \label{ssa:sec:54}
            }

            We prove \ref{ssa:sec:d}, \ref{ssa:sec:e} and \ref{ssa:sec:f} by 
            \ref{ssa:sec:53}, \ref{ssa:sec:54} and \ref{ssa:sec:55},
            respectively.

        \end{itemize}

    \end{itemize}

\end{itemize}

\end{proof}

\begin{theorem}[Forward Simulation]\label{theorem:consistency:appendix}
If
$\tossartconfshort{\ssaenvs}{\srcRtConf}{\ssaRtConf}$,
then:
\begin{enumerate}[label=(\alph*)]
  \setlength\itemsep{1pt}
  \item if $\ssaRtConf$ is terminal, then there exists $\srcRtConf'$ \st
        $\stepsmany{\srcRtConf}{\srcRtConf'}$
        and 
        $\tossartconfshort{\ssaenvs}{\srcRtConf'}{\ssaRtConf}$
        \label{theorem:consistency:appendix:a}
  \item if 
        $\steps{\ssaRtConf}{\ssaRtConf'}$,
        then there exists $\srcRtConf'$ \st
        $\stepsmany{\srcRtConf}{\srcRtConf'}$
        and 
        $\tossartconfshort{\ssaenvs}{\srcRtConf'}{\ssaRtConf'}$
        \label{theorem:consistency:appendix:b}
\end{enumerate}
\end{theorem}
\begin{proof}

Part \ref{theorem:consistency:appendix:a} is proven by use of
by Lemma~\ref{lemma:value:monotonicity}, and 
part \ref{theorem:consistency:appendix:b} 
by Lemma~\ref{lemma:ssa:expr:consistency:full}.

\end{proof}

\clearpage

\subsection{Type Safety}

\begin{lemma}[Substitution Lemma]\label{lemma:substitution}
If
\begin{enumerate}[label=(\alph*)]
  \item $\typecheck{\env}{\many{\exprb}}{\many{\rtypeb}}$
  \item $\issubtype{\envextex{\env}{\many{\evar}}{\many{\rtypeb}}}{\many{\rtypeb}}{\many{\rtypeb}'}$
  \item $\typecheck{\envextex{\env}{\many{\evar}}{\many{\rtypeb}'}}{\expr}{\rtype}$
\end{enumerate}
then
\begin{enumerate}
  \item[]
$\wfconstr{\env}
  {\mconcatenate{
    \varbinding{\appsubst{\tsubst{\many{\exprb}}{\many{\evar}}}{\expr}}{\rtypec}
  }{
    \rtypec \subt \rtype
  }}$
  \end{enumerate}
\end{lemma}
\begin{proof}
By induction on the derivation of the statement 
$\typecheck{\envextex{\env}{\many{\evar}}{\many{\rtypeb}}}{\expr}{\rtype}$.
\end{proof}

\begin{lemma}[Environment Substitution]\label{lemma:subst:env}
If 
$\typecheck
  {\mconcatenate{\envext{\env_1}{\evar}{\rtype}}{\env_2}}{\exprb}{\rtypeb}$,
then
$\typecheck
  {\mconcatenate{\envext{\env_1}{\evar}{\rtype}}{\appsubst{\esubst{\evarb}{\evar}}{\env_2}}}
  {\appsubst{\esubst{\evarb}{\evar}}{\exprb}}
  {\appsubst{\esubst{\evarb}{\evar}}{\rtypeb}}$.
\end{lemma}
\begin{proof} Straightforward. \end{proof}

\begin{lemma}[Weakening Subtyping]\label{lemma:subtype:weaken}
  If
  $\issubtype{\env}{\rtypeb}{\rtype}$,
  then
  $\issubtype{\envextex{\env}{\evar}{\rtypec}}{\rtypeb}{\rtype}$.
\end{lemma}
\begin{proof} Straightforward. \end{proof}

\begin{lemma}[Weakening Typing]\label{lemma:typing:weaken}
If
$\typecheck{\env}{\expr}{\rtype}$,
then for  $\env' \supseteq \env$,
it holds that 
$\typecheck{\env'}{\expr}{\rtype}$.
\end{lemma}
\begin{proof} Straightforward. \end{proof}


\begin{lemma}[Store Type]\label{lemma:store:type}
If 
$\newwfstore{\storety}{\ssaHeap}$, 
$\idxMapping{\ssaHeap}{\loc} = \ssaObject$
and
$\idxMapping{\storety}{\loc} = \rtype$,
then
$\rttypecheck{\storety}{\ssaHeap}{\ssaObject}{\rtypeb, \rtype \leq \rtypeb}$.
\end{lemma}
\begin{proof} Straightforward. \end{proof}

\begin{lemma}[Method Body Type -- Lemma~A.3 from \cite{Nystrom2008}]\label{lemma:method:body:type}
If
  \begin{enumerate}[label=(\alph*)]
    \item \emph{$\wfconstr{\envextex{\env}{\evarb}{\rtype}}
      {\has{\evarb}{\parens{\methdef{\methodname}{\evarb}{\rtypec}{\pred}{\rtypeb}{\expr}}}}$}

  \item $\issubtype{\envextex{\envextex{\env}{\evarb}{\rtype}}{\many{\evarb}}{\many{\rtype}}}
                   {\many{\rtype}}{\many{\rtypec}}$
\end{enumerate}
Then for some type $\rtypeb'$ it is the case that:
\begin{enumerate}
  \item[]
$
\wfconstr{\envextex{\envextex{\env}{\evarb}{\rtype}}{\many{\evarb}}{\many{\rtype}}}
         {
           \mconcatenate{\varbinding{\expr}{\rtypeb'}}
                        {\rtypeb' \subt \rtypeb}
         }
$
\end{enumerate}

\end{lemma}
\begin{proof} Straightforward. \end{proof}

\begin{lemma}[Cast]{\label{lemma:cast}}
If
    $\newwfstore{\storety}{\ssaHeap}$ and 
    $\dwfconstr{\env}{\storety}{
      \envbinding{\loc}{\rtypeb},
      \rtypeb \semsubt \rtype
    }$,
then
$\dwfconstr{\env}{\storety}{
      \envbinding{\idxMapping{\ssaHeap}{\loc}}{\rtypec},
      \rtypec \subt \rtype
    }$.
\end{lemma}
\begin{proof}
  Straightforward.
\end{proof}

\begin{lemma}[Evaluation Context Typing]{\label{lemma:eval:ctx:typing}}
If
$\typecheck{\env}{\idxEctx{\ssaEvalCtx}{\expr}}{\rtype}$,
then for some type $\rtypeb$ it holds that
$\typecheck{\env}{\expr}{\rtypeb}$.
\end{lemma}
\begin{proof}
By induction on the structure of the evaluation context $\ssaEvalCtx$.
\end{proof}

\begin{lemma}[Evaluation Context Step Typing]{\label{lemma:eval:ctx:step:typing}}
If
\begin{itemize}
  \item[] 
    $\dtypechecktwo{\env}{\storety}{\idxEctx{\ssaEvalCtx}{\expr}}{\rtype}{\expr}{\rtypeb}$
\end{itemize}
and for some expression $\expr'$ and heap typing $\storety' \supseteq \storety$ it holds that
\begin{itemize}
  \item[] 
    $\dwfconstr{\env}{\storety'}{\mconcatenate{\varbinding{\expr'}{\rtypeb'}}{\rtypeb' \semsubt \rtypeb}}$
\end{itemize}
then
\begin{itemize}
  \item[] 
    $\dwfconstr{\env}{\storety'}
        {\mconcatenate
          {\varbinding{\idxEctx{\ssaEvalCtx}{\expr'}}{\rtype'}}
          {\rtype' \semsubt \rtype}
        }$
\end{itemize}
\end{lemma}
\begin{proof}
By induction on the structure of the evaluation context $\ssaEvalCtx$.
\end{proof}

\begin{lemma}[Selfification]{\label{lemma:self}}
If
$\issubtype{\envextex{\env}{\evar}{\rtypeb}}{\rtypeb}{\rtype}$
then
$\issubtype{\envextex{\env}{\evar}{\rtypeb}}{\rtypeb}{\singleton{\rtype}{\evar}}$.
\end{lemma}
\begin{proof} Straightforward. \end{proof}

\begin{lemma}[Existential Weakening]{\label{lemma:exist:weaken}}
If
$\issubtype{\env}{\rtypec}{\rtypec'}$
then
$\issubtype{\env}
{\texist{\evar}{\rtypec}{\rtype}}
{\texist{\evar}{\rtypec'}{\rtype}}
$.
\end{lemma}
\begin{proof} Straightforward. \end{proof}


\begin{lemma}[Boolean Facts]{\label{lemma:bool:fact}}
If
\begin{enumerate}[label=(\alph*)]
  \item 
    $\typechecksubt{\env}{\evar}{\rtype}{\rtype}{\reftp{\tbool}{\vv = \ssaTrue}}$
  \item 
    $\typechecksubt{\envgrdext{\env}{\evar}}{\expr}{\rtypeb}{\rtypeb}{\rtype}$
\end{enumerate}
then
\begin{itemize}
  \item[] 
    $\typechecksubt{\env}{\expr}{\rtypeb}{\rtypeb}{\rtype}$
\end{itemize}
\end{lemma}
\begin{proof} Straightforward. \end{proof}

\begin{theorem}[Subject Reduction]\label{theorem:subj:reduc}
If
\begin{enumerate}[label=(\alph*)]
  \item $\dtypecheck{\env}{\storety}{\expr}{\rtype}$                    \label{sr:a}
  \item $\stepsFour{\ssaRtState}{\expr}{\ssaRtState'}{\expr'}$          \label{sr:b}
  \item $\newwfstore{\storety}{\accessState{\ssaRtState}{\ssaHeap}}$ \label{sr:c}
\end{enumerate}
then for some $\rtype'$ and $\storety' \supseteq \storety$:
\begin{enumerate}[label=(\alph*)]
\setcounter{enumi}{3}
  \item $\dtypecheck{\env}{\storety'}{\expr'}{\rtype'}$   \label{sr:d}
  \item $\semsubtype{\env}{\rtype'}{\rtype}$              \label{sr:e}
  \item $\newwfstore{\storety'}{\ssaHeap'}$.           \label{sr:f}
\end{enumerate}
\end{theorem}

\begin{proof}
We proceed by induction on the structure of fact \ref{sr:b}:
$$
\stepsFour{\ssaRtState}{\expr}{\ssaRtState'}{\expr'}
$$
We have the following cases:
\begin{itemize}


  \item \brackets{\rcevalctx}: Fact \ref{sr:b} has the form:
    \la{
      \stepsFour{\ssaRtState }{\idxEctx{\ssaEvalCtx}{\expr_0}}
                {\ssaRtState'}{\idxEctx{\ssaEvalCtx}{\expr_0'}}
      \label{sr:ectx:30}
    }
    From \ref{sr:a}:
    \la{
      \dtypecheck{\env}{\storety}{\idxEctx{\ssaEvalCtx}{\expr_0}}{\rtype}  \label{sr:ectx:1}
    }
    By Lemma~\ref{lemma:eval:ctx:typing} on \ref{sr:ectx:1}:
    \la{
      \dtypecheck{\env}{\storety}{\expr_0}{\rtype_0}  
      \label{sr:ectx:2}
    }

    By inverting Rule \rcevalctx on \ref{sr:ectx:30}:
    \la{
      \stepsFour{\ssaRtState }{\expr_0}
                {\ssaRtState'}{\expr_0'}
      \label{sr:ectx:31}
    }

    By induction hypothesis, using \ref{sr:ectx:2}, \ref{sr:ectx:31} and \ref{sr:c} we get:
    \la{
      \dtypecheck{\env}{\storety'}{\expr_0'}{\rtype_0'}      \label{sr:ectx:d}
      \\
      \dsemsubtype{\env}{\storety'}{\rtype_0'}{\rtype_0}                 \label{sr:ectx:e}
      \\
      \newwfstore{\storety'}{\accessState{\ssaRtState'}{\ssaHeap}}                 \label{sr:ectx:f}
      \\
      \storety' \supseteq \storety      \label{sr:ectx:10}
    }

    For some type $\rtype_0'$ and heap $\accessState{\ssaRtState'}{\ssaHeap}$.
    \\\\
    From \ref{sr:ectx:f} we prove \ref{sr:f}.
    \\\\
    By Lemma~\ref{lemma:eval:ctx:step:typing} using \ref{sr:ectx:1}, \ref{sr:ectx:2},
    \ref{sr:ectx:d}, \ref{sr:ectx:e} and \ref{sr:ectx:10}:
    \la{
      \dwfconstr{\env}{\storety'}
        {\mconcatenate{\varbinding{\idxEctx{\ssaEvalCtx}{\expr_0'}}{\rtype'}}{\rtype' \semsubt \rtype}}
        \label{sr:ectx:7}
    }
    From \ref{sr:ectx:7} we prove \ref{sr:d} and \ref{sr:e}.


  \item \brackets{\rfield}: Fact \ref{sr:b} has the form:
    \la{
      \stepsFour{\ssaRtState}{\dotref{\loc}{\fieldnamec}}{\ssaRtState}{\val}      
      \label{sr:rfield:3}
    }
    By Fact~\ref{sr:a} for $\expr \equiv \dotref{\loc}{\fieldnamec}$ we have:
    \la{
      \dtypecheck{\env}{\storety}{\dotref{\loc}{\fieldnamec}}{\rtype}     
      \label{sr:rfield:1}
    }

    By inverting \rfield on \ref{sr:rfield:3}:
    \la{
      \idxMapping{\accessState{\ssaRtState}{\ssaHeap}}{\loc} \equiv \ssaObject =
        \heapObject{\loc'}{\fieldDefsSym}
      \label{sr:rfield:4}
      \\
      \fieldDef{\fieldname}{\val} \in \fieldDefsSym 
    }


    By inverting \hbndinst on \ref{sr:c} for location $\loc$:
    \la{
      \fieldDefsSym \defeq \mconcatenate
        {\immFieldDefs{\fieldname }{\val}_{\immindex}}
        {\mutFieldDefs{\fieldnameb}{\val}_{\mutindex}}
      \label{sr:rfield:56}
      \\
      \basetype{\idxMapping{\storety}{\loc}} = \cname
      \label{sr:rfield:40}
      \\
      \wfconstr{\envext{\env}{\evarb}{\cname}}
               {\fieldsOfVar{\evarb} = 
                  \mconcatenate{\immfieldbindings{\fieldname}{\rtypec}}
                               {\mutfieldbindings{\fieldnameb}{\rtyped}}
               }
      \label{sr:rfield:41}
      \\
      \rttypecheckval{\storety}{\many{\val}_{\immindex}}{\many{\rtype}_{\immindex}}
      \label{sr:rfield:7}
      \\
      \rttypecheckval{\storety}{\many{\val}_{\mutindex}}{\many{\rtype}_{\mutindex}}
      \\
      \wfconstr{
          \envextex
            {\envext{\env}{\evarb}{\cname}}
            {\many{\evarb}_{\immindex}}
          {\singleton{\many{\rtype}_{\immindex}}{\tdotref{\evarb}{\many{\fieldname}}}}
        }
        {
          \mconcatenate{
            \mconcatenate{\many{\rtype}_{\immindex} \subt \many{\rtypec}}
            {\many{\rtype}_{\mutindex} \subt \many{\rtyped}}
          }
          {\classinv{\cname}{\evarb}}
        }
      \label{sr:rfield:10}
    }

    By applying \rtchkobj on \ref{sr:rfield:40}, \ref{sr:rfield:56} and
    \ref{sr:rfield:7}:
    \la{
      \dtypecheck{\env}{\storety}{\ssaObject}{\rtypeb'}
    }

    Where:
    \la{
      \rtypeb' \equiv \texist
            {\many{\evarb}_{\immindex}}
            {\many{\rtype}_{\immindex}}
            {
              \reftp{\cname}
                    {
                      \concatpreds
                        {\tdotref{\vv}{\many{\fieldname}} = \many{\evarb}_{\immindex}}
                        {\classinv{\cname}{\vv}}
                    }
            }     
      \label{sr:rfield:42}
    }

    By Lemma~\ref{lemma:store:type} using \ref{sr:c}, \ref{sr:rfield:4} and
    \ref{sr:rfield:40}:
    \la{
      \issubtype{\env}{\rtypeb}{\rtypeb'}
        \label{sr:rfield:43}
    }

    Where:
    \la{
      \idxMapping{\storety}{\loc} = \rtypeb
    }

    We examine cases on the typing statement \ref{sr:rfield:1}:

    \begin{itemize}

      \item \brackets{\lqchkfieldimm}:
        Field $\fieldnamec$ is an immutable field $\fieldname_{\index}$, 
        so fact \ref{sr:rfield:1} becomes:
        \la{
          \dtypecheck{\env}{\storety}
                    {\dotref{\loc}{\fieldname_{\index}}}
                    {\texist{\evarb}{\rtypeb}{\singleton{\rtypec_{\index}}{\tdotref{\evarb}{\fieldname_{\index}}}}}   \label{sr:rfield:2}
        }

        By inverting \lqchkfieldimm on \ref{sr:rfield:2}:
        \la{
          \rttypecheckval{\storety}{\loc}{\rtypeb}
          \label{sr:rfield:57}
          \\
          \dwfconstr{\envextex{\env}{\evarb}{\rtypeb}}{\storety}
                    {\varbinding
                      {\hasimm{\evarb}{\fieldname_{\index}}}
                    {\rtypec_{\index}}}
        }
        For a fresh $\evarb$.
        \\\\
        Keeping only the relevant part of \ref{sr:rfield:7} and \ref{sr:rfield:10}:
        \la{
          \dtypecheck{\env}{\storety}{\val_{\index}}{\rtype_{\index}}     \label{sr:rfield:20}
          \\
          \dwfconstr{
            \envextex{\envext{\env}{\evarb}{\cname}}
                     {\many{\evarb}_{\immindex}}
                     {\singleton{\many{\rtype}_{\immindex}}{\tdotref{\evarb}{\many{\fieldname}}}}
            }
            {\storety}
            {\rtype_{\index} \subt \rtypec_{\index}}     
            \label{sr:rfield:11}
        }
        By \ref{sr:rfield:20} we prove \ref{sr:d}.
        \\\\
        By Lemma~\ref{lemma:self} using \ref{sr:rfield:11} and picking
        $\evarb_{\index}$ as the selfification variable:
        \la{
           \dwfconstr{
            \envextex{\envext{\env}{\evarb}{\cname}}
                     {\many{\evarb}_{\immindex}}
                     {\singleton{\many{\rtype}_{\immindex}}{\tdotref{\evarb}{\many{\fieldname}}}}
                   }
            {\storety}
            {\rtype_{\index} \subt \singleton{\rtypec_{\index}}{\evarb_{\index}}}     \label{sr:rfield:12}
        }
        For the above environment it holds that:
        \la{
          \embed{
              \envextex{\envext{\env}{\evarb}{\cname}}
                       {\many{\evarb}_{\immindex}}
                       {\singleton{\many{\rtype}_{\immindex}}{\tdotref{\evarb}{\many{\fieldname}}}};
                       \storety
          }
          \imp
          \evarb_{\index} = \tdotref{\evarb}{\fieldname_{\index}}
          \label{sr:rfield:13}
        }

        By \subrefl and By Lemma~\ref{lemma:self} using \ref{sr:rfield:13}:
        \la{
           \dwfconstr{
            \envextex{\envext{\env}{\evarb}{\cname}}
                     {\many{\evarb}_{\immindex}}
                     {\singleton{\many{\rtype}_{\immindex}}{\tdotref{\evarb}{\many{\fieldname}}}}
            }
            {\storety}
            {
              \singleton{\rtypec_{\index}}{\evarb_{\index}}
              \subt 
              \singleton{\singleton{\rtypec_{\index}}{\evarb_{\index}}}
                        {\tdotref{\evarb}{\fieldname_{\index}}}
            }     \label{sr:rfield:14}
        }

        By simplifying \ref{sr:rfield:14} using \subtrans on \ref{sr:rfield:12} and \ref{sr:rfield:14} we get:
        \la{
           \dwfconstr{
            \envextex{\envext{\env}{\evarb}{\cname}}
                     {\many{\evarb}_{\immindex}}
                     {\singleton{\many{\rtype}_{\immindex}}{\tdotref{\evarb}{\many{\fieldname}}}}
            }
            {\storety}
            {\rtype_{\index} \subt \singleton{\rtypec_{\index}}{\tdotref{\evarb}{\fieldname_{\index}}}}
            \label{sr:rfield:35}
        }

        By \ref{sr:rfield:35} it also holds that:
        \la{
          \dissubtype{\envextex{\env}{\evarb}
            {\texist{\many{\evarb}_{\immindex}}
                    {\singleton{\many{\rtype}_{\immindex}}{\tdotref{\evarb}{\many{\fieldname}}}}
                    {\cname}
            }}
            {\storety}
            {\rtype_{\index}}
            {\singleton{\rtypec_{\index}}{\tdotref{\evarb}{\fieldname_{\index}}}}
            \label{sr:rfield:36}
        }

        By \ref{sr:rfield:36} it also holds that:
        \la{
          \dissubtype{\envextex{\env}{\evarb}
            {\texist{\many{\evarb}_{\immindex}}{\many{\rtype}_{\immindex}}
                    {\singleton{\cname}{\many{\evarb}_{\immindex}}}
            }}
            {\storety}
            {\rtype_{\index}}
            {\singleton{\rtypec_{\index}}{\tdotref{\evarb}{\fieldname_{\index}}}}
            \label{sr:rfield:50}
        }

        By expanding \ref{sr:rfield:50} and \ref{sr:rfield:10}:
        \la{
          \dissubtype{\envextex{\env}{\evarb}
            {\texist
              {\many{\evarb}_{\immindex}}
              {\many{\rtype}_{\immindex}}
              {\reftp{\cname}
                    {\concatpreds
                        {\tdotref{\vv}{\many{\fieldname}} = \many{\evarb}_{\immindex}}
                        {\classinv{\cname}{\vv}}}
              }}}     
            {\storety}
            {\rtype_{\index}}
            {\singleton{\rtypec_{\index}}{\tdotref{\evarb}{\fieldname_{\index}}}}
            \label{sr:rfield:51}
        }

        By using \ref{sr:rfield:42} on \ref{sr:rfield:51}:
        \la{
          \dissubtype{\envextex{\env}{\evarb}{\rtypeb'}}{\storety}
                    {\rtype_{\index}}
                    {\singleton{\rtypec_{\index}}{\tdotref{\evarb}{\fieldname_{\index}}}}
            \label{sr:rfield:52}
        }

        By Lemma~\ref{lemma:subtype:weaken} using \ref{sr:rfield:52} and
        \ref{sr:rfield:43}:
        \la{
          \dissubtype{\envextex{\env}{\evarb}{\rtypeb}}{\storety}
                    {\rtype_{\index}}
                    {\singleton{\rtypec_{\index}}{\tdotref{\evarb}{\fieldname_{\index}}}}
            \label{sr:rfield:44}
        }

        From Rule~\subwitness using \ref{sr:rfield:44}:
        \la{
          \dissubtype{\env}{\storety}
                    {\rtype_{\index}}
                    {\texist{\evarb}{\rtypeb}
                    {\singleton{\rtypec_{\index}}{\tdotref{\evarb}{\fieldname_{\index}}}}}     
          \label{sr:rfield:15}
        }

        Using \ref{sr:rfield:2}, \ref{sr:rfield:7} and \ref{sr:rfield:15} we prove \ref{sr:e}.
        \\\\
        Heap $\accessState{\ssaRtState}{\ssaHeap}$ does not evolve so \ref{sr:f} holds trivially.

      \item \brackets{\lqchkfieldmut}:
        Field $\fieldnamec$ is a mutable field $\fieldnameb_{\index}$, so fact \ref{sr:a} becomes:
        \la{
          \dtypecheck{\env}{\storety}
                     {\dotref{\loc}{\fieldnameb_{\index}}}
                     {\texist{\evarb}{\rtypeb}{\rtypee_{\index}}}                     \label{sr:rfield:22}
        }
        By inverting \lqchkfieldmut on \ref{sr:rfield:22}:
        \la{
          \typecheck{\env}{\loc}{\rtypeb}
          \\
          \typecheck{\envextex{\env}{\loc}{\rtypeb}}
                    {\hasmut{\evarb}{\fieldnameb_{\index}}}
                    {\rtyped_{\index}}
        }
        For a fresh $\evarb$.
        \\\\
        Keeping only the relevant parts of \ref{sr:rfield:7} and \ref{sr:rfield:10}:
        \la{
          \typecheck{\env}{\val_{\index}}{\rtype_{\index}}     
          \label{sr:rfield:23}
          \\
          \wfconstr{
            \envextex{\envext{\env}{\evarb}{\cname}}
                     {\many{\evarb}_{\immindex}}
                     {\singleton{\many{\rtype}_{\immindex}}{\tdotref{\evarb}{\many{\fieldname}}}}
            }
            {\rtype_{\index} \subt \rtyped_{\index}} 
          \label{sr:rfield:24}
        }

        By \ref{sr:rfield:23} we prove \ref{sr:d}.
        \\\\
        By similar reasoning as before and using \ref{sr:rfield:24} we get:
        \la{
          \dissubtype{\envextex{\env}{\evarb}{\rtypeb'}}
                     {\storety}
                     {\rtype_{\index}}{\rtyped_{\index}} 
          \label{sr:rfield:45}
        }

        By Lemma~\ref{lemma:subtype:weaken} using \ref{sr:rfield:45} and
        \ref{sr:rfield:43}:
        \la{
          \dissubtype{\envextex{\env}{\evarb}{\rtypeb}}{\storety}
                     {\rtype_{\index}}{\rtyped_{\index}}
            \label{sr:rfield:44}
        }

        By Rule~\subwitness using \ref{sr:rfield:44}:
        \la{
          \dissubtype{\env}{\storety}
            {\rtype_{\index}}
            {\texist{\evarb}{\rtypeb}{\rtyped_{\index}}} 
            \label{sr:rfield:25}
        }

        Using \ref{sr:rfield:22}, \ref{sr:rfield:7} and \ref{sr:rfield:25} we prove \ref{sr:e}.
        \\\\
        Heap $\accessState{\ssaRtState}{\ssaHeap}$ does not evolve so \ref{sr:f} holds trivially.

    \end{itemize}


  \item \brackets{\rinvoke}:
    Fact \ref{sr:b} has the form:
    \la{
      \stepsFour
        {\ssaRtState}
        {\methcall{\loc}{\methodname}{\many{\val}}}
        {\ssaRtState}
        {\appsubst{\tsubsttwo{\many{\val}}{\many{\evarb}}{\loc}{\ethis}}{\expr'}} \label{sr:inv:0}
    }

    By \ref{sr:a} for $\expr \equiv \methcall{\loc}{\methodname}{\many{\val}}$ we have:
    \la{
      \dtypecheck{\env}{\storety}{\methcall{\loc}{\methodname}{\many{\val}}}
      {\texist{\evarb}{\rtype}{\texist{\many{\evarb}}{\many{\rtype}}{\rtypeb}}}    \label{sr:inv:1}
    }

    By inverting \lqchkinv on \ref{sr:inv:1}:
    \la{
      \dtypechecktwo{\env}{\storety}{\loc}{\rtype}{\many{\val}}{\many{\rtype}}     \label{sr:inv:6}
      \\
      \wfconstr{\envextex{\envextex{\env}{\evarb}{\rtype}}{\many{\evarb}}{\many{\rtype}}}
      {\has{\evarb}{\parens{\methdef{\methodname}{\evarb}{\rtypec}{\pred}{\rtypeb}{\expr'}}}}  \label{sr:inv:2}
      \\
      \wfconstr{\envextex{\envextex{\env}{\evarb}{\rtype}}{\many{\evarb}}{\many{\rtype}}}
      {\many{\rtype} \subt \many{\rtypec}}   \label{sr:inv:3}
      \\
      \wfconstr{\envextex{\envextex{\env}{\evarb}{\rtype}}{\many{\evarb}}{\many{\rtype}}}{\pred}
    }

    With fresh $\evarb$ and $\many{\evarb}$.
    \\\\
    By inverting \rinvoke on \ref{sr:inv:0}:
    \la{
      \resolveMeth{\ssaHeap}{\loc}{\methodname}
                  {\parens{\methdef{\methodname}{\evarb}{\rtypec}{\pred}{\rtypeb}{\expr}}}
      \label{sr:inv:8}
      \\
      \eval{\pred} = \ssaTrue
    }
    Note that $\this$ has already been substituted by $\loc$ in $\rtypeb$ and $\pred$.


    By Lemma~\ref{lemma:method:body:type} using \ref{sr:inv:2} and \ref{sr:inv:3}:
    \la{
      \wfconstr{\envextex{\envextex{\env}{\evarb}{\rtype}}{\many{\evarb}}{\many{\rtype}}}
      {\mconcatenate{\varbinding{\expr'}{\rtypeb'}}{\rtypeb' \subt \rtypeb}} \label{sr:inv:4}
    }

    By \ref{sr:inv:4} we prove \ref{sr:d}.
    \\

    By Rule~\subwitness using \ref{sr:inv:4}:
    \la{
      \issubtype{\env}{\rtypeb'}{\texist{\evarb}{\rtype}{\texist{\many{\evarb}}{\many{\rtype}}{\rtypeb}}}
    }

    By Lemma~\ref{lemma:substitution} using \ref{sr:inv:6}, \ref{sr:inv:3} and \ref{sr:inv:4}:
    \la{
      \wfconstr{\env}{
        \mconcatenate
          {\varbinding{\appsubst{\tsubsttwo{\many{\val}}{\many{\evarb}}{\loc}{\ethis}}{\expr'}}{\rtyped}}
          {\rtyped \subt \rtypeb'}
        }
        \label{sr:inv:10}
    }

    By Rule~\subtrans on \ref{sr:inv:4} and \ref{sr:inv:10}:
    \la{
      \issubtype{\env}{\rtyped}{\texist{\evarb}{\rtype}{\texist{\many{\evarb}}{\many{\rtype}}{\rtypeb}}}
      \label{sr:inv:110}
    }

    By \ref{sr:inv:110} we prove \ref{sr:e}.
    \\

    Heap $\accessState{\ssaRtState}{\ssaHeap}$ does not evolve so \ref{sr:f} holds trivially.


  \item \brackets{\rcast}:
    Fact \ref{sr:b} has the form:
    $$
    \stepsFour{\ssaRtState}{\cast{\rtype}{\loc}}{\ssaRtState}{\loc}
    $$
    By \ref{sr:a} for $\expr \equiv \cast{\rtype}{\loc}$ we have:
    \la{
      \dtypecheck{\env}{\storety}{\cast{\rtype}{\loc}}{\rtype}    \label{sr:cast:30}
    }

    By inverting \lqchkcast on \ref{sr:cast:30}:
    \la{
      \dtypecheck{\env}{\storety}{\loc}{\rtypeb}  \label{sr:cast:2} \\
      \wftype{\env}{\rtype}                                         \\
      \semsubtype{\env}{\rtypeb}{\rtype}          \label{sr:cast:1}
    }

    By \ref{sr:cast:2} and \ref{sr:cast:1} we get \ref{sr:d} and \ref{sr:e}, respectively.
    \\

    \accessState{\ssaRtState}{\ssaHeap} does not evolve, which proves
    \ref{sr:f}, given \ref{sr:b}.


  \item \brackets{\rnew}:
    Fact \ref{sr:c} has the form:
    \la{
      \stepsFour{\ssaRtState}{\newexpr{\cname}{\many{\val}}}{\ssaRtState'}{\loc}
      \label{sr:new:16}
    }

    By inverting \rnew on \ref{sr:new:16}:
    \la{
      \idxMapping{\srcHeap}{\loc_0} = \heapClassObject{\cname}{\loc_0'}{\methodDefsSym}
      \label{sr:new:17}
      \\
      \rtFields{\ssaSignatures}{\cname} = \fieldbindings{\fieldname}{\rtype}
      \label{sr:new:18}
      \\
      \ssaObject = \heapObject{\loc_0}{\fieldDefs{\fieldname}{\val}}
      \label{sr:new:19}
      \\
      \ssaHeap' = \updMapping{\ssaHeap}{\loc}{\ssaObject}
      \label{sr:new:20}
    }
    
    By \ref{sr:a} for $\expr \equiv \newexpr{\cname}{\many{\val}}$ we have:
    \la{
      \dtypecheck{\env}{\storety}{\newexpr{\cname}{\many{\val}}}{\rtypec_0}  \label{sr:new:0}
    }
    Where:
    \la{
    \rtypec_0 \equiv
      \texist{\many{\evarb}_{\immindex}}{\many{\rtype}_{\immindex}}
              {\reftp
                {\cname}
                {
                  \concatpreds
                    {\tdotref{\vv}{\many{\fieldname}} = \many{\evarb}_{\immindex}}
                    {\classinv{\cname}{\vv}}
                }
              }
    }

    By inverting \lqchknew on \ref{sr:new:0}:
    \la{
      \typecheck{\env}{\many{\val}}{\ppair{\many{\rtype}_{\immindex}}{\many{\rtype}_{\mutindex}}}
      \label{sr:new:15}
      \\
      \wfconstrnopremise{\isclass{\cname}}
      \\
      \wfconstr{\envext{\env}{\evarb}{\cname}}
           {\fieldsOfVar{\evarb} = \mconcatenate
             {\immfieldbindings{\fieldname}{\rtypec}}
             {\mutfieldbindings{\fieldnameb}{\rtyped}}
           }
      \label{sr:new:21}
      \\
      \wfconstr{
          \envstr{\envext{\envext{\env}{\evarb}{\cname}}{\many{\evarb}}{\many{\rtype}}}
                 {\tdotref{\evarb}{\many{\fieldname}} = \many{\evarb}_{\immindex}}
        }
        {
          \mconcatenate{
            \mconcatenate{\many{\rtype}_{\immindex} \subt \many{\rtypec}}
                         {\many{\rtype}_{\mutindex} \subt \many{\rtyped}}
          }
          {\classinv{\cname}{\evarb}}
        }
        \label{sr:new:22}
    }
    For fresh $\evarb$ and $\many{\evarb}$.
    \\

    We choose a heap typing $\storety'$, such that:
    $$
    \storety' = \updMapping{\storety}{\loc}{\rtypec_0}
    $$
    Hence:
    \la{
      \idxMapping{\storety'}{\loc} = \rtypec_0  
      \label{sr:new:7}
    }

    By applying Rule \rtchkloc using \ref{sr:new:7}:
    $$
    \dtypecheck{\env}{\storety'}{\loc}{\rtypec_0}
    $$
    
    Which proves \ref{sr:d}.
    \\

    By applying Rule \rtchkobj using \ref{sr:new:7}, \ref{sr:new:19} and \ref{sr:new:15}:
    \la{
      \rttypecheck{\ssaRtState}{\storety}{\ssaObject}{\rtypec_0}
      \label{sr:new:24}
    }

    By \subid we trivially get:
    \la{
      \dissubtype{\env}{\storety'}{\rtypec_0}{\rtypec_0} 
      \label{sr:new:4}
    }
    Which proves \ref{sr:e}.
    \\

    By applying Rule~\hbndinst on
    \ref{sr:new:19}, 
    \ref{sr:new:17}, 
    \ref{sr:new:7}, 
    \ref{sr:new:21},
    \ref{sr:new:15} and 
    \ref{sr:new:22}:
    $$
    \newwfstore{\storety'}{\accessState{\ssaRtState'}{\ssaHeap}}
    $$
    Which proves \ref{sr:f}.

  \item \brackets{\rletin} 
    \emph{Similar approach to case \rinvoke}.

  \item \brackets{\rdotasgn}: Fact \ref{sr:b} has the form:
    \la{
      \stepsFour{\ssaRtState }{\dotassign{\loc}{\fieldnameb_{\index}}{\val'}}
                {\ssaRtState'}{\val'} 
      \label{sr:asgn:30}
    }

    By inverting Rule~\rdotasgn on \ref{sr:asgn:30}:
    \la{
      \ssaHeap' = \updMapping{\accessState{\ssaRtState}{\ssaHeap}}{\loc}
                             {\updMapping{\idxMapping{\accessState{\ssaRtState}{\ssaHeap}}{\loc}}
                             {\fieldnameb_{\index}}{\val'}}
      \label{sr:asgn:80}
    }

    From \ref{sr:a} for $\expr \equiv \dotassign{\loc}{\fieldnameb_{\index}}{\val'}$:
    \la{
      \dtypecheck{\env}{\storety}{\dotassign{\loc}{\fieldnameb_{\index}}{\val'}}{\rtype'}     
      \label{sr:asgn:31}
    }

    By inverting Rule~\lqchkasgn on \ref{sr:asgn:31}:
    \la{
      \dtypechecktwo{\env}{\storety}{\loc}{\rtype_{\loc}}{\val'}{\rtype'}        
      \label{sr:asgn:32}
      \\
      \dwfconstr{\envext{\env}{\evarb}{\basetype{\rtype_{\loc}}}}{\storety}
                {\mconcatenate
                  {\varbinding{\hasmut{\evarb}{\fieldnameb_{\index}}}{\rtyped_{\index}}}
                  {\rtype' \subt \rtyped_{\index}}
                }     
      \label{sr:asgn:33}
    }

    For a fresh $\evarb$. 
    \\\\
    By \ref{sr:asgn:32} and \subrefl we prove \ref{sr:d} and \ref{sr:e}.
    \\\\
    By inverting \rtchkloc on \ref{sr:asgn:32}:
    \la{
      \idxMapping{\storety}{\loc} = \rtype_{\loc}      
      \label{sr:asgn:5}
    }

    By inverting \hbndinst on \ref{sr:c} for location $\loc$
    and using \ref{sr:asgn:5}:
    \la{
      \ssaObject \defeq \heapObject{\loc'}{\fieldDefsSym}
      \label{sr:asgn:81}
      \\
      \fieldDefsSym \defeq \mconcatenate
        {\immFieldDefs{\fieldname }{\val}_{\immindex}}
        {\mutFieldDefs{\fieldnameb}{\val}_{\mutindex}}
      \label{sr:asgn:82}
      \\
      \basetype{\idxMapping{\storety}{\loc}} = \cname
      \label{sr:asgn:83}
      \\
      \wfconstr{\envext{\env}{\evarb}{\cname}}
               {\fieldsOfVar{\evarb} = 
                  \mconcatenate{\immfieldbindings{\fieldname}{\rtypec}}
                               {\mutfieldbindings{\fieldnameb}{\rtyped}}
               }
      \label{sr:asgn:84}
      \\
      \rttypecheckval{\storety}{\many{\val}_{\immindex}}{\many{\rtype}_{\immindex}}
      \label{sr:asgn:85}
      \\
      \rttypecheckval{\storety}{\many{\val}_{\mutindex}}{\many{\rtype}_{\mutindex}}
      \label{sr:asgn:86}
      \\
      \wfconstr{
          \envextex
            {\envext{\env}{\evarb}{\cname}}
            {\many{\evarb}_{\immindex}}
          {\singleton{\many{\rtype}_{\immindex}}{\tdotref{\evarb}{\many{\fieldname}}}}
        }
        {
          \mconcatenate{
            \mconcatenate{\many{\rtype}_{\immindex} \subt \many{\rtypec}}
            {\many{\rtype}_{\mutindex} \subt \many{\rtyped}}
          }
          {\classinv{\cname}{\evarb}}
        }
      \label{sr:asgn:87}
    }

    Fact \ref{sr:asgn:80} becomes:
    \la{
      \ssaHeap' & = \updMapping{\accessState{\ssaRtState}{\ssaHeap}}{\loc}{\ssaObject'}
      \label{sr:asgn:88}
      \\
      \ssaObject' & = \heapObject{\loc'}{\fieldDefsSym'}
      \label{sr:asgn:89}
      \\
      \fieldDefsSym' & = \mconcatenate
        {\immFieldDefs{\fieldname }{\val}_{\immindex}}
        {\mutFieldDefs{\fieldnameb}{\val}'_{\mutindex}}
      \label{sr:asgn:90}
      \\
      \many{\val}'_{\mutindex} & = \many{\val}_{\mutindex, .. \index -1}, 
                                       \val'_{\mutindex, \index}, 
                                 \many{\val}_{\mutindex, \index + 1.. } 
      \label{sr:asgn:91}
    }

    Also by \ref{sr:asgn:32} and \ref{sr:asgn:86} it holds that:
    \la{
      \rttypecheckval{\storety}{\many{\val}'_{\mutindex}}
        {\parens{
          {\many{\rtype}_{\mutindex, .. \index -1}, 
                 \rtype', 
           \many{\rtype}_{\mutindex, \index + 1.. }} 
          }}
      \label{sr:asgn:92}
    }

    By Lemma~\ref{lemma:subtype:weaken} on \ref{sr:asgn:33}:
    \la{
      \dwfconstr{
          \envextex
            {\envext{\env}{\evarb}{\cname}}
            {\many{\evarb}_{\immindex}}
          {\singleton{\many{\rtype}_{\immindex}}{\tdotref{\evarb}{\many{\fieldname}}}}
        }{\storety}{\rtype' \subt \rtyped_{\index}}
      \label{sr:asgn:93}
    }

    By applying Rule~\hbndinst on 
    \ref{sr:asgn:89},
    \ref{sr:asgn:90},
    \ref{sr:asgn:83},
    \ref{sr:asgn:84},
    \ref{sr:asgn:85},
    \ref{sr:asgn:92},
    \ref{sr:asgn:87} and
    \ref{sr:asgn:93}:
    $$
    \newwfstore{\storety}{\ssaHeap'}
    $$

    Which proves \ref{sr:f}.

  \item \brackets{\rletif}:
    Assume $\vconst \equiv \ssaTrue$ (the case for $\ssaFalse$ is symmetric).
    \\
    
    Fact \ref{sr:b} has the form:
    \la{
      \stepsFour
        {\ssaRtState}
        {\letif{\many{\evar}}{\many{\evar}_1}{\many{\evar}_2}
               {\ssaTrue}{\ssactx_{1}}{\ssactx_{2}}{\expr}
        }
        {\ssaRtState}
        {\ssactxidx{\ssactx_1}
                   {\appsubst{\esubst{\many{\evar}_1}{\many{\evar}}}{\expr}}
        }
      \label{sr:letif:1}
    }

    By Rule \lqchkssactx fact \ref{sr:a} has the form:
    \la{
      \typecheck{\env}
                {\letif{\many{\evar}}{\many{\evar}_1}{\many{\evar}_2}
                       {\ssaTrue}{\ssactx_{1}}{\ssactx_{2}}{\expr}}
                {\texist{\many{\evar}}{\many{\rtypeb}}{\rtypec}}
      \label{sr:letif:10}
    }

    So type $\rtype$ has the form: 
    \la{
      \rtype \equiv \texist{\many{\evar}}{\many{\rtypeb}}{\rtypec}
    }

    By inverting Rule \lqchkssactx on \ref{sr:a}:
    \la{
      \typecheckctx{\env}{
        \letif{\many{\evar}}{\many{\evar}_1}{\many{\evar}_2}
              {\ssaTrue}{\ssactx_{1}}{\ssactx_{2}}{\hole}
        }{\envbinding{\many{\evar}}{\many{\rtypeb}}}
      \label{sr:letif:2}
      \\
      \typecheck{\envextex{\env}{\many{\evar}}{\many{\rtypeb}}}{\expr}{\rtypec}
      \label{sr:letif:3}
    }

    By inverting Ryle \lqchkctxletif on \ref{sr:letif:2}:
    \la{
      \typechecksubt{\env}{\ssaTrue}{\rtypeb}{\rtypeb}{\tbool}
      \label{sr:letif:4}
      \\
      \typecheckctx{\envgrdext{\envext{\env}{\evarb}{\rtypeb}}{\evarb}}    {\ssactx_1}{\env_1}
      \label{sr:letif:5}
      \\
      \typecheckctx{\envgrdext{\envext{\env}{\evarb}{\rtypeb}}{\neg\evarb}}{\ssactx_2}{\env_2}
      \label{sr:letif:6}
      \\
      \issubtype{\mconcatenate{\env}{\env_1}}{\idxMapping{\env_1}{\many{\evar}_1}}{\many{\rtypeb}}
      \label{sr:letif:7}
      \\
      \issubtype{\mconcatenate{\env}{\env_2}}{\idxMapping{\env_2}{\many{\evar}_2}}{\many{\rtypeb}}
      \label{sr:letif:8}
      \\
      \wftype{\env}{\many{\rtypeb}}
      \label{sr:letif:9}
    }

    By Rule \lqchkconst on $\ssaTrue$:
    \la{
      \typecheck{\env}{\ssaTrue}{\reftp{\tbool}{\vv = \ssaTrue}}
    }

    By Lemma~\ref{lemma:bool:fact} on \ref{sr:letif:4} and \ref{sr:letif:5}:
    \la{
      \typecheckctx{\env}{\ssactx_1}{\env_1}
      \label{sr:letif:15}
    }

    Environment $\env_1$ has the form:
    \la{
      \env_1 \equiv 
        \envbinding{\many{\evar}_1 }{\idxMapping{\env_1}{\many{\evar}_1 }},
        \envbinding{\many{\evar}_1'}{\idxMapping{\env_1}{\many{\evar}_1'}}
      \label{sr:letif:11}
    }
    For some $\many{\evar}_1'$.
    \\\\
    By Lemma~\ref{lemma:subst:env} using \ref{sr:letif:3}:
    \la{
      \typecheck{\envextex{\env}{\many{\evar}_1}{\many{\rtypeb}}}
                {\appsubst{\esubst{\many{\evar}_1}{\many{\evar}}}{\expr}}
                {\appsubst{\tsubst{\many{\evar}_1}{\many{\evar}}}{\rtypec}}
      \label{sr:letif:13}
    }

    By Lemma~\ref{lemma:typing:weaken} using \ref{sr:letif:13}:
    \la{
      \typecheck{\env, 
              \envbinding{\many{\evar}_1 }{\many{\rtypeb}},
              \envbinding{\many{\evar}_1'}{\idxMapping{\env_1}{\many{\evar}_1'}}
            }
            {\appsubst{\esubst{\many{\evar}_1}{\many{\evar}}}{\expr}}
            {\appsubst{\tsubst{\many{\evar}_1}{\many{\evar}}}{\rtypec}}
      \label{sr:letif:14}
    }


    By applying rule \lqchkssactx on \ref{sr:letif:15} and \ref{sr:letif:14}:
    \la{
      \typecheck{\env}
                {\ssactxidx{\ssactx}{\appsubst{\esubst{\many{\evar}_1}{\many{\evar}}}{\expr}}}
                {
                  \texist{\many{\evar}_1 }{\idxMapping{\env_1}{\many{\evar}_1 }}
                 {\texist{\many{\evar}_1'}{\idxMapping{\env_1}{\many{\evar}_1'}}
                  {\appsubst{\tsubst{\many{\evar}_1}{\many{\evar}}}{\rtypec}}
                 }
                }
      \label{sr:letif:16}
    }

    Which proves \ref{sr:d}.
    \\\\
    Fact \ref{sr:letif:16} can be rewritten as:
    \la{
       \typecheck{\env}
                {\ssactxidx{\ssactx}{\appsubst{\esubst{\many{\evar}_1}{\many{\evar}}}{\expr}}}
                {
                  \texist{\many{\evar} }{\idxMapping{\env_1}{\many{\evar} }}
                 {\texist{\many{\evar}_1'}{\idxMapping{\env_1}{\many{\evar}_1'}}
                  {\rtypec}
                 }
                }
      \label{sr:letif:17}
    }
    
    Applying Rule \subbind using \ref{sr:letif:17}:
    \la{
       \issubtype{\env}
                {
                  \texist{\many{\evar} }{\idxMapping{\env_1}{\many{\evar} }}
                 {\texist{\many{\evar}_1'}{\idxMapping{\env_1}{\many{\evar}_1'}}
                  {\rtypec}
                 }
                }
                {
                  \texist{\many{\evar}}{\idxMapping{\env_1}{\many{\evar} }}
                         {\rtypec}
                }
      \label{sr:letif:18}
    }

    By Lemma~\ref{lemma:exist:weaken} on the right-hand side of \ref{sr:letif:18}:
    \la{
       \issubtype
        {\env}
        {\texist{\many{\evar}}{\idxMapping{\env_1}{\many{\evar}}}{\rtypec}}
        {\texist{\many{\evar}}{\rtypeb}{\rtypec}}
      \label{sr:letif:19}
    }

    By 
    \ref{sr:letif:17},
    \ref{sr:letif:18} and 
    \ref{sr:letif:19}, and using Rule \subtrans we prove \ref{sr:e}.
    \\\\
    Heap $\accessState{\ssaRtState}{\ssaHeap}$ does not evolve so \ref{sr:f} holds trivially.

\end{itemize}

\end{proof}

\begin{theorem}[Progress]\label{theorem:progress}
If
\begin{enumerate}[label=(\alph*)]
  \item $\dtypecheck{\env}{\storety}{\expr}{\rtype}$,       \label{p:a}
  \item $\newwfstore{\storety}{\ssaHeap}$                  \label{p:b}
\end{enumerate}
then one of the following holds:
\begin{enumerate}
  \item $\expr$ is a value,
  \item there exist $\expr'$, $\ssaHeap'$ and $\storety' \supseteq \storety$ \st
    $\newwfstore{\storety'}{\ssaHeap'}$ and
    $\stepsFour{\ssaHeap}{\expr}{\ssaHeap'}{\expr'}$.
\end{enumerate}

\end{theorem}

\begin{proof}

We proceed by induction on the structure of derivation~\ref{p:a}:

\begin{itemize}

  \item \brackets{\lqchkfieldimm}
    \la{
      \dtypecheck{\env}{\storety}{\dotref{\expr_0}{\fieldname_{\index}}}
                 {\texist{\evarb}{\rtype_0}{\singleton{\rtype}{\tdotref{\evarb}{\fieldname_{\index}}}}}
                 \label{p:fieldi:0}
    }
    By inverting \lqchkfieldimm on \ref{p:fieldi:0}:
    \la{
      \dtypecheck{\env}{\storety}{\expr_0}{\rtype_0}        \label{p:fieldi:1}
      \\
      \dwfconstr{\envextex{\env}{\evarb}{\rtype_0}}{\storety}
                {\varbinding{\hasimm{\evarb}{\fieldname_{\index}}}{\rtype}}
                                                          \label{p:fieldi:2}
    }

    By i.h. using \ref{p:fieldi:1} and \ref{p:b} there are two possible cases on $\expr_0$:
    \begin{itemize}

      \item \brackets{$\expr_0 \equiv \loc_0$}
        Statement \ref{p:fieldi:1} becomes:
        \la{
          \dtypecheck{\env}{\storety}{\loc_0}{\rtype_0}   
          \label{p:fieldi:11}
        }

        By \ref{p:b} for location $\loc_0$:
        \la{
          \newwfstore{\storety}{\updMapping{\ssaHeap}{\loc_0}{\ssaObject}}
          \label{p:fieldi:12}
        }

        Where:
        \la{
          \ssaObject \equiv \heapObject{\loc_0'}{\fieldDefsSym}
          \label{p:fieldi:13}
        }

        By Lemma~\ref{lemma:store:type} using \ref{p:b} and \ref{p:fieldi:12}:
        \la{
          \idxMapping{\storety}{\loc_0} = \rtype_0          
          \label{p:fieldi:14}
          \\
          \dwfconstr{\env}{\storety}{\mconcatenate{\varbinding{\ssaObject}{\rtypeb_0}}{\rtypeb_0 \subt \rtype_0}}
          \label{p:fieldi:15}
        }

        By Lemma~A.6 in~\cite{Nystrom2008} using \ref{p:fieldi:2} and
        \ref{p:fieldi:15}:
        \la{
          \dwfconstr{\envextex{\env}{\evarb}{\rtypeb_0}}{\storety}
                    {\varbinding{\hasimm{\evarb}{\fieldname_{\index}}}{\rtype}}
                                                          \label{p:fieldi:18}
        }

        By applying Rule~\rfield using \ref{p:fieldi:12}, \ref{p:fieldi:13} and \ref{p:fieldi:18}:
        $$
        \stepsFour{\ssaHeap}{\dotref{\loc_0}{\fieldname_{\index}}}{\ssaHeap}{\val_{\index}}
        $$

      \item \brackets{$\exists \expr_0'$ \st $\stepsFour{\ssaHeap}{\expr_0}{\ssaHeap'}{\expr_0'}$}
        By applying Rule~\rcevalctx:
        $$
        \stepsFour{\ssaHeap}{\dotref{\expr_0}{\fieldname_{\index}}}
              {\ssaHeap'}{\dotref{\expr_0'}{\fieldname_{\index}}}
        $$

    \end{itemize}

  \item \brackets{\lqchkfieldmut}
    \emph{Similar to previous case.}

  \item \brackets{\lqchkinv}, \brackets{\lqchknew}
    \emph{Similar to the respective case of CFJ~\cite{Nystrom2008}.}

  \item \brackets{\lqchkcast}:
    \la{
      \dtypecheck{\env}{\storety}{\cast{\rtype}{\expr_0}}{\rtype}
                                                            \label{p:cast:0}
    }

    By inverting \lqchkcast on \ref{p:cast:0}:
    \la{
      \typecheck{\env}{\expr_0}{\rtypeb_0}                    \label{p:cast:1}
      \\
      \dwftype{\env}{\storety}{\rtype}
      \\
      \dsemsubtype{\env}{\storety}{\rtypeb_0}{\rtype}         \label{p:cast:14}
    }

    By i.h. using \ref{p:cast:1} and \ref{p:b} there are two possible cases on $\expr_0$:
    \begin{itemize}

      \item \brackets{$\expr_0 \equiv \loc_0$}
        Statement \ref{p:cast:1} becomes:
        \la{
          \dtypecheck{\env}{\storety}{\loc_0}{\rtypeb_0}    \label{p:cast:11}
        }
        
        By Lemma~\ref{lemma:cast} using \ref{p:b} and \ref{p:cast:14}:
        \la{
          \dwfconstr{\env}{\storety}
            {\varbinding{\idxMapping{\ssaHeap}{\loc_0}}{\rtypec_0}, \rtypec_0 \subt \rtype}
          \label{p:cast:13}
        }

        From \rulename{\rcast} using \ref{p:cast:13}:
        $$
        \stepsFour{\ssaHeap}{\cast{\rtype}{\loc_0}}{\ssaHeap}{\loc_0}
        $$

      \item \brackets{$\exists \expr_0'$ \st $\stepsFour{\ssaHeap}{\expr_0}{\ssaHeap'}{\expr_0'}$}
        By rule \rcevalctx:
        $$
        \stepsFour{\ssaHeap}{\cast{\rtype}{\expr_0}}{\ssaHeap'}{\cast{\rtype}{\expr'_0}}
        $$

    \end{itemize}

  \item \brackets{\lqchkletin}, \brackets{\lqchkasgn}, \brackets{\lqchkite}
    \emph{These cases are handled in a similar manner.}

\end{itemize}

\end{proof}

\end{document}